\documentclass{article}
\usepackage{blindtext}
\usepackage{geometry}
\usepackage{amsmath}
\usepackage{amsfonts}
\usepackage{amssymb}
\usepackage{url}
\usepackage{amsthm}
\usepackage{easyReview}
\usepackage{scicite}
\setreviewson
\usepackage[utf8]{inputenc}
\usepackage{enumerate}
\usepackage{changepage}
\usepackage[mathlines]{lineno}

\newtheorem{thm}{Theorem}
\newtheorem{prop}{Proposition}
\newtheorem{conj}{Conjecture}
\newtheorem*{theorem*}{Theorem}

\usepackage[utf8]{inputenc}
\usepackage{graphicx}
\usepackage[capitalise]{cleveref}
\usepackage{subcaption}

\theoremstyle{remark}
\newtheorem{remark}{Remark}
\usepackage[T1]{fontenc}
\usepackage[utf8]{inputenc}
\usepackage{tabularx,ragged2e,booktabs,caption}
\newcolumntype{C}[1]{>{\Centering}m{#1}}

\usepackage{cancel}
\usepackage{secdot}
\usepackage[utf8]{inputenc}
\usepackage{tikz}
\usetikzlibrary{matrix,shapes,arrows,positioning,chains}
\usetikzlibrary{shapes.geometric, arrows}
\usetikzlibrary{shapes.geometric, arrows, shadows}

\usepackage{color}
\usepackage{xcolor}
\usepackage{mathtools}          

\newcounter{lastnote}

\title{Using the xr package on Overleaf}

\usepackage{xr}
\makeatletter

\newcommand*{\addFileDependency}[1]{
\typeout{(#1)}
\@addtofilelist{#1}
\IfFileExists{#1}{}{\typeout{No file #1.}}
}\makeatother

\newcommand*{\myexternaldocument}[1]{%
\externaldocument{#1}%
\addFileDependency{#1.tex}%
\addFileDependency{#1.aux}%
}

\myexternaldocument{main_appendix}

\begin{document}

\vspace{10pt}
\setlength{\marginparwidth}{0pt}
\setlength{\marginparsep}{0pt}
\marginparwidth = 10pt

\begin{center}

{\Large \bf Epidemic thresholds and disease dynamics in metapopulations: the role of network geometry and human mobility.} 

\vspace{.25cm}
 
{\large Haridas K. Das}\footnote[1]{Department of Mathematics, Oklahoma State University, Stillwater, OK 74078, e-mail: haridas.das$@$okstate.edu},\footnote[2]{Department of Mathematics, Dhaka University, Dhaka 1000, Bangladesh, e-mail: hkdas$\_$math$@$du.ac.bd}
{\large and Lucas M. Stolerman}\footnote[3]{Department of Mathematics, Oklahoma State University, Stillwater, OK 74078, e-mail:lucas.martins$\_$stolerman$@$okstate.edu}
\vspace{.03cm}
\end{center}

\parskip=8pt

\begin{abstract} 
We calculate epidemic thresholds and investigate the dynamics of a disease in a networked metapopulation model. To study the specific role of mobility levels and network geometry, we utilize the SIR-Network model and consider a range of geometric structures. For \emph{star-shaped} networks where all nodes only connect to a center, we obtain the same epidemic threshold formula as previously found for fully connected networks in the case where all nodes have the same infection rate except one. Next, we analyze \emph{cycle-shaped} networks that yield different epidemic thresholds than star-shaped ones. We then analyze more general classes of networks by combining the star, cycle, and other structures, obtaining classes of networks with the same epidemic threshold formulas. We present some conjectures on even more flexible networks and complete our analysis by presenting simulations to explore the epidemic dynamics for the different geometries.

\end{abstract}


\vspace{.25cm}

\section{Introduction}


Epidemic transmission has become a significant public health concern in the twenty-first century, with the recent COVID-19 pandemic servings as a noteworthy example \cite{who_coronavirus_website}. It is widely acknowledged that globalization and environmental factors have led to numerous infectious disease outbreaks \cite{jones2008global,morse2001factors,Quinn_supriya}. A recent study from Piret and Boivin \cite{Piret_Boivin} provided a comprehensive overview of pandemics throughout history, noting that many of them occurred as a result of increased global activities. Human mobility, in particular, is recognized to play a significant role in shaping epidemic dynamics of respiratory \cite{charu2017human,merler2010role}, vector-borne \cite{adams2009man,wesolowski2015impact} and zoonotic diseases  \cite{kramer2016spatial}. The complex interplay between disease transmission and contact patterns at multiple levels, from day-to-day local interactions to international travels, has been a topic of intense study in the past decades \cite{lloyd2005superspreading,bloom2019infectious,tatem2006global}.

\vspace{0.10cm}

Historically, mathematical models have been used to explore the impacts of human mobility in the dynamics of infectious diseases \cite{viboud2006synchrony,ferguson2005strategies,colizza2006modeling}. The special class of \emph{metapopulation} (or patch) models can include different subpopulation sizes, movement rates, and location-specific parameters \cite{Sattenspiel1995,keeling2005,arino2003,fulford2002metapopulation,calvetti2020metapopulation}. In those models, the disease is assumed to spread in each subpopulation, which can be divided into disease-related states (susceptible, exposed, infected, among others). Sattenspiel and Dietz \cite{Sattenspiel1995} introduced a metapopulation demographic model with two distinct mobility patterns and studied generalizations to describe a measles epidemic on the Caribbean island of Dominica.  Keeling et al. \cite{keeling2005} investigated the role of individual identity and its impact on epidemic dynamics. Other important contributions in the field were made by Arino \& Van Den Driessche \cite{arino2003} and Fulford et al. \cite{fulford2002metapopulation}. More recently, the SIR-network model was introduced in the context of dengue fever epidemics \cite{lucas_dengue_model}. The simple structure of the SIR-Network model equations has been used in a range of generalized systems, including models for  COVID-19 epidemics with hospitalization and lockdown strategies \cite{della2020network}, vaccination behavior \cite{chang2019effects,chang2020impact}, and spatially structured social dynamics \cite{harding2020population}.

Epidemic thresholds for infectious diseases are typically calculated with the basic reproduction number $R_0$, the averaged cases generated by one infected individual in a susceptible population \cite{dietz1993estimation,delamater2019complexity}. In the past, several studies have attempted to estimate $R_0$ for metapopulation models through the next-generation matrix approach  \cite{delamater2019complexity,diekmann2010construction,diekmann2000mathematical}. Fulford, Roberts, and Heesterbeek  provided a methodology to compute $R_0$ for  Tuberculosis in Possums \cite{fulford2002metapopulation}. Arino and Van Den Driessche \cite{arino2003} proposed a general multi-city SIS model and computed bounds for $R_0$ with explicit formulas in particular cases. Recently, more elaborated estimates of $R_0$ were proposed for vector-borne diseases \cite{iggidr2016dynamics,tocto2021effect}.  Despite the efforts to characterize epidemic thresholds for disease-specific models with traditional methods, little attention has been paid to the impact of the metapopulation network geometry on these thresholds. For example, urban environments are intricately connected through heterogeneous movement patterns, which can be described as complex networks of transportation flow. These networks are far from equally connected: they have big hubs, poorly connected nodes, and other nuanced connections that can significantly impact the spread of infectious diseases. Understanding the relationship between network geometry and epidemic dynamics is thus critical for effective disease control and prevention efforts.

In this work, we calculated epidemic thresholds in a metapopulation model to identify how the structure of connections between nodes and the intensity of human movement may control or promote outbreaks. To this end, we built upon previous work and established epidemic thresholds on various network geometries using the SIR-network model \cite{lucas_dengue_model}. The model describes the dynamics of a disease in a city, where people travel between neighborhoods. Each neighborhood is represented by a node of a network and contains a population divided into susceptible, infected, and recovered classes.  To extend previous results on fully connected networks, we assume that all nodes have the same infection rate except one, which we call a \emph{heterogeneous} node.  Our main finding is the discovery of a class of networks with the same epidemic thresholds as the fully connected network, in which the movement of people is critical for epidemic control. Such class is led by \emph{star-shaped} networks with the heterogeneous node in the center. In contrast, \emph{cycle} networks exhibit different epidemic thresholds and thus are not part of the class.  We also performed numerical simulations to complement our theoretical analysis and  found significant differences between the disease dynamics depending on the network structure. In particular, we observed how poorly connected networks promoted delayed and damped epidemic peaks on nodes that are distant from the heterogeneous nodes, while fully connected and star-shapes exhibited more synchronized outbreaks.

The paper is organized as follows. First, we briefly review the SIR-network model assumptions (Section 2) and the previously established epidemic thresholds for fully connected networks \cite{lucas_dengue_model} (Section 3). In section 4, we introduce the new theoretical estimates for \emph{star-shaped} networks, \emph{cycle} networks, and the extended combinations of star-shaped and cycle  (\emph{cycle-support} and \emph{star-cycle} networks). We further establish more thresholds for other networks such as \emph{star-triangle} and \emph{star-background} networks. We complete our analysis with two conjectures about thresholds for the most flexible network structures.  In Section 5, we perform numerical simulations to complement our theoretical results and investigate the epidemic dynamics depending on the number of nodes (network size) and geometry. Finally, in Section 6 we discuss our theoretical and numerical findings and their practical importance.

\section{The SIR-network model} \label{The SIR-network model}

In the SIR-network model \cite{lucas_dengue_model}, it is assumed that a city is divided into $M$ neighborhoods, which are the nodes of a network represented by the set $V= \{1, 2, \hdots, M \}$. The weighted edges of the network are the fractions $\phi_{ij} \in [0,1]$  of residents moving daily between nodes $i$ and $j$.  The geometry of the network model can thus be set by the \textit{flux matrix} $ \phi_{M \times M} = [\phi_{ij}]$, where $(i,j) \in V \times V$, or explicitly 
\begin{equation}\label{flux.matrix}
    \phi_{M \times M}= \begin{pmatrix}
   \phi_{11}  & \phi_{12} & \phi_{13} & \hdots & \phi_{1M}\\
   \phi_{21}  & \phi_{22} & \phi_{23} &  \hdots & \phi_{2n}\\
   \vdots  & \vdots & \vdots & \ddots & \vdots\\
      \phi_{M1}  & \phi_{M2} & \phi_{M3} & \hdots & \phi_{MM}
           \end{pmatrix}.
    \end{equation}
The flux matrix is stochastic by assumption since the population residents at each node $i \in V$ is conserved; hence the fractions $\phi_{ij}$ (also called fluxes) satisfy the following criterion:

\begin{equation*}
    \sum_{j=1}^{M} \phi_{ij}=1. 
\end{equation*}

The populations of susceptible, infected, and removed individuals at node $i \in V$ and time $t$ are given by $S_i(t), I_i(t),$ and $R_i(t)$ respectively. Each node $i$ has a total population $N_i$ and its transmission rate $\beta_i$, which can be related to local transmission probabilities and/or contact patterns driven by environmental factors, and human behavior, among others. The model equations are given by

\begin{equation} \label{SIRnetwork}
\begin{aligned}
    & \Dot{S_i}(t)=- \sum_{j=1}^{M} \sum_{k=1}^{M} \beta_j \phi_{ij}  S_{i}  \frac{\phi_{kj}I_k}{N_{j}^{p}}\\  
  & \Dot{I_i}(t)= \sum_{j=1}^{M} \sum_{k=1}^{M} \beta_j \phi_{ij}  S_{i}  \frac{\phi_{kj}I_k}{N_{j}^{p}}-\gamma I_{i}\\
&    \Dot{R_i}(t)= \gamma I_{i},
\end{aligned}
\end{equation}
where $N_j^p= \sum_{k=1}^M\phi_{kj}N_k$ is the \emph{present} population at node $j$ and $\gamma>0$ is the recovery rate for the disease.  The basic reproduction number $R_0$ is the largest eigenvalue of the next generation matrix (NGM) \cite{diekmann2010construction, diekmann2000mathematical}, which in this case can be calculated as 

\begin{equation}\label{NGMFormulya}
    \kappa= \frac{1}{\gamma} \begin{pmatrix}
   \sum_{j=1}^M \beta_j\phi_{1j}\phi_{1j} \frac{N_1}{N_j^p}   &    \sum_{j=1}^M \beta_j\phi_{1j}\phi_{2j} \frac{N_1}{N_j^p} &  ...      &  \sum_{j=1}^M \beta_j\phi_{1j}\phi_{Mj} \frac{N_1}{N_j^p} \\
   \sum_{j=1}^M \beta_j\phi_{2j}\phi_{1j} \frac{N_2}{N_j^p}   & \sum_{j=1}^M \beta_j\phi_{2j}\phi_{2j} \frac{N_2}{N_j^p} &  ...      & \sum_{j=1}^M \beta_j\phi_{2j}\phi_{Mj} \frac{N_2}{N_j^p} \\
    \vdots & \vdots & \ddots & \vdots\\ 
    \sum_{j=1}^M \beta_j\phi_{Mj}\phi_{1j} \frac{N_M}{N_j^p}   &      \sum_{j=1}^M \beta_j\phi_{Mj}\phi_{2j} \frac{N_M}{N_j^p} &  ...      &  \sum_{j=1}^M \beta_j\phi_{Mj}\phi_{Mj} \frac{N_M}{N_j^p} \\
        \end{pmatrix}.
    \end{equation}

In the special case where the flux matrix is symmetric ($\phi_{ij}=\phi_{ji}$ for all $i$ and $j$ in $V$) and the populations at each node are equal ($N_i=N>0$ for all $i \in V$), we obtain the simplified present population

$$ N_j^p= \sum_{k=1}^M\phi_{kj}N_k=\left(\sum_{k=1}^M\phi_{kj}\right)N=N$$

 and the NGM becomes 
\begin{equation} \label{NGMFormulyaReduction}
    \kappa=    \frac{1}{\gamma} \begin{pmatrix}
   \sum_{j=1}^M \beta_j\phi_{1j}\phi_{1j}   &    \sum_{j=1}^M \beta_j\phi_{1j}\phi_{2j} &  ...      &  \sum_{j=1}^M \beta_j\phi_{1j}\phi_{Mj}  \\
   \sum_{j=1}^M \beta_j\phi_{2j}\phi_{1j}   & \sum_{j=1}^M \beta_j\phi_{2j}\phi_{2j} &  ...      & \sum_{j=1}^M \beta_j\phi_{2j}\phi_{Mj} \\
    \vdots & \vdots & \ddots & \vdots\\ 
    \sum_{j=1}^M \beta_j\phi_{Mj}\phi_{1j} &      \sum_{j=1}^M \beta_j\phi_{Mj}\phi_{2j} &  ...      &  \sum_{j=1}^M \beta_j\phi_{Mj}\phi_{Mj}  \\
        \end{pmatrix}.
    \end{equation}

\begin{remark}
Finding $R_0$ from the NGM approach is equivalent to performing a local stability analysis of the disease-free equilibrium (DFE) $(S^*_i, I^*_i, R^*_i) = (N_i,0,0)$ for $i \in V$. Substituting the perturbation form $S_{i}(t)=S_i^*+\Delta S_{i}(t)$ and $I_{i}(t)=I_i^*+\Delta I_{i}(t)$, where $i=1, \hdots, M$ into (Eq. \ref{SIRnetwork}) and linearizing, we are able to obtain the following relation between the Jacobian matrix $J$, and the next generation matrix $\kappa$  
\begin{equation*}
    J=\gamma (\kappa -I). 
\end{equation*}

To analyze the Jacobian matrix $J$ stability at the DFE, we find locally asymptotic stable (LAS) DFE conditions when all eigenvalues have negative real part or unstable DFE if at least one eigenvalue has a positive real part \cite{Strogatz_book,perko2013differential}. 
\end{remark}

\section{Epidemic thresholds: previous results} \label{previous_results}
This section briefly presents two theorems from Stolerman et al. \cite{lucas_dengue_model}. The first theorem established the outbreak criteria for \emph{homogeneous networks}, where all infection rates are equal in each node. 

\begin{theorem*} [Homogeneous networks in \cite{lucas_dengue_model}] \label{theorem1}
Let $\left(V, \phi_{M \times M} \right)$ be a network associated with the SIR-network system (Eq. \ref{SIRnetwork}) in the case where all nodes have the same transmission rate, $\beta_j=\beta_0  ~~\forall j \in V$. Then basic reproduction number of the homogeneous network models is  $R_0^{\text{hom}}=\frac{\beta_0}{\gamma}$,  and an epidemic outbreak occurs if and only if $\beta_0 > \gamma$. 
\end{theorem*}

The second theorem is devoted to fully connected networks in the case where a single node has a different infection rate (heterogeneous node).  The particular flux matrix $\phi_{M \times M}$ in (Eq. \ref{flux.matrix}) is given by 

\begin{equation} \label{LucasfullConnectedFluxEntries}
 \phi_{ij}= \begin{cases}
                                   \phi_0 & \text{if $i \neq j$} \\
                                 1-(M-1)\phi_0  &  \text{if $i=j$}, 
   \end{cases}
\end{equation}
where $0< \phi_0 \leq \frac{1}{M-1}$, since $\phi_{ij}>0$.

\begin{theorem*}[Fully connected networks with one heterogeneous node \cite{lucas_dengue_model}]  \label{theoremfullcon}
Let $\left(V, \phi_{M \times M} \right)$ be a fully connected network, where $\phi_{ij}= \phi_0$  if $i \neq j$ and $\phi_{ii} = 1-(M-1)\phi_0$. We suppose that $\beta_i= \zeta \beta_0$ for $i \in V$ and $\beta_k= \beta_0$, for all  $k \in V$, where $k \neq i$. Under these conditions, if the background nodes are stable, i.e., $R_0^{\text{hom}}<1$, there is a minimum value of $\zeta$ given by 
\begin{equation*} 
    \zeta^{\text{crit}}= \left( \frac{M-(M-1)R_0^{\text{hom}}}{R_0^{\text{hom}}} \right)= 1+\left( \frac{1-R_0^{\text{hom}}}{R_0^{\text{hom}}} \right)M.
    \label{zeta_critic}
\end{equation*}
In the case $ \frac{1}{R_0^{\text{hom}}}<\zeta<\zeta^{\text{crit}}$, there is an interval  
\begin{equation} \label{interval}
    \mathcal{I}= \left( \max \left\{0, \frac{1}{M}-\Tilde{\phi} \right\}, \min \left\{\frac{1}{M-1}, \frac{1}{M}+\Tilde{\phi} \right\} \right),
\end{equation}

where 
\begin{equation}
    \begin{aligned}
       & \Tilde{\phi}= \frac{1}{M} \sqrt{ \frac{ M(1-R_0^{\text{hom}})-R_0^{\text{hom}}(\zeta-1) }{R_0^{\text{hom}} \left(M \zeta(1-R_0^{\text{hom}})-(\zeta-1)\right)}},
    \end{aligned}
    \label{interv_critic}
\end{equation}
such that if $\phi_0 \in \mathcal{I}$, then the DFE is stable, and otherwise it is unstable. Finally, when $\zeta <\frac{1}{R_0^{\text{hom}}}$, the DFE is always stable.   
\end{theorem*}

When the infection rate $\beta_0$ is lower than the recovery rate $\gamma$, theorem \ref{theorem1} assures that the outbreak is controlled, which is the same outbreak criteria of the SIR model. In theorem \ref{theoremfullcon}, Stolerman et al. \cite{lucas_dengue_model} \emph{perturbed} the homogeneous networks by adding a single node with a different infection rate. In this case, the flux of people between nodes plays a vital role in the epidemic criteria; there exists a critical value of the ratio $\zeta = \frac{\beta_M}{\beta_0}$ such that there exists an interval of flux values that control the epidemics. In section 4, we determine similar epidemic thresholds for the SIR-network model in the different network geometries.

\section{New epidemic thresholds} \label{NewTheory}
This section aims to introduce new classes of networks where the flux may still play an important role in controlling epidemic outbreaks. Previously, Stolerman et al. \cite{lucas_dengue_model} considered the flux matrix (Eq. \ref{LucasfullConnectedFluxEntries}) to describe the \emph{fully connected} networks where the fractions of movement between any two nodes are given by a constant $\phi_0$. Therefore, the following question naturally arises: 
\begin{itemize}
     \item[]  \emph{What happens when the network is not fully connected?}
       \end{itemize}

To answer that question and extend the results of Theorem 2 in \cite{lucas_dengue_model}, we explore other classes of networks, still assuming that a single node has a different infection rate. More precisely, given a network with $M$ nodes, we follow the approach of \cite{lucas_dengue_model} and define the positive quantities $\beta_0$ and $\zeta$ such that 

\begin{equation} \label{BetaConditionstarshape}
     \beta_{i} =
  \begin{cases}
                                   \beta_0 & \text{if $i=1,2, ..., (M-1)$} \\
                                   \zeta \beta_0 & \text{if $i=M$} 
   \end{cases}
\end{equation}
Here $\zeta$ is a dimensionless multiplicative factor that modifies the infection rate at the heterogeneous node $M$. In this study, we are particularly interested in the case where all nodes are ``stable'' except one, i.e, $\beta_0 < \gamma$ and $ \zeta \beta_0 >  \gamma$, so we may explore the case where the single node $M$ can be unstable and thus promote an epidemic outbreak in the whole network.

Fig  \ref{fig:diff_networks} illustrates some of the networks we consider in this study. For \emph{star-shaped} networks (subsection 4.1), we consider that the heterogeneous node $M$ (the one with a different infection rate) connects to all other nodes by the same flux $\phi_0$, and all the outer nodes are not connected. These networks mimic a city where a central node receives a daily flux of people. On the other hand, \emph{cycle-shaped} networks (subsection 4.2), which can be unidirectional or bidirectional, do not have such a central node. In this case, all nodes are only connected to their immediate two neighbors, potentially representing small towns where people usually commute short distances. We also consider extensions of the cycle networks defining \emph{cycle-support} and \emph{star-cycle} networks (subsection 4.3). Cycle support networks are built by connecting the heterogeneous node $M$ in the cycle to the rest of the nodes of the networks. On the other hand, for \emph{star-cycle} networks we connect the immediate neighbors of the outer nodes in the star-shaped network.

Throughout this paper, we consider the SIR-network model with nonnegative and symmetric flux matrix $\phi_{M \times M}$ (Eq. \ref{flux.matrix}) and consider that all nodes have the same population size, i.e,  $N_i=N>0, \text{for} \ i= ~1,2,\hdots, M$. These assumptions allow us to solely investigate the impact of the network geometry on the epidemic thresholds and establish a comparison with previous results obtained in \cite{lucas_dengue_model}.

\begin{figure}[!htbp]
   \centering
     \begin{adjustwidth}{0.00in}{0in} 
   \includegraphics[width=1.0\textwidth]{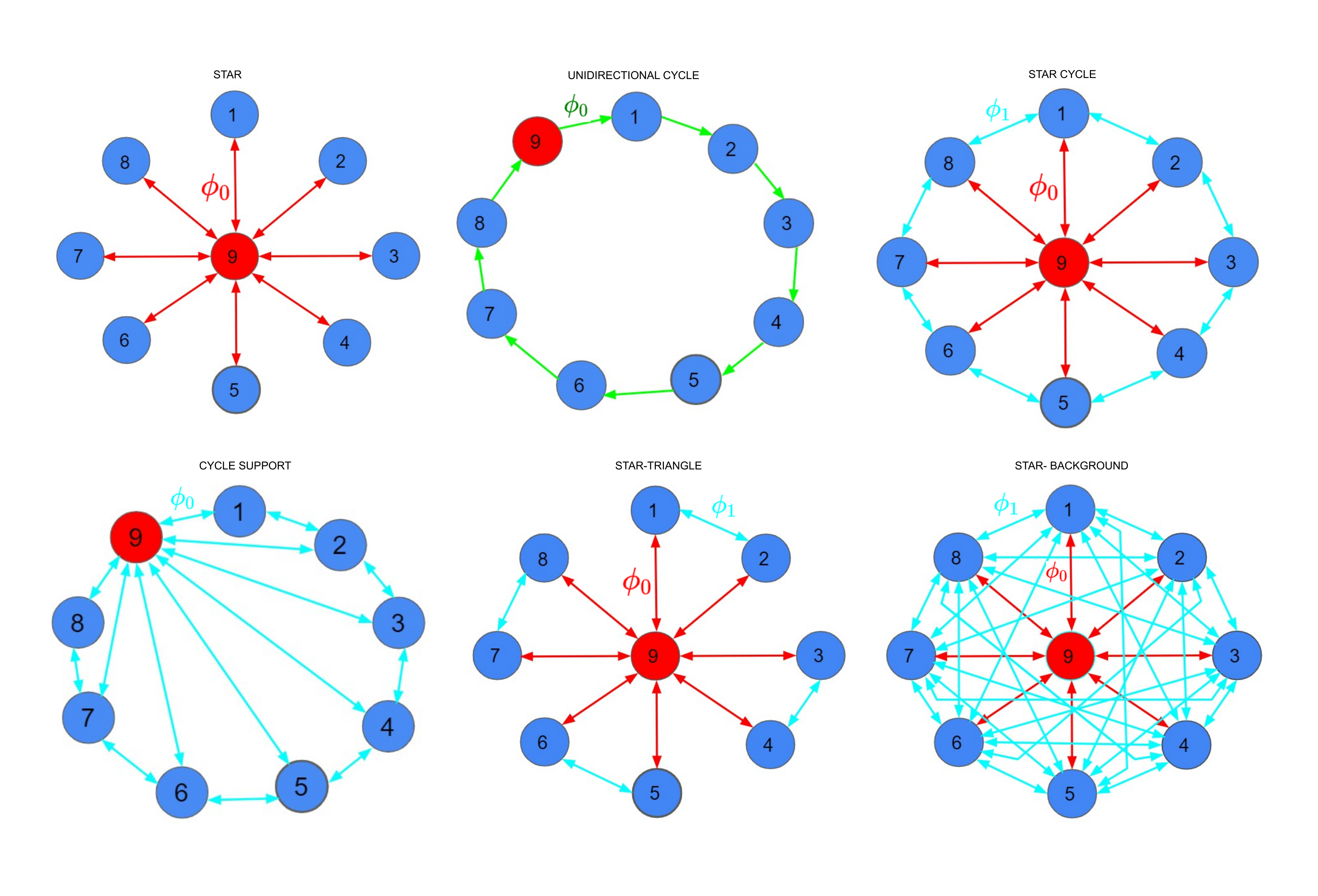}
   \end{adjustwidth}
 \caption{\textbf{Different network geometries ($M=9$):} All blue nodes  have the same infection rates $\beta_0>0$ except the red node, where $\beta_9= \zeta \beta_0$. Here $\zeta>0$ is a dimensionless multiplicative factor.} 
  \label{fig:diff_networks} 
\end{figure}

\subsection{Star-shaped networks} \label{StarShapeSection}
We begin our analysis on star-shaped networks, where all nodes $i \in V$ are uniquely connected to a central node $M$. A star-shaped network mimics a city where a well-connected node (for instance, the downtown of a big city) has a different infection rate and thus may influence the disease spread in the network.
The flux matrix  $\phi_{M \times M}$ of a star-shaped network is defined by  
\begin{equation} \label{starshape}
 \phi_{ij}= \begin{cases}
                                  \phi_0 & \text{if $i=M,~\text{and}~  j \in V \setminus M $} ~ \text{ or ~if $j=M, ~ \text{and} ~ i \in V \setminus M $} \\
                                 1-\phi_0  &  \text{if $i=j,~ \text{and} ~ ~ i,j \in V \setminus M$}\\
                                 1-(M-1)\phi_0  &  \text{if $i=j=M$}\\
                                 0  & \text{otherwise},
   \end{cases}
\end{equation}
where $0< \phi_0 \leq \frac{1}{M-1}$ to ensure that all matrix elements are nonnegative. We may also write the fluxes for Eq. \ref{starshape} in the matrix format 

\begin{equation*} 
    \phi_{M \times M}= \begin{pmatrix}
   1-\phi_0  & 0 & 0 &  \hdots & \phi_0\\
     0  & 1-\phi_0 & 0 & \hdots &\phi_0\\
     0  & 0 & 1-\phi_0 & \hdots &\phi_0\\
         \vdots & \vdots &  \vdots &  \ddots & \vdots\\ 
             \phi_0  & \phi_0 & \phi_0 &  \hdots & 1-(M-1)\phi_0
           \end{pmatrix}.
    \end{equation*}

For \emph{star-shaped} networks, we find the same minimum  $\zeta$ (Eq. \ref{zeta_critic}), and interval $\mathcal{I}$ (Eq. \ref{interval}) as the ones in the fully connected case (Theorem 2 in \cite{lucas_dengue_model}), where the flux between regions can control an epidemic outbreak. We start with a proposition for $M=3$ to show a step-by-step calculation using local stability analysis at the DFE.  
    
\begin{prop} \label{starprop_m3}
Let $V=\left\{1,2,3 \right\}$ be the nodes of the star-shaped network with flux matrix $\phi_{3 \times 3}$ be defined by (Eq. \ref{starshape}) for $M=3$, where $\phi_0 \in (0, \frac{1}{2}]$. For the infection rates $\beta_i$ (Eq. \ref{BetaConditionstarshape}), we assume that all nodes are stable except the center, i.e, $R_0^{\text{hom}}=\frac{\beta_0}{\gamma}<1$ and $\frac{\beta_3}{\gamma} = \frac{\zeta \beta_0}{\gamma}>1$. Then there exists a critical value $\zeta^{\text{crit}}$ given by 
\begin{equation} \label{zeta_condition_M=3}
    \zeta^{\text{crit}}= 1+\left( \frac{1-R_0^{\text{hom}}}{R_0^{\text{hom}}} \right)3.
\end{equation}
such that if  $\zeta< \zeta^{\text{crit}}$, there is an interval 
\begin{equation} \label{interval_star_m_3}
    \mathcal{I}= \left( \max \left\{0, \frac{1}{3}-\Tilde{\phi} \right\}, \min \left\{\frac{1}{2}, \frac{1}{3}+\Tilde{\phi} \right\} \right),
\end{equation}
where 
\begin{equation}\label{star_phi_tilda_m_3}
    \begin{aligned}
        \Tilde{\phi}= \frac{1}{3} \sqrt{ \frac{\gamma (3\gamma-2\beta_0-\zeta\beta_0 )}{\beta_0(\gamma+2\zeta\gamma-3\zeta\beta_0)}}= \frac{1}{3} \sqrt{ \frac{ 3(1-R_0^{\text{hom}})-R_0^{\text{hom}}(\zeta-1) )}{R_0^{\text{hom}} \left(3 \zeta(1-R_0^{\text{hom}})-(\zeta-1) \right)}},
    \end{aligned}
\end{equation}
such that if $\phi_0 \in \mathcal{I}$, then the DFE is stable, and otherwise it is unstable. 
\end{prop}

\begin{proof}
To establish the stability criteria for the DFE, we may perform a local stability analysis of the DFE. The Jacobian matrix at the DFE corresponding to the dynamical system of the $3 \times 3$ star-shaped flux matrix is given by the following expression:

\begin{equation} \label{jacob2}
    \begin{aligned}
     J_{3 \times 3}= \begin{pmatrix}
   \beta_1(-1+\phi_0)^2+ \beta_3 \phi_0^2-\gamma   & \beta_3 \phi_0^2 & \phi_0 \left(\beta_1+\beta_3-\beta_1 \phi_0-2\beta_3 \phi_0 \right)\\
   \beta_3 \phi_0^2 &   \beta_2(-1+\phi_0)^2+ \beta_3 \phi_0^2-\gamma &   \phi_0 \left(\beta_2+\beta_3-\beta_2 \phi_0-2\beta_3 \phi_0 \right)\\
\phi_0 \left(\beta_1+\beta_3-\beta_1 \phi_0-2\beta_3 \phi_0 \right)  &  \phi_0 \left(\beta_2+\beta_3-\beta_2 \phi_0-2\beta_3 \phi_0 \right) &   \beta_3(1-2\phi_0)^2+ (\beta_1+\beta_2) \phi_0^2-\gamma 
           \end{pmatrix}.
    \end{aligned}
    \end{equation}
    
Since symmetric matrices have only real eigenvalues, the eigenvalues of the Jacobian matrix (Eq. \ref{jacob2}) are real. We calculate the eigenvalues using the symbolic packages of the software Mathematica, which are given by 

\begin{equation*}
    \begin{aligned}
       & \lambda_1= \beta_0(-1+\phi_0)^2-\gamma \\
          &   \lambda_2=\frac{1}{2} \beta_0 \left(1+\zeta-2 \phi_0 -4 \zeta \phi_0+3\phi_0^2 + 6 \zeta \phi_0^2 + B-\gamma \right)\\
          &   \lambda_3=\frac{1}{2} \beta_0 \left(1+\zeta-2 \phi_0 -4 \zeta \phi_0+3\phi_0^2 + 6 \zeta \phi_0^2 - B-\gamma \right)
    \end{aligned}
\end{equation*}
   where $ B:=\sqrt{(1-2\phi_0+3\phi_0^2)^2 +\zeta^2 \left(1-4\phi_0+6\phi_0^2 \right)^2+2\zeta \left(-1+6\phi_0-\phi_0^2-24\phi_0^3+ 18 \phi_0^4 \right)}.$ Since $\lambda_2$ and $\lambda_3$ are real, the square root inside $B$ must be nonnegative for $\phi_0 \in (0, \frac{1}{2}]$ and for any $\zeta$, $B$ is well defined and positive real number. 
   
   Next, we find the conditions for which all eigenvalues are negative (LAS DFE), or at least one eigenvalue is positive (unstable DFE).  First, since $\phi_0 \leq 1/2 <1 $, we have $\lambda_1= \beta_0(-1+\phi_0)^2-\gamma<\beta_0-\gamma<0$, since by hypothesis $\frac{\beta_0}{\gamma}<1$. It is thus sufficient to check the critical points of $\lambda_2$ and $\lambda_3$. However,  since $\lambda_2>\lambda_3$, it is thus sufficient to check the sign of $\lambda_2$. It can be  verified that $\phi_0 = \frac{1}{3}$ is a critical point of both eigenvalues $\lambda_2$, and $\lambda_3$, since  $\frac{\partial \lambda_2}{\partial \phi_0} (\phi_0=\frac{1}{3},\beta_0, \zeta, \gamma)=0$, and $\frac{\partial \lambda_3}{\partial \phi_0} (\phi_0=\frac{1}{3},\beta_0, \zeta, \gamma)=0$. Moreover, we can compute the second derivatives at the critical point and obtain
\begin{equation*}
    \begin{aligned}
       &  \frac{\partial^2 \lambda_2}{\partial \phi_0^2} (\phi_0=\frac{1}{3},\beta_0, \zeta, \gamma) =  \frac{12 \beta_0 (\zeta-1)^2}{2+\zeta}>0\\
& \frac{\partial^2 \lambda_3}{\partial \phi_0^2} (\phi_0=\frac{1}{3},\beta_0, \zeta, \gamma) = \frac{54 \zeta \beta_0}{2+\zeta}>0\\
    \end{aligned}
\end{equation*}
and so $\lambda_2$, and $\lambda_3$ have reached to the minimum value when $\phi_0=\frac{1}{3}$. Finally, 
\begin{equation*}
\begin{aligned}
            & \lambda_2(\phi_0=\frac{1}{3},\beta_0, \zeta, \gamma)=\frac{1}{6} \left(2\beta_0+\zeta\beta_0+\sqrt{(2+\zeta)^2} \beta_0 -6\gamma \right)=\frac{1}{3} \left(\beta_0(2+\zeta)-3\gamma \right)\\
            & \lambda_3(\phi_0=\frac{1}{3},\beta_0, \zeta, \gamma)=\frac{1}{6} \left(2\beta_0+\zeta\beta_0-\sqrt{(2+\zeta)^2} \beta_0 -6\gamma \right)=-\gamma <0. 
    \end{aligned}
\end{equation*}
In this way, if $\lambda_2(\phi_0=\frac{1}{3},\beta_0, \zeta, \gamma)=\frac{1}{3} \left(\beta_0(2+\zeta)-3\gamma \right)<0$, then there exists an interval of $\phi_0$ values containing $\phi_0=\frac{1}{3}$ such that all eigenvalues are negative and thus the DFE is stable.  The above condition yields the critical value $\zeta^{crit}$, i.e, 
\begin{equation*}
   \begin{aligned}
          & \left(\beta_0(2+\zeta)-3\gamma \right)<0 \Leftrightarrow \zeta < \frac{3 \gamma- 2 \beta_0}{\beta_0} = 1+\left( \frac{1-R_0^{\text{hom}}}{R_0^{\text{hom}}} \right)3 :=\zeta^{\text{crit}}.\\
   \end{aligned}
\end{equation*}
In addition, the function $\phi_0 \mapsto \lambda_2 (\phi_0, \zeta, \beta_0, \gamma)$ has two roots 
\begin{equation*}
    \begin{aligned}
        & \phi_{c1}= \frac{1}{3} +\frac{1}{3} \sqrt{ \frac{\gamma (3\gamma-2\beta_0-\zeta\beta_0 )}{\beta_0(\gamma+2\zeta\gamma-3\zeta\beta_0)}}=\frac{1}{3}+ \frac{1}{3} \sqrt{ \frac{ 3(1-R_0^{\text{hom}})-R_0^{\text{hom}}(\zeta-1) )}{R_0^{\text{hom}} \left(3 \zeta(1-R_0^{\text{hom}})-(\zeta-1) \right)}},\\
        & \phi_{c2}= \frac{1}{3} -\frac{1}{3} \sqrt{ \frac{\gamma (3\gamma-2\beta_0-\zeta\beta_0 )}{\beta_0(\gamma+2\zeta\gamma-3\zeta\beta_0)}}=\frac{1}{3}- \frac{1}{3} \sqrt{ \frac{ 3(1-R_0^{\text{hom}})-R_0^{\text{hom}}(\zeta-1) )}{R_0^{\text{hom}} \left(3 \zeta(1-R_0^{\text{hom}})-(\zeta-1) \right)}}.
    \end{aligned}
\end{equation*}
 Therefore, taking into account the restriction $\phi_0 \in \left(0,\frac{1}{2}\right]$, the interval where there will be no epidemic is given by $\mathcal{I}(\zeta,R_0^{\text{hom}})=\left( \max \left\{0, \frac{1}{3}-\Tilde{\phi} \right\}, \min \left\{\frac{1}{2}, \frac{1}{3}+\Tilde{\phi} \right\} \right)$, where 

\begin{equation*}
    \begin{aligned}
        \Tilde{\phi}= \frac{1}{3} \sqrt{ \frac{ 3(1-R_0^{\text{hom}})-R_0^{\text{hom}}(\zeta-1)}{R_0^{\text{hom}} \left(3 \zeta(1-R_0^{\text{hom}})-(\zeta-1) \right)}}. 
    \end{aligned}
\end{equation*}
It is worth noting that $\frac{1}{R_0^{\text{hom}}}< \zeta< \zeta^{\text{crit}}$  guarantees that the term inside of the above square root is nonnegative. Therefore, $\Tilde{\phi}$ is a real number, and the interval $\mathcal{I}$ is well defined. This last observation concludes the proof that  the DFE is stable if  $\phi_0 \in \mathcal{I}$, and otherwise unstable.
\end{proof}

The existence of a critical value $\zeta^{crit}$ suggests that the infection rate at the central node plays a crucial role in determining the spread of the epidemic. Specifically, when the value of $\zeta$ is sufficiently low, we observe that the disease-free equilibrium (DFE) remains stable over a range of flux values. On the other hand, if $\zeta > \zeta^{crit}$, then the DFE is unstable for all flux values, a scenario where the central node promotes an outbreak given the sufficiently high infection rate. Finally, if $\zeta< \frac{1}{R_0^{\text{hom}}}$, an algebraic manipulation can be performed to show that $\Tilde{\phi}>\frac{1}{3}$ and thus $\mathcal{I}(\zeta,R_0^{\text{hom}}) = (0,\frac{1}{2})$, which means that the epidemic is controlled for all flux values. This is consistent with our intuition since in this case all nodes would be stable.

We also note that the threshold condition for the local stability of the DFE can be written in terms of a critical infection rate $\beta_0^{\text{crit}}$, which is given by 

\begin{equation} \label{beta0_condition_M=3}
    \beta_0^{\text{crit}}= \gamma\left( \frac{3}{2+\zeta} \right).
\end{equation}

In section 5, we perform numerical simulations changing the $\beta_0$ and observing the emergence of the $\phi_0$ intervals where the epidemic is controlled.

Now we state and prove a general theorem for star-shaped networks of arbitrary size. The particular case for $M=3$ is pedagogical since it highlights the mathematical nature of such \emph{flux-driven} epidemic control. The main tool is the classic local stability analysis, which  is ultimately equivalent to the next-generation matrix approach \cite{driessche2002}. In fact, the central argument of the proof of the next theorem is the particular decomposition of the next-generation matrix in the same format as found by Stolerman et al. \cite{lucas_dengue_model}.

\begin{thm} \label{theoremstar}
Let $V=\left\{1,2,3, \hdots, M \right\}$ be the nodes of a \emph{star-shaped} network. We denote the infection rate $\beta_i = \beta_0$ at node $i$ for all $i = 1,2,3, \cdots, M-1$ and $\beta_M = \zeta \beta_0$. We also assume that all the 
nodes are stable except the center $M$, i.e, $R_0^{\text{hom}}=\frac{\beta_0}{\gamma}<1$ and $\frac{\beta_M}{\gamma} = \frac{\zeta \beta_0}{\gamma}>1$. Then we have a minimum  $\zeta^{\text{crit}}$ given by  


\begin{equation} 
    \zeta^{\text{crit}}= 1+\left( \frac{1-R_0^{\text{hom}}}{R_0^{\text{hom}}} \right)M.
    \label{zeta_critic_star}
\end{equation}

such that if $\frac{1}{R_0^{\text{hom}}}<\zeta< \zeta^{\text{crit}}$, there exists an interval $\mathcal{I}= \left( \max \left\{0, \frac{1}{M}-\Tilde{\phi} \right\}, \min \left\{\frac{1}{M-1}, \frac{1}{M}+\Tilde{\phi} \right\} \right)$, where $\Tilde{\phi}$ is given by (Eq. \ref{interv_critic}),
such that if $\phi_0 \in \mathcal{I}$, then the DFE is stable, and otherwise it is unstable. Finally, if  $\zeta<\frac{1}{R_0^{\text{hom}}}$ i.e. $\beta_M < \gamma$, then the DFE is always stable.      
\end{thm}

\begin{proof}
We first find the next generation matrix (NGM) $\kappa= ( \kappa_{ij})_{M\times M}$, where $ \kappa_{ij}=\frac{1}{\gamma} \sum_{j=1}^M \beta_j\phi_{ij}\phi_{kj} \frac{N_i}{N_j^p}$ (Eq. \ref{NGMFormulya}). Similar to Section \ref{The SIR-network model}, we start from the premise that the populations of each node are equal. i.e. $N_i=N>0$ for all $i \in V$. We also consider a different infection rate in the central node (Eq. \ref{BetaConditionstarshape}). 

Now since we are assuming equal populations and symmetric flux matrix, the NGM can be written as in Eq. \ref{NGMFormulyaReduction}). For a general star-shaped network, the NGM entries can be explicitly calculated. 

In fact, for any $u \in {1,2,3,4,5,...,M-1}$, and $v \in {1,2,3,4,5,...,M-1}$, and $u=v$, then we have the following diagonal elements of the $\kappa$ are
\begin{equation*}
 \begin{aligned}
    &\kappa_{uu}=   \sum_{j=1}^M \beta_j\phi_{uj}\phi_{uj}= \sum_{j=1}^M \beta_j\phi_{uj}^2\\
        &= \beta_1 \phi_{u1}^2+\beta_2\phi_{u2}^2+ \beta_3 \phi_{u3}^2+ \hdots+\beta_u\phi_{uu}^2 +\hdots +\beta_M\phi_{uM}^2  \\
       & = \beta_0 .0^2+ \beta_0. 0^2+\beta_0 . 0^2+ \hdots+\beta_0(1-\phi_0)^2 +\hdots+ \zeta \beta_0 \phi_0^2 \\ 
        & =\beta_0 \left( \zeta \phi_0^2 + (1-\phi_0)^2  \right)\\
        & =\beta_0 \left( p+c_0 \right),
 \end{aligned}
\end{equation*}
and, if $u=v=M$, 
\begin{equation*}
    \begin{aligned}
        &\kappa_{MM}= \sum_{j=1}^M \beta_j\phi_{Mj}^2= \beta_1 \phi_{M1}^2+\beta_2\phi_{M2}^2+ \beta_3 \phi_{M3}^2+...+\beta_M\phi_{MM}^2  \\
       & = \beta_0\phi_0^2+ \beta_0\phi_0^2+\beta_0\phi_0^2+\hdots+ \zeta \beta_0 (1-(M-1)\phi_0)^2  \\ 
    & = \zeta \beta_0 (1-(M-1)\phi_0)^2 +(M-1) \beta_0\phi_0^2 \\ 
    & =(M-1) \beta_0\phi_0^2+\zeta \beta_0 (1-(M-1)\phi_0)^2 \\
             & =\beta_0 \left( M \phi_0^2-\phi_0^2+ \zeta (1-2(M-1)\phi_0+(M-1)^2\phi_0^2 \right) \\ 
 & = \beta_0 \left( M \phi_0^2-\phi_0^2 + \zeta-2 \zeta M\phi_0+2\zeta\phi_0+\zeta M^2\phi_0^2-2\zeta M\phi_0^2+\zeta \phi_0^2 \right)\\ 
 & = \beta_0 \left(\zeta \phi_0^2 + M \phi_0^2-\phi_0^2 + \zeta -2 \zeta M\phi_0 + 2\zeta\phi_0+\zeta M^2\phi_0^2-2\zeta M\phi_0^2 \right)\\ 
   & = \beta_0 \left( \zeta \phi_0^2  +\zeta-2\zeta(M-1)\phi_0+\phi_0^2 \left( M-1-2M\zeta+M^2\zeta \right) \right)\\  
     & = \beta_0 \left( p+ c_1\right).
    \end{aligned}
\end{equation*}
Finally, for any non-diagonal elements of the NGM $\kappa$, when $u \in {1,2,3,4,5,...,M-1}$, and $v \in {1,2,3,4,5,...,M-1}$, and $u \not=v$ we have 
\begin{equation*}
 \begin{aligned}
         &\kappa_{uv}=   \sum_{j=1}^M \beta_j\phi_{uj}\phi_{vj} \\
        &= \beta_1 \phi_{u1}\phi_{v1}+\beta_2 \phi_{u2}\phi_{v2}+\hdots+ \beta_u \phi_{uu}\phi_{vu}+...+\beta_v \phi_{uv}\phi_{vu}...+\beta_M \phi_{uM} \phi_{vM}  \\
          &=\beta_0 0 (0) + \beta_0 .0 + \hdots+ \beta_0 (1-\phi_0).0+\hdots +  0 + \zeta \beta_0 (\phi_0)\phi_0 \\
          &=   \beta_0 ( \zeta \phi_0^2 ) \\
                    & =\beta_0 p,
 \end{aligned}
\end{equation*}
and for $u=1,2,3,4, \hdots, M-1$, 
\begin{equation*}
   \begin{aligned}
          &\kappa_{uM}=   \sum_{j=1}^M \beta_j\phi_{uj}\phi_{Mj} \\
        &= \beta_1 \phi_{u1}\phi_{M1}+\beta_2 \phi_{u2}\phi_{M2}+ \beta_3 \phi_{u3}\phi_{M3}+\hdots+\beta_i \phi_{uu} \phi_{Mu} +\hdots+\beta_M \phi_{uM} \phi_{MM}   \\
          &= \beta_0 0. \phi_0+ \beta_0 0. \phi_0+\beta_0 0. \phi_0+\hdots+\beta_0 (1-\phi_0)\phi_0+...+\zeta \beta_0 \phi_0 (1-(M-1)\phi_0)  \\
          &= \beta_0[ \phi_0-\phi_0^2+\zeta \left(\phi_0- (M-1)\phi_0^2 \right)]  \\
          & =\beta_0 q,
   \end{aligned}
\end{equation*}
and also for $v=1,2,3,4, \hdots, M-1$,
\begin{equation*}
    \begin{split}
        &\kappa_{Mv}=   \sum_{j=1}^M \beta_j\phi_{Mj}\phi_{vj} \\
        &= \beta_1 \phi_{M1}\phi_{v1}+\beta_2 \phi_{M2}\phi_{v2}+ \beta_3 \phi_{M3}\phi_{v3}+\hdots+\beta_i \phi_{Mv} \phi_{vv} +\hdots+\beta_M \phi_{MM} \phi_{vM}   \\
          &= \beta_0 \phi_0.0+ \beta_0  \phi_0 . 0+\beta_0  \phi_0 . 0+\hdots+\beta_0 \phi_0 (1-\phi_0)+...+\zeta \beta_0 \phi_0 (1-(M-1)\phi_0)  \\
          &= \beta_0[ \phi_0-\phi_0^2+\zeta \left(\phi_0- (M-1)\phi_0^2 \right)]  \\
          & =\beta_0 q.
    \end{split}
\end{equation*}

The resulting NGM can be thus written as  
\begin{equation} 
\begin{aligned}\label{NGM2}
&\kappa=\frac{\beta_0}{\gamma} \begin{pmatrix}
   p+c_0  &  p & p & \hdots & q  \\
 p & p+c_0   & p & \hdots & q  \\
                   p & p &  p+c_0 & \hdots & q\\
    \vdots & \vdots & \vdots & \ddots  & \vdots\\
      q  &  q & q & \hdots & p+c_1  \\
        \end{pmatrix} 
        =\frac{\beta_0}{\gamma} \begin{pmatrix}
   p  &  p & p & \hdots & q  \\
            p & p & p & \hdots & q\\
                p & p & p & \hdots & q\\
    \vdots & \vdots  & \vdots & \ddots  & \vdots\\
      q  &  q & q & \hdots & p  \\
        \end{pmatrix} + \frac{\beta_0}{\gamma} \begin{pmatrix}
   c_0  \\
 c_0   \\
  c_0 \\
       \vdots \\
      c_1  \\
        \end{pmatrix} I 
\end{aligned}
    \end{equation}
where $p=\zeta \phi_0^2$, $q=\phi_0-\phi_0^2+\zeta \left(\phi_0- (M-1)\phi_0^2 \right) $, $c_0=(1-\phi_0)^2$, and $$c_1=\zeta-2\zeta(M-1)\phi_0+\phi_0^2 \left( M-1-2M\zeta+M^2\zeta \right).$$  An alternative representation of the NGM can be done by writing 
\begin{equation}
            \kappa= \frac{\beta_0}{\gamma} (\mathcal{P} + \mathcal{D}), \label{StarShapedReturnNGM}
\end{equation}

where $\mathcal{P}=\frac{\beta_0}{\gamma} \begin{pmatrix}
   p  &  p & p & \hdots & q  \\
            p & p & p & \hdots & q\\
                p & p & p & \hdots & q\\
    \vdots & \vdots  & \vdots & \ddots  & \vdots\\
      q  &  q & q & \hdots & p  \\
        \end{pmatrix}$ and the diagonal matrix $\mathcal{D}= \begin{pmatrix}
   c_0  \\
 c_0   \\
  c_0 \\
       \vdots \\
      c_1  \\
        \end{pmatrix} I $.

Since the next generation matrix $\kappa$ in (Eq. \ref{NGM2}) is symmetric, the eigenvalues of $\kappa$ are all real. By doing elementary row or column operations, we see that the rank of the matrix $\mathcal{P}$ is $2$. So all submatrices bigger than $2 \times 2$ have the determinant zero (see, e. g. Friedberg \cite{friedberg_book}). As $\kappa$ in (Eq. \ref{NGM2}) is the sum of a low-rank matrix and the diagonal matrix $\mathcal{D}$; therefore, we write the characteristic polynomial of $\mathcal{P}+\mathcal{D}$ in the form
\begin{equation*}
  (c_0-\lambda)^{M-2} [ \lambda^2-(c_0+c_1+pM)\lambda+\left(c_1c_0+ p(M-1)c_1 +pc_0+(M-1)(p^2-q^2) \right) ]=0 
\end{equation*}
using the expansion (see, Collings \cite{collings}) for the eigenvalues $\lambda$.  
For the square matrix, $\mathcal{P}+\mathcal{D}$, an  eigenvector $V$ and eigenvalue $\lambda$ satisfy the equation 
\begin{equation*}
   \left(\mathcal{P}+\mathcal{D}\right)V= \lambda V,
\end{equation*}
and so for any scalar $R_0^{\text{hom}}$ we have 
\begin{equation*}
   R_0^{\text{hom}} \kappa V = R_0^{\text{hom}}  \left(\mathcal{P}+\mathcal{D}\right)V= R_0^{\text{hom}}\lambda V.
\end{equation*}
Therefore, $R_0^{\text{hom}}\lambda$ is the eigenvalue of the next generation matrix $\kappa$. From the characteristic polynomial, the first eigenvalue of $\kappa$ is $\frac{\beta_0}{\gamma}c_0=R_0^{\text{hom}}c_0=R_0^{\text{hom}}(1-\phi_0)^2$ with multiplicity $M-2$, and the other two eigenvalues are determined from the quadratic term 
\begin{equation}\label{quardraticstar}
    P(\lambda)=\lambda^2-(c_0+c_1+pM)\lambda+\left(c_1c_0+ p(M-1)c_1 +pc_0+(M-1)(p^2-q^2) \right).
\end{equation}
Furthermore, we find the parameter conditions so that all eigenvalues of $\kappa$ are within the unit circle. The first set of eigenvalues automatically satisfies the condition 
\begin{equation} \label{eigenvcon1star}
    R_0^{\text{hom}}(1-\phi_0)^2 <1, 
\end{equation}
since $R_0^{\text{hom}}<1$ by assumption, and $\phi_0 \in   (0,\frac{1}{2}]$. Now, the quadratic (Eq. \ref{quardraticstar}) gives the other two eigenvalues lying within the unit circle which can be tested by applying the Jury conditions (Murray \cite{murray2002mathematical}, page $507$). For a quadratic equation $P(\lambda)=\lambda^2+a_1\lambda+a_0=0$, the Jury conditions tell us $P(1)=1+a_1+a_0>0$, $P(-1)=1-a_1+a_0>0$, and $P(0)=a_0<1$. 

As $R_0^{\text{hom}}\lambda:=\lambda'$ is the eigenvalue of $\kappa$, and so substituting $\lambda=\frac{\lambda'}{R_0^{\text{hom}}}$ in (Eq. \ref{quardraticstar}) we have 
\begin{equation}\label{quardratic2star}
    P(\lambda')=\frac{1}{(R_0^{\text{hom}})^2}[(\lambda')^2-a_1 R_0^{\text{hom}}\lambda'+a_0(R_0^{\text{hom}})^2,
\end{equation}
where $a_0:=\left(c_1c_0+ p(M-1)c_1 +pc_0+(M-1)(p^2-q^2) \right)$ and $a_1:=(c_0+c_1+pM)$.
By applying the Jury conditions to (Eq. \ref{quardratic2star}) we obtain the following inequalities:   
\begin{equation}\label{cond1star}
    \begin{aligned}
          \left\{R_{0}^{\text{hom}} [ \zeta \left( R_{0}^{\text{hom}} (M \phi_0-1)^2-(M-1)\phi_0(M\phi_0-2)-1  \right)  +\phi_0(2-M\phi_0)-1]+1  \right\}>0,
              \end{aligned}
\end{equation}
\begin{equation}\label{cond2star}
    \begin{aligned}
          \left\{R_{0}^{\text{hom}} [ \zeta \left( R_{0}^{\text{hom}} (M \phi_0-1)^2+(M-1)\phi_0(M\phi_0-2)+1  \right)  -\phi_0(2-M\phi_0)+1]+1  \right\}>0,
              \end{aligned}
\end{equation}

\begin{equation}\label{cond3star}
    \begin{aligned}
         (R_{0}^{\text{hom}})^2 \zeta (M \phi_0-1)^2 <1.
              \end{aligned}
\end{equation}

The Jury conditions imply that (Eqs. \ref{cond1star}-\ref{cond3star}) are identical to (Eqs. A.2-A.4) in Stolerman et al. \cite{lucas_dengue_model}. Therefore, the layout of the proof follows the same steps.  
\end{proof}

 Theorem \ref{theoremstar} is not only mathematically remarkable, but it opens a new question:  \emph{Is there a class of networks where the same critical $\zeta$ value and interval $\mathcal{I}$ always emerge?} A common feature of fully connected and star-shaped networks is the connection that all nodes have with the node with a different infection rate (the central node in the star). It is thus natural to explore networks without this feature and observe if the results obtained in Theorem \ref{theoremstar} (and Theorem 2 in \cite{lucas_dengue_model}) still hold. We start the next section with the so-called \emph{cycle-shaped} network, perhaps the most natural example of a network with no such highly connected nodes. 
 

\subsection{Cycle networks}
In the cycle networks, the nodes only connect to their immediate neighbors. These networks are a simple model for small cities or towns where people usually commute short distances. For \emph{unidirectional} cycle networks, people only travel from node $i$ to node $i+1$. This network can be represented by the following flux matrix:

\begin{equation} \label{CycleShape}
\phi_{M \times M}= \begin{pmatrix}
   1-\phi_0  & \phi_0 & 0 & 0 & \hdots  & 0\\
     0  & 1-\phi_0 & \phi_0  & 0 & \hdots & 0\\
     0  & 0 & 1-\phi_0 & \phi_0 & \hdots & 0\\
         \vdots & \vdots &  \vdots & \ddots & \vdots & \vdots\\
           \phi_0  & 0 & 0 & 0 & \hdots &  1-\phi_0
           \end{pmatrix}.
          \end{equation}
    
    where $0 < \phi_0 < 1$. On the other hand, \emph{bidirectional} cycle networks can be summarized as follows by the geometry of the flux matrix:


\begin{equation} \label{CycleShapeReturn}
    \phi_{M \times M}= \begin{pmatrix}
   1-2\phi_0  & \phi_0 & 0 & \hdots & 0 & \phi_0\\
     \phi_0  & 1-2\phi_0 & \phi_0 & \hdots & 0 & 0\\
     0  & \phi_0 & 1-2\phi_0 & \hdots & 0 & 0\\
         \vdots & \vdots &  \vdots & \ddots & \vdots & \vdots\\
           \phi_0  & 0 & 0 & \hdots & \phi_0 &  1-2\phi_0
           \end{pmatrix}.
    \end{equation}
    where $0 < \phi_0 \leq \frac{1}{2}$.  

We start by exploring a particular case of a unidirectional cycle with $M=3$, which indicates that unidirectional cycles are not in the class of networks where the critical value $\zeta$ (Eq. \ref{zeta_critic}) and interval $\mathcal{I}$ (Eq. \ref{interval}) emerge. The proof of the following propositions can be found in the appendix.

\begin{prop} \label{controlCycleNeighbourhood}
Let $V=\left\{1,2,3 \right\}$ be the node of a \emph{unidirectional cycle} network and the flux matrix $\phi_{3 \times 3}$ be defined by (Eq. \ref{CycleShape}). We consider the infection rates $\beta_i$ from (Eq. \ref{BetaConditionstarshape}), where all the 
nodes are stable except the third one, i.e, $R_0^{\text{hom}} < 1$ and $\frac{\zeta \beta_0}{\gamma}>1$. Then we have a critical value   $\zeta^{\text{crit}}$ given by 

\begin{equation} \label{zeta_condition_M=3_circle}
\begin{aligned}
        \zeta^{\text{crit}} =  \frac{1}{R_0^{\text{hom}}} \left( \frac{3 R_0^{\text{hom}}-4}{R_0^{\text{hom}}-2} \right)
\end{aligned}
\end{equation}
such that  if $\frac{1}{R_0^{\text{hom}}}<\zeta< \zeta^{\text{crit}}$, then there exists an interval $\mathcal{I}$ around the point $\frac{1}{2}$ given by 
\begin{equation}\label{unicycleinterval}
    \mathcal{I}= \left( \max \left\{0, \frac{1}{2}-\Tilde{\phi} \right\}, \min \left\{1, \frac{1}{2}+\Tilde{\phi} \right\} \right),
\end{equation}
where 
\begin{equation} \label{cycle_tilda_phi_M3}
    \begin{aligned} 
       & \tilde{\phi} =\frac{1}{2} \sqrt{ \frac{ \left(\zeta  R_0^{\text{hom}}(R_0^{\text{hom}}-2)+(4-3R_0^{\text{hom}})\right)}{R_0^{\text{hom}} \left(1+2\zeta-3\zeta R_0^{\text{hom}}
        \right)}}
    \end{aligned}
\end{equation}
such that if $\phi_0 \in \mathcal{I}$, the DFE is stable, and otherwise, it is unstable.   
\end{prop}

The  $\zeta^{\text{crit}}$ found for \emph{unidirectional cycle} networks  (Eq. \ref{zeta_condition_M=3_circle}) are different from the \emph{star-shaped} networks  (Eq. \ref{zeta_condition_M=3}). The interval where there will be no epidemic is given by (Eq. \ref{unicycleinterval}) in the \emph{unidirectional cycle} networks, which is different from the interval obtained for \emph{star-shaped} networks (Eq. \ref{interval_star_m_3}). We conjecture that this difference also holds for $M>3$.  In section \ref{performence}, we run numerical simulations for different network sizes and observe that our conjecture is likely true.  

If we connect all nodes of a cycle network to the last node of index $M$, we can build different networks and further analyze the corresponding epidemic thresholds.   In the following section, we combine the bidirectional cycle-shaped networks with the star-shaped networks.  

\subsection{Extension of heterogeneous cycle networks}
This section aims to modify bidirectional cycle networks and investigate where we find the minimum $\zeta$ (Eq. \ref{zeta_critic}) and the interval $\mathcal{I}$ (Eq. \ref{interval}). Specifically, we modify the bidirectional cycle by joining the heterogeneous node $M$ directly to the other nodes across the cycle. In this case, the adjacent nodes to the heterogeneous node are disjoint. We call this network \emph{cycle support}. We then extend these simple heuristic networks to the \emph{star-cycle} networks by joining the heterogeneous node's immediate nodes $1$ and $M-1$.       

\subsubsection{Cycle support networks}\label{cycle_supports} 
We introduce new networks called \emph{cycle support}, formed by the bidirectional cycle networks and joining the heterogeneous node $M$ directly to the nodes across the cycle. Their general flux matrix is given by:  

   \begin{equation} \label{cycleBacknodes}
    \phi_{M \times M}= \begin{pmatrix}
   1-2\phi_0  & \phi_0 & 0 & 0& \hdots & 0 & \phi_0\\
     \phi_0  & 1-3\phi_0 & \phi_0 & 0 & \hdots & 0 & \phi_0\\
     0  & \phi_0 & 1-3\phi_0 & \phi_0 & \hdots & 0 & \phi_0\\
         \vdots & \vdots &  \vdots & \vdots & \ddots & \vdots & \vdots\\
         0  & 0 & 0 &  \hdots &  \phi_0 & 1-2\phi_0 & \phi_0\\
           \phi_0  &  \phi_0 &     \phi_0 &    \hdots  & \phi_0 & \phi_0 &  1-(M-1)\phi_0
           \end{pmatrix},
    \end{equation}
where $0< \phi_0 \leq \frac{1}{M-1}$, since $\phi_{ij} \geq 0$. We then obtain the following proposition for \emph{cycle support} networks with $5$ nodes (the proof can be found in the Appendix). 

\begin{prop}  \label{star_cycle_without_adj}
Let $V=\left\{1,2,3,4,5 \right\}$ be a node of \emph{cycle-supports} network and let the geometry of its flux matrix be defined by (Eq. \ref{cycleBacknodes}), where $\phi_0 \in (0, \frac{1}{4}]$. The infection rate $\beta_i$ is defined in (Eq. \ref{BetaConditionstarshape}), and we also assume that all the outer nodes are stable ($R_0^{\text{hom}}=\frac{\beta_0}{\gamma}<1$) except the center node $M=5$. Then the critical value of $\zeta$ called $\zeta^{\text{crit}}$ is given by 
\begin{equation*} 
    \zeta^{\text{crit}}= 1+\left( \frac{1-R_0^{\text{hom}}}{R_0^{\text{hom}}} \right)5.
\end{equation*}
Alternatively, the infection threshold of $\beta_0$, $\beta_0^{\text{crit}}$, is given by 
\begin{equation} \label{beta0_condition_M=5Cycle_support}
    \beta_0^{\text{crit}}= \gamma\left( \frac{5}{4+\zeta} \right).
\end{equation}
If $\frac{1}{R_0^{\text{hom}}}< \zeta< \zeta^{\text{crit}}$, then there is an interval around the point $\frac{1}{5}$ given by $\mathcal{I}= \left( \max \left\{0, \frac{1}{5}-\Tilde{\phi} \right\}, \min \left\{\frac{1}{4}, \frac{1}{5}+\Tilde{\phi} \right\} \right)$, where 
\begin{equation*}
    \begin{aligned}
        \Tilde{\phi}= \frac{1}{5} \sqrt{ \frac{\gamma (5\gamma-4\beta_0-\zeta\beta_0 )}{\beta_0(\gamma+4\zeta\gamma-5\zeta\beta_0)}}= \frac{1}{5} \sqrt{ \frac{ 5(1-R_0^{\text{hom}})-R_0^{\text{hom}}(\zeta-1) )}{R_0^{\text{hom}} \left(5 \zeta(1-R_0^{\text{hom}})-(\zeta-1) \right)}},
    \end{aligned}
\end{equation*}
such that if $\phi_0 \in \mathcal{I}$, the DFE is stable; and otherwise it is unstable.  
\end{prop}

\begin{remark}
 We observe that if $\zeta <\frac{1}{R_0^{\text{hom}}}$, the DFE is always stable.  Note that we have a similar proof to $\{1, \hdots,5\}$, \ref{star_cycle_without_adj} for the nodes $M=4,5,6,7$. On the other hand, simulations can be performed for every node $M$. 
\end{remark}

 

\subsubsection{Star-cycle networks}
For star-cycle networks, we start from the heterogeneous node $M$ as the center node of the network. Then the center node $M$ is connected to all other outer nodes to create a \emph{star-shaped} network. Finally, the rest of the nodes are in the form of bidirectional cycle networks called \emph{star-cycle} networks. The geometry of the flux matrix is given by:       
\begin{equation} \label{StarCycleReturn}
    \phi_{M \times M}= \begin{pmatrix}
   1-3\phi_0  & \phi_0 & 0 & 0& \hdots & \phi_0 & \phi_0\\
     \phi_0  & 1-3\phi_0 & \phi_0 & 0 & \hdots & 0 & \phi_0\\
     0  & \phi_0 & 1-3\phi_0 & \phi_0 & \hdots & 0 & \phi_0\\
         \vdots & \vdots &  \vdots & \vdots & \ddots & \vdots & \vdots\\
         \phi_0  & 0 & 0 &  \hdots &  \phi_0 & 1-3\phi_0 & \phi_0\\
           \phi_0  &  \phi_0 &     \phi_0 &    \hdots  & \phi_0 & \phi_0 &  1-(M-1)\phi_0
           \end{pmatrix},
    \end{equation}
where $0< \phi_0 \leq \frac{1}{M-1}$, since $\phi_{ij} \geq 0$. Note that the \emph{cycle-support} networks described in Section \ref{cycle_supports} are \emph{star-cycle} networks, where the adjacent node of $M$ does not join with $M$.  We then obtain the following proposition for \emph{star-cycle} networks with $5$ nodes (the proof can be found in the Appendix). 

\begin{prop} \label{star_cycle_th}
Let $V=\left\{1,2,3,4,5 \right\}$ be the nodes of a \emph{star-cycle} network, and the geometry of the corresponding network is defined by the flux (Eq. \ref{StarCycleReturn}), where $M=5, ~\text{and} ~\phi_0 \in (0, \frac{1}{4}]$. The infection rate $\beta_i$ is defined in (Eq. \ref{BetaConditionstarshape}), and we also assume that all the 
nodes are stable ($R_0^{\text{hom}}=\frac{\beta_0}{\gamma}<1$) except the node $M=5$. Then the critical value of $\zeta$ called $\zeta^{\text{crit}}$ is given by 
\begin{equation} \label{zeta_condition_M=5Cycle_star}
    \zeta^{\text{crit}}= 1+\left( \frac{1-R_0^{\text{hom}}}{R_0^{\text{hom}}} \right)5.
\end{equation}
Alternatively, the infection threshold of $\beta_0$, $\beta_0^{\text{crit}}$ is given by 
\begin{equation} \label{beta0_condition_M=5Cycle_star}
    \beta_0^{\text{crit}}= \gamma\left( \frac{5}{4+\zeta} \right).
\end{equation}
If $\frac{1}{R_0^{\text{hom}}}< \zeta< \zeta^{\text{crit}}$, there is an interval around the point $\frac{1}{5}$ given by $\mathcal{I}= \left( \max \left\{0, \frac{1}{5}-\Tilde{\phi} \right\}, \min \left\{\frac{1}{4}, \frac{1}{5}+\Tilde{\phi} \right\} \right)$, where 
\begin{equation*}
    \begin{aligned}
        \Tilde{\phi}= \frac{1}{5} \sqrt{ \frac{\gamma (5\gamma-4\beta_0-\zeta\beta_0 )}{\beta_0(\gamma+4\zeta\gamma-5\zeta\beta_0)}}= \frac{1}{5} \sqrt{ \frac{ 5(1-R_0^{\text{hom}})-R_0^{\text{hom}}(\zeta-1) )}{R_0^{\text{hom}} \left(5 \zeta(1-R_0^{\text{hom}})-(\zeta-1) \right)}},
    \end{aligned}
\end{equation*}
such that if $\phi_0 \in \mathcal{I}$, the DFE is stable, and otherwise it is unstable.      
\end{prop}

We observe that, if $\zeta <\frac{1}{R_0^{\text{hom}}}$, DFE is always stable. Using symbolic calculations with the software Mathematica, we can also perform a similar analysis for  $M=4,5,6,7,9,11, 13$. In the following conjecture, we expand the Proposition \ref{star_cycle_th} to an arbitrary size of \emph{star-cycle} networks.

\begin{conj} \label{starcycle_conj}
Let $V=\left\{1,2,3, \hdots, M \right\}$ be the nodes of a \emph{star-cycle} network defined by the symmetric flux matrix (Eq. \ref{StarCycleReturn}), where $0< \phi_0 \leq \frac{1}{M-1}$. We assume that the infection rate $\beta_i = \beta_0$ at node $i$ for all $i = 1,2,3, \cdots, M-1$ and $\beta_M = \zeta \beta_0$. We also suppose that all the 
nodes are stable $(R_0^{\text{hom}}=\frac{\beta_0}{\gamma}<1)$  except the center $M$. Then we have a minimum  $\zeta^{\text{crit}}$ given by (Eq. \ref{zeta_critic} or \ref{zeta_critic_star}). If $\frac{1}{R_0^{\text{hom}}}<\zeta< \zeta^{\text{crit}}$, there exists an interval $\mathcal{I}= \left( \max \left\{0, \frac{1}{M}-\Tilde{\phi} \right\}, \min \left\{\frac{1}{M-1}, \frac{1}{M}+\Tilde{\phi} \right\} \right)$, where $\Tilde{\phi}$ is given by  (Eq. \ref{interv_critic}),
such that if $\phi_0 \in \mathcal{I}$, then the DFE is stable. Otherwise, it is unstable. Finally, if  $\zeta<\frac{1}{R_0^{\text{hom}}}$ i.e. $\beta_M < \gamma$, then DFE is always stable.
\end{conj}

\subsection{More networks based on the star-shaped geometry}

In the previous sections, we studied networks with different structures, but assuming an equal flux $\phi_0$ between any two connected nodes. Realistically, the fluxes between nodes could be different. In this section, we explore more general networks with some of the structure of star-shaped networks, but assuming the existence of two different flux values $\phi_0$ and $\phi_1$.


\subsubsection{Star-triangle networks}

We consider a \emph{star-triangle} network with an odd number of nodes $M$. The nonnegative entries of its flux matrix $\phi_{M \times M}$ are defined by  

 \begin{equation} \label{startriangleflux}
 \phi_{ij}= \begin{cases}
                                   \phi_0 & \text{if $i=M, ~\text{and}~ j \in V \setminus M $} ~ \text{ or ~if $j=M, ~ \text{and} ~ i \in V\setminus M $} \\
                                 1-\phi_0-\phi_1  &  \text{if $i=j,~ \text{and} ~ ~ i,j \in V \setminus M$}\\
                                 1-(M-1)\phi_0  &  \text{if $i=j=M$}\\
                                  \phi_1  &  \text{for all $i \in V \setminus \{M,M-1\}  ~ \&~ j=i+1$}\\
                                 \phi_1  &  \text{for all $j \in V \setminus \{M,M-1\}  ~ \&~ i=j+1$}\\
                                  0 & \text{otherwise},
                                     \end{cases}
\end{equation}
where $\phi_0+\phi_1<1, 0<\phi_0 \leq \frac{1}{M-1}$, and $\phi_1>0$. Alternatively, the explicit form of the flux matrix (Eq. \ref{starshape}) is given by 

\begin{equation*} 
    \phi_{M \times M}= \begin{pmatrix}
   1-\phi_0-\phi_1  & \phi_1 & \hdots  & 0 & 0 & \phi_0\\
     \phi_1 & 1-\phi_0-\phi_1 &  \hdots & 0 & 0  &\phi_0\\
        \vdots & \vdots &   \ddots &  \vdots & \vdots &\vdots\\
            0  & 0 &  \hdots & 1-\phi_0-\phi_1 & \phi_1 &\phi_0\\
            0 & 0   & \hdots & \phi_1 & 1-\phi_0-\phi_1 & \phi_0\\
           \phi_0  & \phi_0 &   \hdots & \phi_0 & \phi_0 & 1-(M-1)\phi_0
           \end{pmatrix}.
    \end{equation*}

In the following theorem, we find the exact minimum $\zeta$ and interval $\mathcal{I}$ for \emph{star-triangle} networks as we do for \emph{fully connected} networks, where the flux between nodes can prevent an epidemic outbreak if it lies in a specific interval. The proof can be found in the appendix. 

\begin{thm} \label{theoremstartriangle}
Let $V=\left\{1,2,3, \hdots, M \right\}$ be the nodes of a star-triangle network defined by the flux matrix (\ref{startriangleflux}), where $\phi_0+\phi_1<1, 0<\phi_0 \leq \frac{1}{M-1}$, and $\phi_1>0$. We assume that the infection rate $\beta_i = \beta_0$ at node $i$ for all $i = 1,2,3, \cdots, M-1$ and $\beta_M = \zeta \beta_0$. We also suppose that all the 
nodes are stable $(R_0^{\text{hom}}=\frac{\beta_0}{\gamma}<1)$  except the center node $M$. Then we have a minimum  $\zeta^{\text{crit}}$ given by (Eq. \ref{zeta_critic} or \ref{zeta_critic_star}). If $\frac{1}{R_0^{\text{hom}}}<\zeta< \zeta^{\text{crit}}$, there exists an interval $\mathcal{I}= \left( \max \left\{0, \frac{1}{M}-\Tilde{\phi} \right\}, \min \left\{\frac{1}{M-1}, \frac{1}{M}+\Tilde{\phi} \right\} \right)$, where $\Tilde{\phi}$ is given by (Eq. \ref{interv_critic}),
such that if $\phi_0 \in \mathcal{I}$, then the DFE is stable, and otherwise it is unstable.     
\end{thm}

In the previous theorem, we note that if $\zeta<\frac{1}{R_0^{\text{hom}}}$, i.e., $\beta_M < \gamma$, then DFE is always stable. Let us now generalize the flux matrix (Eq. \ref{startriangleflux}) accounting for the other nodes with different fluxes in the following way: 

\begin{equation} \label{genstartriangle}
 \phi_{M \times M}= \begin{cases}
                                \phi_{2i-1,2i}=   \phi_{i}  & \text{ $1 \leq i \leq  (\frac{M-1}{2})$} \\
                                \phi_{2i,2i-1}=   \phi_{i}  & \text{ $1 \leq i \leq (\frac{M-1}{2})$} \\
                                \phi_{M,i}=   \phi_{(\frac{M+1}{2})}  & \text{ $1 \leq i \leq M-1$} \\
                                \phi_{i,M}=   \phi_{(\frac{M+1}{2})}  & \text{ $1 \leq i \leq M-1$}\\
                                 \phi_{i,i}= 1- \sum_{j=1}^{M} \phi_{i,j}  & \text{ $1 \leq i \leq M-1$ ~ and ~$i \neq j$}\\
                                 \phi_{M,M}=1-(M-1) \phi_{(\frac{M+1}{2})} \\
                                 0  & \text{otherwise},
   \end{cases}
    \end{equation}
where for all $\phi_{i,j} >0, i , j \in V$, and $M$ is an odd number. The fraction of movement from the center node $M$ to the other nodes is defined as $ \phi_{(\frac{M+1}{2})}=\phi_0$ to be consistent with the other networks. We prove the following result for this more general \emph{star-triangle} network for the network size $M=11$ (the proof can be found in the Appendix).

\begin{prop} \label{star_triangle_gen_M_11}
Let $V=\left\{1,2,3,4, \hdots, 11 \right\}$ be the nodes of a star-triangle network defined by the symmetric flux matrix (Eq. \ref{genstartriangle}) and $M=11$ be an odd number. We assume that the infection rate $\beta_i = \beta_0$ at node $i$ for all $i = 1,2,3, \cdots, 10$ and $\beta_{11} = \zeta \beta_0$. We also suppose that all the nodes are stable $(R_0^{\text{hom}}=\frac{\beta_0}{\gamma}<1)$  except the center node $M$. Then the critical value of $\zeta$ called $\zeta^{\text{crit}}$ is given by 

\begin{equation} \label{star_triangle_gen_M_11_zeta_cond}
    \zeta^{\text{crit}}= 1+\left( \frac{1-R_0^{\text{hom}}}{R_0^{\text{hom}}} \right)11. 
\end{equation}
Alternatively, the infection threshold of $\beta_0$, $\beta_0^{\text{crit}}$ is given by 
\begin{equation} \label{beta0_condition_M=11star_triangle}
    \beta_0^{\text{crit}}= \gamma\left( \frac{11}{10+\zeta} \right).
\end{equation}
If $\frac{1}{R_0^{\text{hom}}}<\zeta< \zeta^{\text{crit}}$, there exists an interval $\mathcal{I}= \left( \max \left\{0, \frac{1}{11}-\Tilde{\phi} \right\}, \min \left\{\frac{1}{10}, \frac{1}{11}+\Tilde{\phi} \right\} \right)$, where $\Tilde{\phi}$ is given by (Eq. \ref{interv_critic}),
such that if $\phi_0 \in \mathcal{I}$, then the DFE is stable, and otherwise it is unstable. 
\end{prop}

Also if  $\zeta<\frac{1}{R_0^{\text{hom}}}$, i.e., $\beta_M < \gamma$, in Proposition \ref{star_triangle_gen_M_11}, then the DFE is always stable. Proposition  \ref{star_triangle_gen_M_11} also holds for odd $M$ values $M=3, 5, 7, 9$ and higher. Our results for particular network sizes can be thus generalized in the following conjecture.

\begin{conj} \label{startrinagleconjecture}
Let $V=\left\{1,2,3,4, \hdots, M \right\}$ be the nodes of a \emph{star-triangle} network defined by the symmetric flux matrix (Eq. \ref{genstartriangle}), and let $M$ be an odd number. We assume that the infection rate $\beta_i = \beta_0$ at node $i$ for all $i = 1,2,3, \cdots, M-1$ and $\beta_M = \zeta \beta_0$. We also suppose that all the nodes are stable ($R_0^{\text{hom}}=\frac{\beta_0}{\gamma}<1$)  except the center node $M$. Then we have a minimum  $\zeta^{\text{crit}}$ given by (Eq. \ref{zeta_critic} or \ref{zeta_critic_star}). If $\frac{1}{R_0^{\text{hom}}}<\zeta< \zeta^{\text{crit}}$, there exists an interval $\mathcal{I}= \left( \max \left\{0, \frac{1}{M}-\Tilde{\phi} \right\}, \min \left\{\frac{1}{M-1}, \frac{1}{M}+\Tilde{\phi} \right\} \right)$, where $\Tilde{\phi}$ is given by (Eq. \ref{interv_critic}),
such that if $\phi_0 \in \mathcal{I}$, then the DFE is stable, otherwise, it is unstable. Finally, if  $\zeta<\frac{1}{R_0^{\text{hom}}}$, i.e., $\beta_M < \gamma$, then the DFE is always stable.
\end{conj}

\subsubsection{Star-background networks}
In Section \ref{StarShapeSection}, we studied \emph{star-shaped} networks and introduced theoretical estimates for stability conditions for networks of arbitrary size $M$. In Section \ref{previous_results}, we reviewed a previous result for \emph{fully connected} networks (Theorem \ref{theoremfullcon}). In this section, we define \emph{star-background} networks, from which both fully connected and star-shaped networks can be obtained.  Specifically, we consider a heterogeneous node $M$ (the one with a different infection rate) connecting to all other nodes with the same flux $\phi_0$, and all the outer nodes are connected to all other nodes by a different flux ($\phi_1$). The corresponding flux matrix is given by
\begin{equation} \label{starshapebackgroundgn}
 \phi_{ij}= \begin{cases}
                                   \phi_0 & \text{if $i=M, ~\text{and}~ j \in V$} ~ \text{ or ~if $j=M, ~ \text{and} ~ i \in V$} \\
                                 1-\phi_0-(M-2)\phi_1  &  \text{if $i=j,~ \text{and} ~ ~ i,j \in V \setminus M$}\\
                                 1-(M-1)\phi_0  &  \text{if $i=j=M$}\\
                                \phi_1 & \text{otherwise},
   \end{cases}
\end{equation}
where $\phi_0+(M-2)\phi_1<1, 0<\phi_0 \leq \frac{1}{M-1}$, and $\phi_1>0$. The explicit form of the flux matrix (Eq. \ref{starshapebackgroundgn}) is the following: 

\begin{equation} \label{starshapebackground}
    \phi_{M \times M}= \begin{pmatrix}
   1-\phi_0-(M-2)\phi_1  & \phi_1 & \phi_1 & \hdots & \phi_0\\
     \phi_1 & 1-\phi_0-(M-2)\phi_1 & \phi_1 &  \hdots &\phi_0\\
            \phi_1  & \phi_1 & 1-\phi_0-(M-2)\phi_1 &  \hdots &\phi_0\\
         \vdots & \vdots &   \vdots &  \ddots & \vdots\\
             \phi_0  & \phi_0 & \phi_0 &  \hdots & 1-(M-1)\phi_0
           \end{pmatrix}.
    \end{equation}

The above flux matrix somehow connects both \emph{fully connected} and \emph{star-shaped} networks. In fact, when $\phi_1 = 0$ we obtain \emph{star-shaped} networks. On the other hand, when $\phi_1=\phi_0$, we obtain the \emph{fully connected} networks. \emph{Star-background} networks are thus an interesting ``bridge'' between fully connected and \emph{star-shaped} networks. The following theorem presents the epidemic thresholds for this class of networks (see appendix for the proof).

\begin{thm} \label{theoremstarbg} 
Let $V=\left\{1,2,3, \hdots, M \right\}$ be the nodes of a \emph{star-background} networks defined by the flux matrix (Eq. \ref{starshapebackgroundgn}), where $\phi_0+(M-2)\phi_1<1, 0<\phi_0 \leq \frac{1}{M-1}$, and $\phi_1>0$. We assume that the infection rate $\beta_i = \beta_0$ at node $i$ for all $i = 1,2,3, \cdots, M-1$ and $\beta_M = \zeta \beta_0$. We also suppose that all the 
nodes are stable ($R_0^{\text{hom}}=\frac{\beta_0}{\gamma}<1$)  except the center $M$. Then we have a minimum  $\zeta^{\text{crit}}$ given by (Eq. \ref{zeta_critic} or \ref{zeta_critic_star}). If $\frac{1}{R_0^{\text{hom}}}<\zeta< \zeta^{\text{crit}}$, there exists an interval $\mathcal{I}= \left( \max \left\{0, \frac{1}{M}-\Tilde{\phi} \right\}, \min \left\{\frac{1}{M-1}, \frac{1}{M}+\Tilde{\phi} \right\} \right)$, where $\Tilde{\phi}$ is given by (Eq. \ref{interv_critic}),
such that if $\phi_0 \in \mathcal{I}$, then the DFE is stable, and otherwise it is unstable. Finally, if  $\zeta<\frac{1}{R_0^{\text{hom}}}$, i.e., $\beta_M < \gamma$, then the DFE is always stable.     
\end{thm}

\begin{remark}
In particular, if $\phi_1=\phi_0$ in (Eq. \ref{starshapebackgroundgn}), then the flux creates the \emph{fully connected} networks described in Theorem \ref{theoremfullcon}, and the results of Theorem \ref{theoremstarbg} are identical to Theorem \ref{theoremfullcon}. Furthermore, $\phi_1= 0$ in (Eq. \ref{starshapebackgroundgn}) the network becomes \emph{star-shaped} (Fig. \ref{fig:diff_networks}) network revealed in Theorem \ref{theoremstar}. Then the results of the Theorem \ref{theoremstarbg} are identical with Theorem \ref{theoremstar}.
\end{remark}

\subsubsection{Conjectures on more flexible network structures} \label{gen_net_models}

To extend the previous section, we have constructed a general network diagram representing a \emph{star-class} network. The network construction begins by connecting a center node $M$ to outer nodes with flux value $\phi_0$ (by the solid red arrows in Fig. \ref{fig:StarToGeneral}). Then all the outer or external neighboring nodes are joined (solid green arrows) with the new flux parameter $\phi_1$. Next, we connect any other outer nodes with a gap of one node between them by the flux $\phi_{2}$ (solid blue line). Then we connect all the external nodes with a spacing of two nodes between them with the new flux parameter $\phi_3$ (maroon dashed line) and so on. Finally, the general flux of two interconnecting outer nodes is $\phi_{p+1}$, where p is the number of nodes between them. 

\begin{figure}[!htbp] 
\centering
   \includegraphics[width=0.4\textwidth]{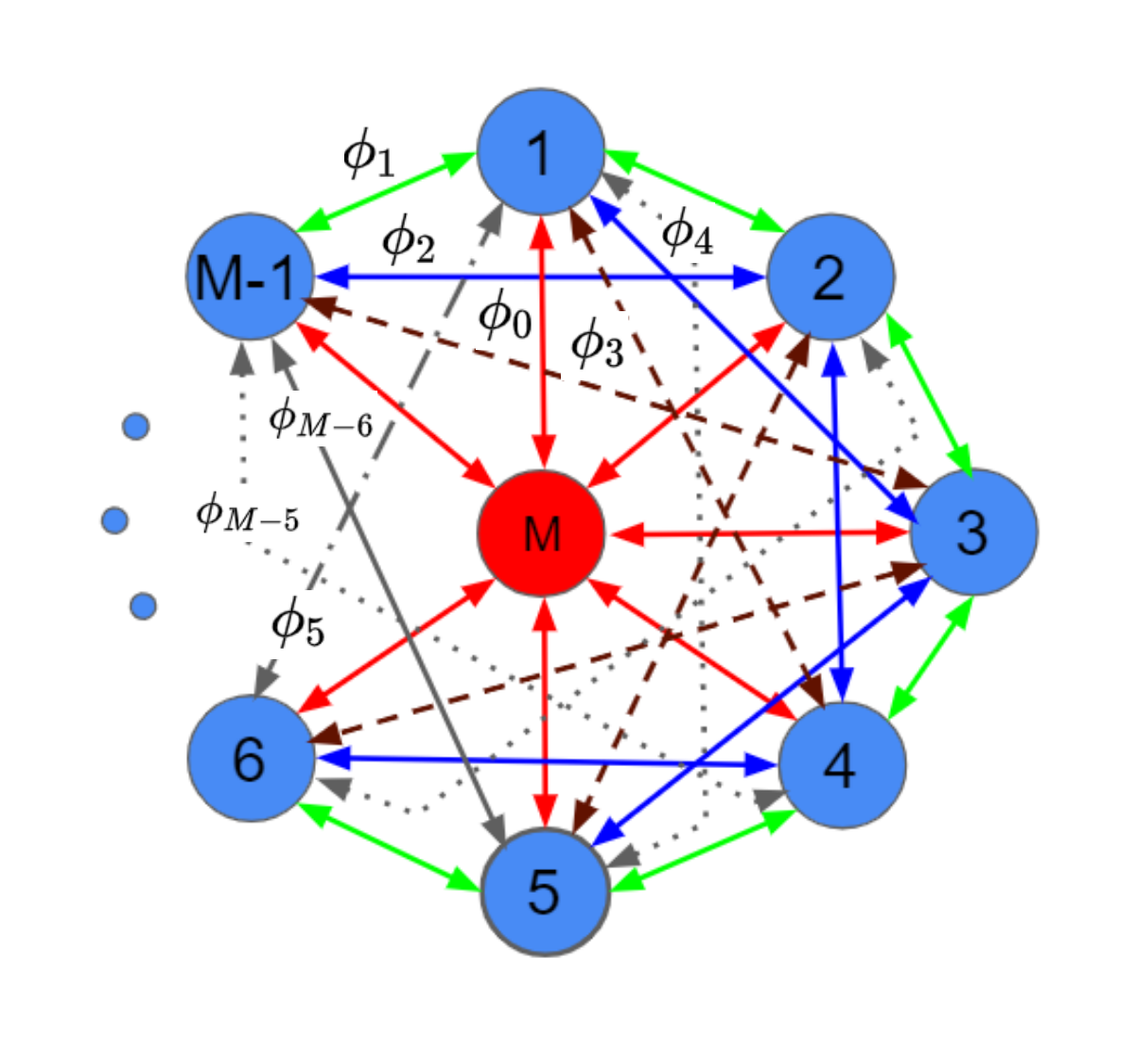}
  \caption{\textbf{A more flexible \emph{star-class} network of $M$ nodes:} a center node is connected to the outer nodes with a flux $\phi_0$. The outer notes with a gap of $i$ nodes are connected with a different flux $\phi_i$. } 
  \label{fig:StarToGeneral} 
\end{figure}  
\vspace{0.12cm}
 
The flux matrix for a network with an odd number $M$ of nodes can be written as

{
   \centering
   \begin{adjustwidth}{-.5in}{0in} 
\footnotesize
\setlength{\arraycolsep}{2.5pt} 
\medmuskip = 1mu 
\begin{equation}\label{generalfluxModd}
    \phi_{M \times M}= \begin{pmatrix}
   \phi_d & \phi_1  & \phi_2 & \phi_3 & \hdots & \phi_{\frac{M-1}{2}-1}  & \phi_{\frac{M-1}{2}} & \phi_{\frac{M-1}{2}-1} & \phi_{\frac{M-1}{2}-2} & \hdots & \phi_4  & \phi_3 & \phi_2 & \phi_1 & \phi_0\\
    \phi_1 & \phi_d  & \phi_1 & \phi_2 & \hdots & \phi_{\frac{M-1}{2}-2} & \phi_{\frac{M-1}{2}-1} &  \phi_{\frac{M-1}{2}} & \phi_{\frac{M-1}{2}-1} & \hdots & \phi_5 & \phi_4 & \phi_3 & \phi_2 & \phi_0 \\
 \phi_2 & \phi_1 & \phi_d  & \phi_1 & \hdots & \phi_{\frac{M-1}{2}-3} & \phi_{\frac{M-1}{2}-2} & \phi_{\frac{M-1}{2}-1} &  \phi_{\frac{M-1}{2}}  & \hdots  & \phi_6 & \phi_5 & \phi_4 & \phi_3 & \phi_0 \\
  \phi_3 & \phi_2 & \phi_1 & \phi_d  & \hdots & \phi_{\frac{M-1}{2}-4} & \phi_{\frac{M-1}{2}-3} &  \phi_{\frac{M-1}{2}-2} & \phi_{\frac{M-1}{2}-1}  & \hdots & \phi_7 & \phi_6 & \phi_5 & \phi_4 & \phi_0 \\
\vdots & \vdots & \vdots & \vdots & \ddots & \vdots & \vdots & \vdots  & \vdots & \ddots &  \vdots & \vdots & \vdots &  \vdots  & \vdots \\   

\phi_{\frac{M-1}{2}-1} & \phi_{\frac{M-1}{2}-2} & \phi_{\frac{M-1}{2}-3} & \phi_{\frac{M-1}{2}-4}  & \hdots & \phi_{d} & \phi_{1} & \phi_{2} &  \phi_{3}  & \hdots &  \phi_{\frac{M-1}{2}-3} & \phi_{\frac{M-1}{2}-2} & \phi_{\frac{M-1}{2}-1} &  \phi_{\frac{M-1}{2}} & \phi_0 \\

\phi_{\frac{M-1}{2}} & \phi_{\frac{M-1}{2}-1} & \phi_{\frac{M-1}{2}-2} & \phi_{\frac{M-1}{2}-3}  & \hdots & \phi_{1} & \phi_{d} &  \phi_{1} & \phi_{2}  & \hdots & \phi_{\frac{M-1}{2}-4}  & \phi_{\frac{M-1}{2}-3} & \phi_{\frac{M-1}{2}-2} & \phi_{\frac{M-1}{2}-1} & \phi_0 \\

\phi_{\frac{M-1}{2}-1} & \phi_{\frac{M-1}{2}} & \phi_{\frac{M-1}{2}-1} & \phi_{\frac{M-1}{2}-2}  & \hdots & \phi_{2} & \phi_{1} &  \phi_{d} & \phi_{1}  & \hdots & \phi_{\frac{M-1}{2}-5}  & \phi_{\frac{M-1}{2}-4} & \phi_{\frac{M-1}{2}-3} & \phi_{\frac{M-1}{2}-2} & \phi_0 \\

\phi_{\frac{M-1}{2}-2} & \phi_{\frac{M-1}{2}-1} & \phi_{\frac{M-1}{2}} & \phi_{\frac{M-1}{2}-1}  & \hdots & \phi_{3} & \phi_{2} &  \phi_{1} & \phi_{d}  & \hdots & \phi_{\frac{M-1}{2}-6}  & \phi_{\frac{M-1}{2}-5} & \phi_{\frac{M-1}{2}-4}& \phi_{\frac{M-1}{2}-3} & \phi_0 \\

\vdots & \vdots & \vdots & \vdots & \ddots & \vdots & \vdots & \vdots  & \vdots & \ddots &  \vdots & \vdots & \vdots &  \vdots  & \vdots \\   

 \phi_4 & \phi_5  & \phi_6 & \phi_7 & \hdots & \phi_{\frac{M-1}{2}-3} & \phi_{\frac{M-1}{2}-4} &  \phi_{\frac{M-1}{2}-5} & \phi_{\frac{M-1}{2}-6} & \hdots & \phi_d & \phi_1 & \phi_2 & \phi_3 & \phi_0 \\
 
\phi_3 & \phi_4  & \phi_5 & \phi_6 & \hdots & \phi_{\frac{M-1}{2}-2} & \phi_{\frac{M-1}{2}-3} &  \phi_{\frac{M-1}{2}-4} & \phi_{\frac{M-1}{2}-5} & \hdots & \phi_1 & \phi_d & \phi_1 & \phi_2 & \phi_0 \\

\phi_2 & \phi_3  & \phi_4 & \phi_5 & \hdots & \phi_{\frac{M-1}{2}-1} & \phi_{\frac{M-1}{2}-2} &  \phi_{\frac{M-1}{2}-3} & \phi_{\frac{M-1}{2}-4} & \hdots & \phi_2 & \phi_1 & \phi_d & \phi_1 & \phi_0 \\
  
 \phi_1 & \phi_2  & \phi_3 & \phi_4 & \hdots & \phi_{\frac{M-1}{2}} & \phi_{\frac{M-1}{2}-1} &  \phi_{\frac{M-1}{2}-2} & \phi_{\frac{M-1}{2}-3} & \hdots & \phi_3 & \phi_2 & \phi_1 & \phi_d & \phi_0 \\

\phi_0 & \phi_0 & \phi_0 & \phi_0 & \hdots & \phi_0 & \phi_0 & \phi_0 & \phi_0 & \hdots & \phi_0 & \phi_0 & \phi_0 & \phi_0  & \phi_{MM}
\end{pmatrix}
\end{equation}
  \end{adjustwidth}
}

where $\phi_d:=1-2(\phi_0+\phi_1+ \hdots + \phi_{\frac{M-1}{2}-1})-\phi_{\frac{M-1}{2}}-\phi_0$ and $\phi_{MM}=1-(M-1)\phi_0$. The interested reader can find a similar matrix for a network of an even number of nodes $M$ in the appendix. 

We conjecture that the \emph{star-class} network has the same epidemic threshold as the fully connected networks. The following proposition can be proved for $M=9$ (the proof can be found in the appendix).

\begin{prop} \label{general_M9}
Let $V=\left\{1,2, \hdots, 9\right\}$ be the nodes of a \emph{general networks} defined by the symmetric flux matrix (Eq. \ref{generalfluxModd}). We assume that the infection rate $\beta_i = \beta_0$ at node $i$ for all $i = 1,2,3, \hdots, 9$ and $\beta_9 = \zeta \beta_0$. We also suppose that all the nodes are stable i.e. $R_0^{\text{hom}}=\frac{\beta_0}{\gamma}<1$  except the center node $M$. Then the critical value of $\zeta$ called $\zeta^{\text{crit}}$ is given by 
\begin{equation*} 
    \zeta^{\text{crit}}= 1+\left( \frac{1-R_0^{\text{hom}}}{R_0^{\text{hom}}} \right)9.
\end{equation*}
Alternatively, the infection threshold of $\beta_0$, $\beta_0^{\text{crit}}$ is given by 
\begin{equation*} 
    \beta_0^{\text{crit}}= \gamma\left( \frac{9}{8+\zeta} \right).
\end{equation*}
If $\frac{1}{R_0^{\text{hom}}}<\zeta< \zeta^{\text{crit}}$, and taking into account the restriction $\phi_{ij}$, where $i, j = 1,2,3, \hdots, 9$, then the interval where there will be no epidemic is given by $\mathcal{I}(\zeta,R_0^{\text{hom}})=\left( \max \left\{0, \frac{1}{9}-\Tilde{\phi} \right\}, \min \left\{\frac{1}{8}, \frac{1}{9}+\Tilde{\phi} \right\} \right)$, where $\Tilde{\phi}$ is given by 

\begin{equation*}
    \begin{aligned}
        \Tilde{\phi}=\frac{1}{9} \sqrt{ \frac{ 9(1-R_0^{\text{hom}})-R_0^{\text{hom}}(\zeta-1) )}{R_0^{\text{hom}} \left(9 \zeta(1-R_0^{\text{hom}})-(\zeta-1) \right)}}
    \end{aligned}
\end{equation*}

such that if $\phi_6 \in \mathcal{I}$, then the DFE is stable, and, otherwise, it is unstable. 
\end{prop}

We present a conjecture which contains the findings in Proposition \ref{general_M9} for a general network with $M$ nodes, where $M$ is an odd number.

\begin{conj} \label{gen_conjecture}
Let $V=\left\{1,2,3, \hdots, M \right\}$ be the nodes of a network with different fluxes defined by the symmetric flux matrix (Eq. \ref{generalfluxModd}), where $ \phi_{ij} \geq 0$. We assume that the infection rate $\beta_i = \beta_0$ at node $i$ for all $i = 1,2,3, \cdots, M-1$ and $\beta_M = \zeta \beta_0$. We also suppose that all the 
nodes are stable ($R_0^{\text{hom}}=\frac{\beta_0}{\gamma}<1$) except the center node $M$. Then we have a minimum  $\zeta^{\text{crit}}$ given by (Eq. \ref{zeta_critic} or \ref{zeta_critic_star}). If $\frac{1}{R_0^{\text{hom}}}<\zeta< \zeta^{\text{crit}}$, there exists an interval $\mathcal{I}= \left( \max \left\{0, \frac{1}{M}-\Tilde{\phi} \right\}, \min \left\{\frac{1}{M-1}, \frac{1}{M}+\Tilde{\phi} \right\} \right)$, where $\Tilde{\phi}$ is given by (Eq. \ref{interv_critic})),
such that if $\phi_0 \in \mathcal{I}$, then the DFE is stable, and otherwise it is unstable. Finally, if  $\zeta<\frac{1}{R_0^{\text{hom}}}$ i.e. $\beta_M < \gamma$, DFE is always stable.
\end{conj}

To analyze a more general symmetric network, we may consider a flux matrix of the form

\begin{equation} \label{alt_generalflux}
 \phi= \begin{cases}
                                \phi_{i,j+1}=   \phi_j  & \text{ $1 \leq i \leq M-1, ~\text{and}~ i \leq j \leq M-1$} \\
                                                       \phi_{j+1,i}=   \phi_j  & \text{ $1 \leq i \leq M-1, ~\text{and}~ i \leq j \leq M-1$} \\
                                 \phi_{i,i}= 1- \sum_{j=1}^{M-1}\phi_{i,j} &  ~ 1 \leq i \leq M\\
                                                                    0  & \text{ otherwise} \\
   \end{cases}
\end{equation}
where $\phi_{ij} >0, ~\forall~ i , j \in V$. Here we define the fraction of movement from the center node $M$ to the other nodes by $\phi_M=\phi_0$ to be consistent with the other networks. We obtain a particular result for $M=7$  (see appendix for the proof) that we extend as a second conjecture for this class of most flexible networks. 

\begin{conj} \label{alt_gen_conjecture}
Let $V=\left\{1,2,3, \hdots, M \right\}$ be the nodes of networks with the different wights which are defined by the symmetric flux matrix (Eq. \ref{alt_generalflux}), where $ \phi_{ij} \geq 0$. We assume that the infection rate $\beta_i = \beta_0$ at node $i$ for all $i = 1,2,3, \cdots, M-1$ and $\beta_M = \zeta \beta_0$. We also suppose that all the 
nodes are stable i.e. $R_0^{\text{hom}}=\frac{\beta_0}{\gamma}<1$  except the center $M$. Then we have a minimum  $\zeta^{\text{crit}}$ given by (Eq. \ref{zeta_critic} or \ref{zeta_critic_star}). If $\frac{1}{R_0^{\text{hom}}}<\zeta< \zeta^{\text{crit}}$, there exists an interval $\mathcal{I}= \left( \max \left\{0, \frac{1}{M}-\Tilde{\phi} \right\}, \min \left\{\frac{1}{M-1}, \frac{1}{M}+\Tilde{\phi} \right\} \right)$, where $\Tilde{\phi}$ is given by (Eq. \ref{interv_critic})),
such that if $\phi_0 \in \mathcal{I}$, then the DFEis stable, and otherwise it is unstable. Finally, if  $\zeta<\frac{1}{R_0^{\text{hom}}}$ i.e. $\beta_M < \gamma$, then the DFE is always stable.
\end{conj}

\section{Numerical simulations} \label{performence}
We perform numerical simulations to complement our analytical estimates. Our goal is twofold:  (i) to numerically verify the epidemic thresholds obtained from our theoretical estimates and (ii) to explore the impact of the network structures on the  temporal dynamics of epidemic outbreaks. All simulations are performed in MATLAB R2020a using the ode45 (4th/5th order Runge-Kutta-Fehlberg method) function, keeping in mind the order stability area when running simulations by numerically integrating the dynamics of our model. 
 

\subsection*{Epidemic thresholds and network structure} \label{colormaps}
In section 4, we discovered a class of networks with the same epidemic thresholds as the fully connected network. Here we explore the threshold phenomenon using numerical simulations to determine parametric epidemic or non-epidemic regions. More specifically, we investigate when the fluxes and infection rates yielded an epidemic.  Fig. \ref{fig:Pcolor_panel_1} illustrates these results for networks with nine nodes ($M=9$). We numerically determine that an outbreak occurred when $\max_{t \in \tau} I_{\text{tot}}(t) \geq 2$, where $I_{\text{tot}}(t)= \sum_{i=1}^M I_i (t)$ is obtained from the simulations. In Fig. \ref{fig:Pcolor_panel_1} a., we depict star-shaped, star-cycle, and a more general star-background network. For those networks, each  outer node (blue colored) has the same transmission rate $\beta_0$, and the center node (red colored) has a different transmission rate, $\beta_9= \zeta \beta_0$. We then plot colormaps in the $\phi_0 \times \beta_0$ plane to observe epidemic (yellow-colored) and non-epidemic (blue-colored) parametric regions. Here $\phi_0$ represents the flux corresponding to the nodes connecting with the center and $\beta_0$ is the infection rates at the outer nodes. The similarity among the colormaps (including the same critical infection rate $\beta_c$) is consistent with our analytical estimates.  Fig \ref{fig:Pcolor_panel_1} b highlights the difference in the colormaps of cycle networks (uni and bi-directional), indicating that cycles do not have the same epidemic threshold as fully connected or star-shaped networks. A larger epidemic region for cycle networks suggests that a low number of connections may promote outbreaks more easily, likely due to the concentration of infected individuals in the troublesome node (red-colored in the figure). 

\begin{figure}[!htbp]
   \centering
     \begin{adjustwidth}{0.00in}{0in} 
   \includegraphics[width=1.0\textwidth]{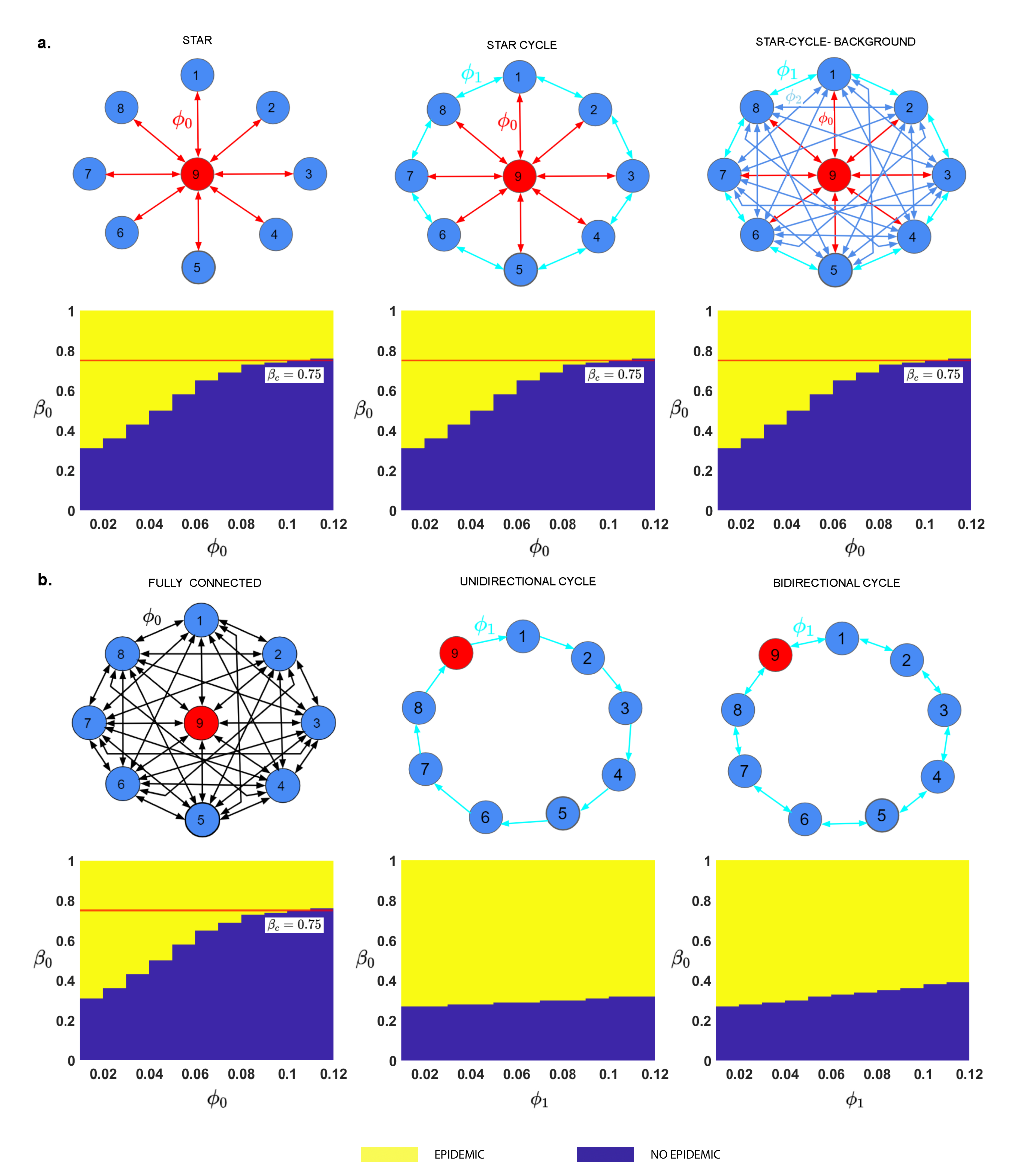}
   \end{adjustwidth}
 \caption{\textbf{Epidemic (yellow colored) \emph{vs.} non-epidemic (blue colored) parametric regions depending on the flux between nodes and infection rates:} Panel (a) shows \emph{star-shaped, star-cycle, and star-cycle-background} networks with their colormaps in the range of the flux $\phi_0=[0.01, 0.12]$ and transmission rate $\beta_0=[0, 1]$, with the step sizes $\Delta \phi_0=0.01$ and $\Delta \beta=0.01$ for the network of size $M=9$. Panel (b) represents the \emph{fully connected, unidirectional, and bidirectional-cycle} networks along with the colormaps using the same information of $\phi_0$, $\beta_0$, and $\phi_1=[0.01, 0.12]$. The area above the threshold (shown by the red line) reveals where the outbreak (yellow parametric region) occurs as predicted by our analytical estimates. For comparison, in all panels we assume $\beta_9 = 4 \beta_0$ and  $\phi_0 \in [0.01, \frac{1}{8}]$, so star-shaped networks have only positive fluxes. For the more general \emph{star-background} network connections, we consider two additional fluxes, $\phi_1=0.08$ and $\phi_2=0.1$ respectively with flux $\phi_0=[0.01, \frac{1}{M-1}]$. For these simulations, we also consider $\gamma=1$, and the initial conditions $I_2(0)=1, I_i(0)=0$ for all $i \in V \setminus 2$; $S_2(0)= 9,999, S_i(0)=10,000$ for all $i \in V \setminus 2$, where $V=\{1,\hdots, 9\}$. We run simulations over the time range $\tau=[0,300]$ with step size $\Delta \tau =0.01$.} 
  \label{fig:Pcolor_panel_1} 
\end{figure}

\subsection*{Epidemic thresholds and network size}

  Cycle and star-class networks differ on the number of node connections, so it is worth asking if the epidemic thresholds change significantly as the network size $M$ increases. He numerically obtained $\beta_0$ threshold values by plotting colormaps as in Fig \ref{fig:Pcolor_panel_1} (see Fig 1 in the appendix for details). Fig \ref{fig:Pcolor_threashold_compare} shows the network geometries and corresponding plots of $M \times \beta_0^{\text{crit}}$  where $M$ is the number of nodes and $\beta_0^{\text{crit}}$ is the minimum infection rate such that epidemics can be controlled for certain values of the flux. We plot the critical infection rates for different multiplicative factors $\zeta$, where $\beta_M = \zeta \beta_0$. When $M$ increases, unidirectional and bidirectional cycles exhibit epidemic thresholds that seem constant with respect to the network size. A possible explanation for this phenomenon might be the little importance of having more nodes in the cycle, given that the epidemic likely occurs in the heterogeneous node and spreads to the neighboring nodes only. On the other hand, star-class networks exhibit a different phenomenon, where the critical infection rate converges to a fixed number. In fact, if our conjecture \ref{gen_conjecture} is true, then $\beta_0^{\text{crit}}$ for the star-class network is given by the expression

$$    \beta_0^{\text{crit}}= \left( \frac{M}{\zeta+M-1} \right)\gamma$$

and thus $ \beta_0^{\text{crit}} \to \gamma $ as $M \to \infty$, which agrees with panel c in Fig \ref{fig:Pcolor_threashold_compare} (in the simulations we consider $\gamma =1$). It is worth noting that such \emph{asymptotic} critical value does not depend on the multiplicative factor $\zeta = \frac{\beta_M}{\beta_0}$. In contrast, for cycle networks $\beta_0^{\text{crit}}$ converges to different values depending on $\zeta$, a feature that is also likely connected with the network geometry.

\begin{figure}[!htbp] 
\centering
  \begin{adjustwidth}{1.25in}{0in}  
   \includegraphics[width=.65\textwidth]{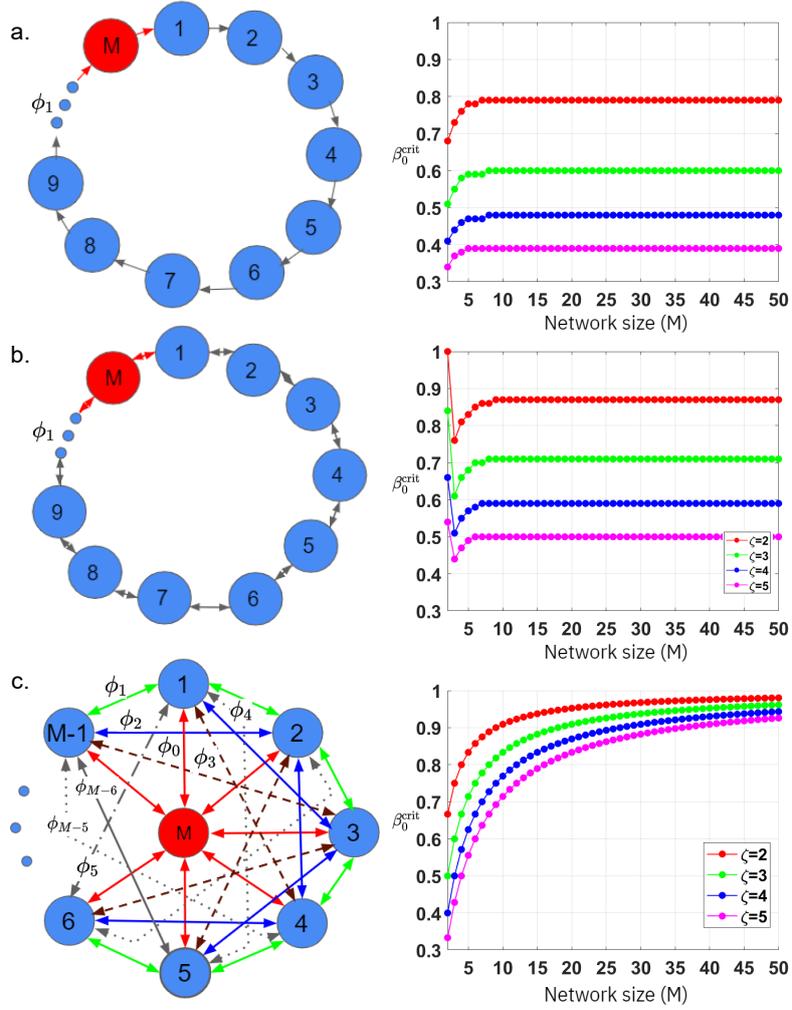}
 \end{adjustwidth}
  \caption{\textbf{Critical infection rates as network size $M$ increases:} We plot $\beta_0^{\text{crit}}$ for $M =1,2,...,50$ and different $\zeta$ values for (a) unidirectional cycles, (b) bidirectional cycles and (c) the general star-class network. Critical values are numerically obtained with colormaps as in Fig \ref{fig:Pcolor_panel_1}  (see Fig \ref{fig:Pcolor_thersholds_compare_with_M_increased} in the Appendix for details). For each simulation, we assume that every node has the same infection rate $\beta_0$, except the center node $M$, where the infection rate is $\beta_M= \zeta \beta_0$.  }
  \label{fig:Pcolor_threashold_compare} 
\end{figure}


\subsection*{Temporal dynamics of epidemic outbreaks}

To better understand the temporal evolution of epidemic outbreaks on different networks, we run numerical simulations for star-shaped, fully connected and cycle (uni and bi-directional) networks and for three scenarios where flux values are low, medium, and high. In Fig. \ref{fig:trajectori_panel_0} illustrates the chosen networks of nine nodes ($M=9$) with stable blue colored nodes ($\beta_0 <\gamma$) and the red colored heterogeneous node ($\beta_9 >\gamma$). We then plot the time series of the infected population for flux values $0.01$, $0.06$, and $0.12$ at a single node $I_2$, the heterogeneous node $I_9$, and for all stable nodes summed (represented by $I_1+I_2+...+I_8$).  The infection dynamics in star-shaped and fully connected networks (Fig. \ref{fig:trajectori_panel_0} a and b) are similar: when the flux increases, the disease also spreads significantly from the center to the other nodes. As a result, the sum of the infected individuals in nodes $1$ to $8$ increases, decreasing the epidemic at node $9$. In contrast, for unidirectional and bidirectional cycles (Fig. \ref{fig:trajectori_panel_0} c and d) the epidemic peak remains high in the heterogeneous node for all three flux values. The outbreak spreads to the other nodes when the flux is high, but the peak still occurs earlier compared to the dynamics at  star-shaped and fully connected networks. The difference between star-shaped and cycle networks is notable for low flux values, as cycle networks exhibit a much more localized outbreak, which is likely due to the lack of connections with the heterogeneous node.

\begin{figure}[!htbp] 
\centering
   \begin{adjustwidth}{-0.40in}{0in} 
     \includegraphics[width=1.10\textwidth]{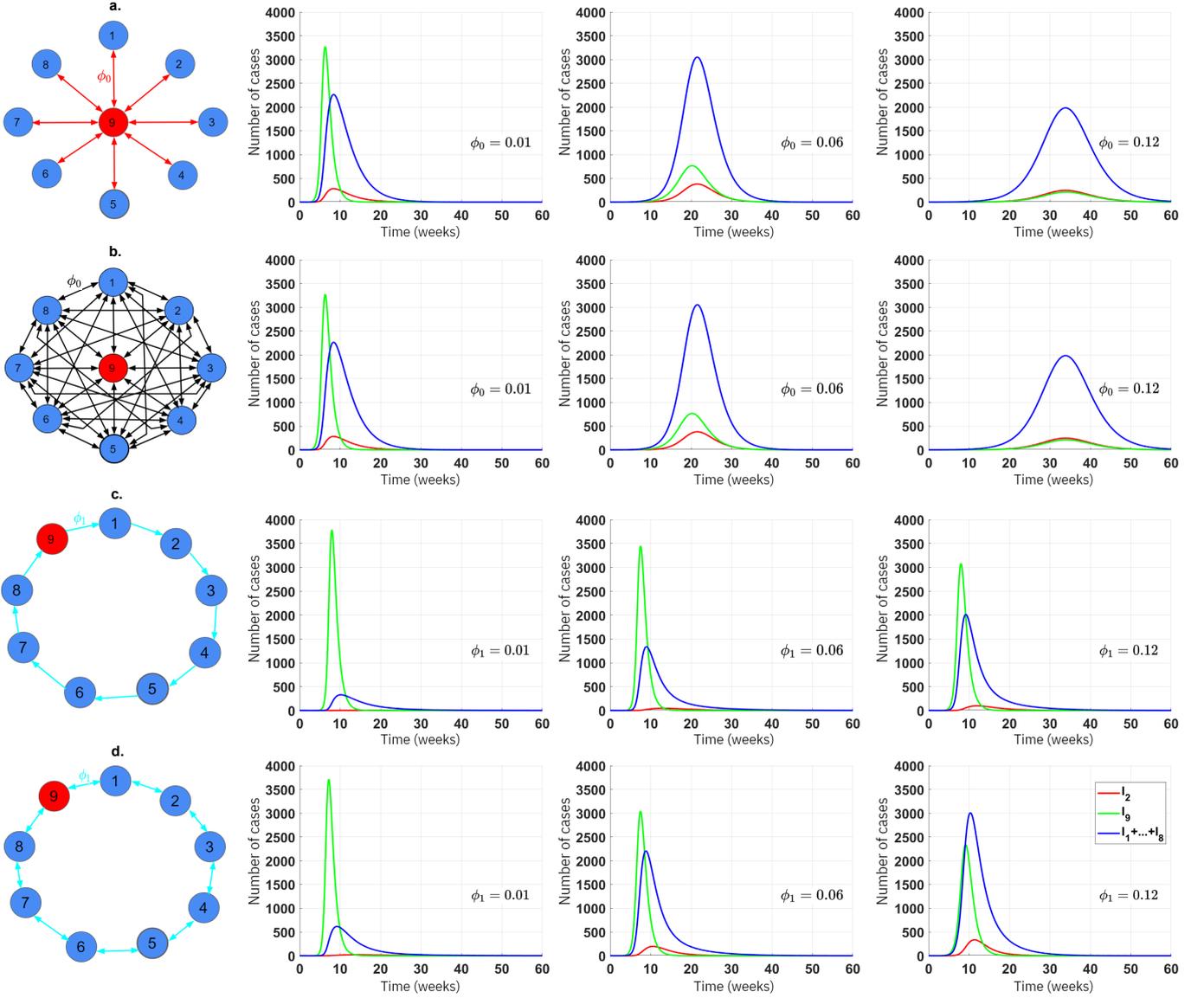}
 \end{adjustwidth}
  \caption{\textbf{Epidemic dynamics for low, medium, and high flux values:} Trajectories of the infected populations for star-shaped, fully connected, and cycle networks, considering three scenarios of low, medium, and high fluxes. Simulated epidemic individuals are obtained for time courses $\tau=[0,60]$ with the step size $\Delta \tau=0.01$, recovery rate $\gamma=1$, and infection rate $ \beta_i= \beta_0=0.95 ~ \forall i \in V \setminus M$, and  $\beta_M= \zeta \beta_0$, where $\zeta=4$. Moreover, the initial conditions are $I_2(0)=1, I_i(0)=0$ for all $i \in V \setminus 2$, and $S_2(0)=9,999, S_i(0)=10,000$ for all $i \in V \setminus 2$. } 
  \label{fig:trajectori_panel_0} 

\end{figure}  

\subsection*{Epidemic peak intensity and timing across networks}

  We further investigate outbreak properties by plotting the maximum peak intensity $I^{max}_{i}$ and  the corresponding peak time $t(I^{max}_i)$ for each node $i$, considering star-shaped, fully connected and cycle networks with nine nodes (Fig \ref{fig:peak_vs_time_node_mInfected}). For low flux values ($\phi_0=\phi_1=0.01$) and for all four network structures, we observe a higher peak at the heterogeneous node ($i=9$) in comparison to the other nodes (green dots). For the peak time $t(I^{max}_i)$, the network structure seemed to play a major role. In fact, for both star-shaped and fully connected networks, $t(I^{max}_i)$ had a similar value of around 8 weeks for all stable nodes ($i=1,2,...,8$) being lower for node 9 (around 6 weeks). In contrast, cycle networks exhibited \emph{delayed} peaks for nodes that are distant from the heterogeneous node 9. For example, nodes 4 and 5 reached their maximum number of infected only 30/40 weeks after the initial time. For higher flux values ($\phi_0=\phi_1=0.06$ and $\phi_0=\phi_1=0.12$), the peak intensities and timing were comparable between stable and heterogeneous nodes for both star-shaped and fully connected networks, thus indicating the role of mobility promoting the spread of the disease. For cycle networks, we again observe delayed and damped peaks for nodes at a higher distance of the heterogeneous node 9 due to the low number of connections.

\begin{figure}[!htbp] 
\centering
    \includegraphics[width=1\textwidth]{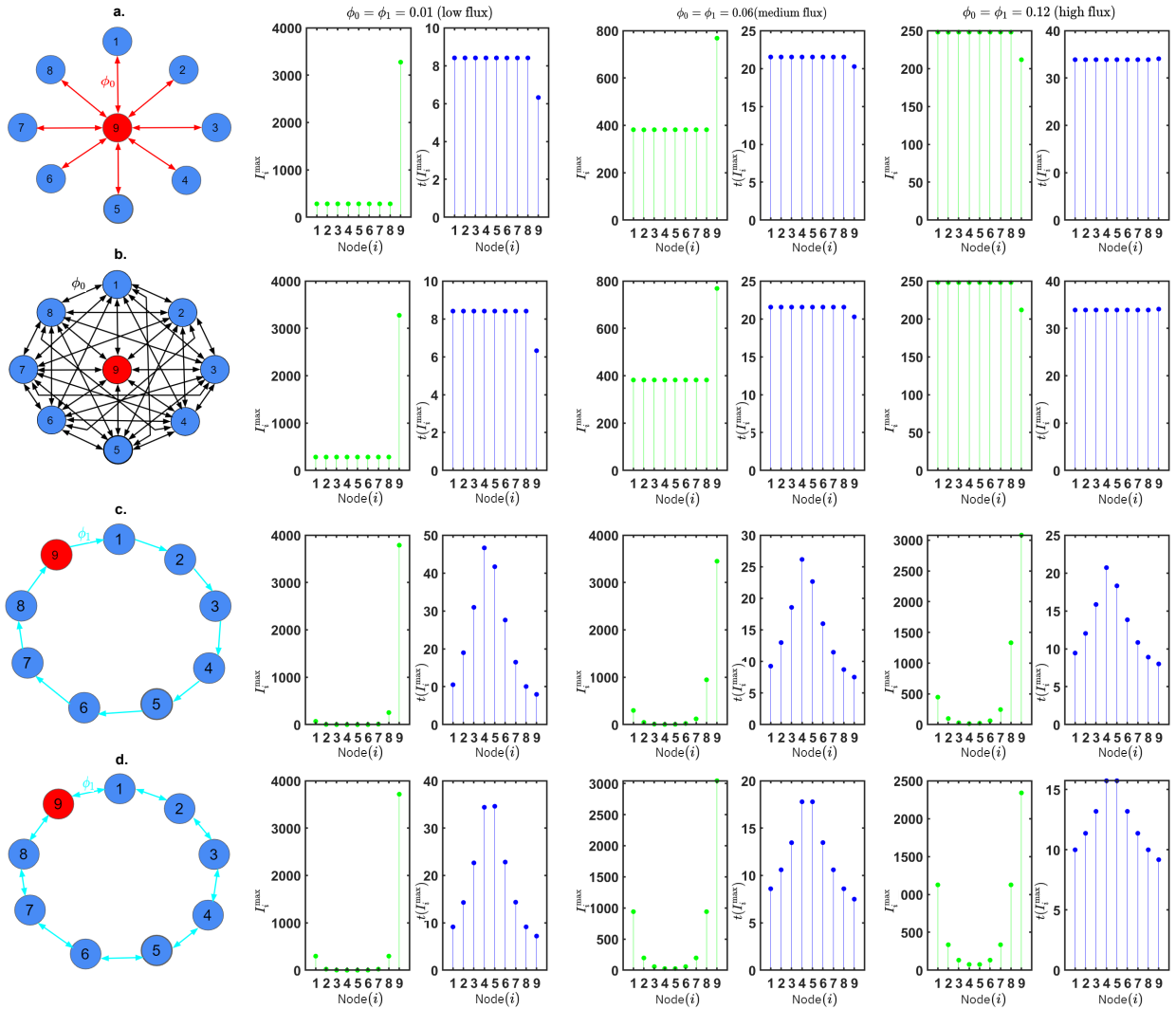}
  \caption{\textbf{The intensity and timing of epidemic peaks:}  the maximum number of infected individuals $I^{max}_{i}$ and peak time $t(I^{max}_{i})$ for each node of star-shaped, fully connected, unidirectional, and bidirectional cycle networks. When the flux is low, the center node significantly has the highest peak in each network, and the epidemic occurs more quickly than the other nodes. However, the disease spreads more rapidly in Star-shaped and fully connected networks compared to cycle networks and shows the same dynamical behavior in Star-shaped and fully connected networks. For cycle networks, the highest peaks still occur at the center node for any fluxes. Results were obtained by using the same parameter values described in Fig. \ref{fig:trajectori_panel_0}.}
  \label{fig:peak_vs_time_node_mInfected} 
\end{figure}


\newpage 

\section{Conclusions}



In this study, we modeled the spread of a disease in a metapopulation, i.e.,  a network of subpopulations (nodes) connected by the flux of individuals (edges).  We investigated epidemic thresholds using the SIR-network model in the case where all subpopulations are equally sized and all nodes have the same infection rate except one (the heterogeneous node).  To determine outbreak conditions, we utilized the classical next-generation matrix approach to find explicit formulas for the intervals in which the flux of people may control epidemics on different network geometries.  Interestingly, we identified a class of networks with the same epidemic thresholds as fully connected networks, previously established in \cite{lucas_dengue_model}.  This class is characterized by a star-shaped structure, where a central heterogeneous node is equally connected to all other nodes.  On the other hand, cycle networks exhibited different thresholds, suggesting that epidemics on poorly connected networks manifest differently while still being influenced by infection rates and mobility factors.


Our numerical simulations revealed distinct epidemic thresholds depending on the network geometry. Fig \ref{fig:Pcolor_panel_1} illustrates how the flux ($\phi_0$) may control or promote outbreaks more prominently in star-shaped networks compared to cycle networks. For cycles, the blue/yellow colormap division was primarily driven by changes in the infection rates $\beta_0$. The epidemic regions for cycles were larger than star-shaped networks, indicating that a highly connected troublesome node may promote and control the disease spread. Regarding temporal dynamics, cycle networks produced more localized outbreaks at low flux values than star-shaped and fully connected networks. At high flux values, epidemics on cycles also peaked faster, mainly in the heterogeneous node. At the same time, star-shaped and fully connected networks could ``spread'' the outbreak to the whole network.


Metapopulation models have played a significant role in exploring  epidemic thresholds and dynamics of infectious diseases on networks. Tocto-Erazo recently explored the contribution of periodic human mobility on the dynamics of Dengue fever in Hermosillo, Mexico \cite{tocto2021effect}. Also in the context of vector-borne diseases, Iggidr and colleagues \cite{iggidr2016dynamics} explored the effects of population circulation on strongly connected networks, estimating $R_0$ and analyzing the stability of both disease-free and endemic equilibrium. Anzo-Hernández et al. \cite{anzo2019risk} conducted a detailed study on the influence of human mobility on $R_0$ for vector-borne diseases, providing explicit formulas and numerical simulations for various network configurations. The impact of vaccination game theoretical approaches \cite{kabir2019evolutionary} and awareness with information spreading \cite{ariful2019impact} has also been explored along with mobility factors and network geometries to provide insights on disease outbreaks and progression. In this work, we adopted a different approach by systematically examining how mobility and geometry shape epidemic thresholds for networks of increased complexity. Unlike previous studies that focused on specific diseases or analyzed reduced-node network structures, we established theoretical estimates that hold for networks with an arbitrary number of nodes, including  star-shaped, star-background, and star-triangle networks. Our formulas provide explicit parametric regions in which the flux of people can either control or promote outbreaks. Additionally, our numerical simulations revealed critical differences in the temporal dynamics of diseases based on the network geometry. 


The insights gained from this pave the way for potential future research directions. Subsequent studies could delve into the role of population heterogeneity by extending our findings to networks with subpopulations of different sizes. For instance, it is reasonable to expect that a central node, such as a city center or transportation hub, can be also highly populated. Investigating the influence of highly populated central nodes could thus provide valuable insights. Furthermore, future studies can relax the assumption of symmetric flux matrices and explore more heterogeneous movement patterns. While all networks considered in this work are theoretical, an intriguing alternative study could involve extracting network features from real-world mobility data.  Finally, it is important to acknowledge that the SIR-network model, originally conceived as a simplified model for dengue fever, can generally describe diseases with direct transmission mechanisms. However, it lacks specific elements such as loss of immunity, exposure periods, and demographic changes. Therefore, an interesting avenue for further exploration would be incorporating those factors in the SIR-network model and examining their impact on epidemic thresholds and disease dynamics.

\clearpage
\newpage

\section*{Acknowledgements}
The authors would like to acknowledge Cort Vanzant (OSU) and Dr. Pedro Maia (UTA)  for their careful reading and feedback on this work. The authors also thank Dr. Anthony Kable (OSU) for the insightful comments, especially for proposing the study of cycle-support networks. LMS thanks Dr. Stefanella Boatto (Federal University of Rio de Janeiro) for initial discussions about the SIR-network model and supportive mentorship during his undergraduate and master's studies.

\bibliography{scibib}
\bibliographystyle{unsrt}

\newpage
\begin{center}
 \section*{Appendix}
\end{center}
\appendix

\section{A few proofs of Propositions and Theorems} 
The purpose of this appendix is to present a few proofs of the Propositions and Theorems that are outlined in the main manuscript.

\textbf{Proposition \ref{controlCycleNeighbourhood}} 
Let $V=\left\{1,2,3 \right\}$ be the node of a \emph{unidirectional cycle} network and the flux matrix $\phi_{3 \times 3}$ be defined by (Eq. \ref{CycleShape}). We consider the infection rates $\beta_i$ from (Eq. \ref{BetaConditionstarshape}), where all the 
nodes are stable except the third one, i.e, $R_0^{\text{hom}} < 1$ and $\frac{\zeta \beta_0}{\gamma}>1$. Then we have a critical value   $\zeta^{\text{crit}}$ given by 

\begin{equation} 
\begin{aligned}
        \zeta^{\text{crit}} =  \frac{1}{R_0^{\text{hom}}} \left( \frac{3 R_0^{\text{hom}}-4}{R_0^{\text{hom}}-2} \right)
\end{aligned}
\end{equation}
such that  if $\frac{1}{R_0^{\text{hom}}}<\zeta< \zeta^{\text{crit}}$, then there exists an interval $\mathcal{I}$ around the point $\frac{1}{2}$ given by 
\begin{equation}  
    \mathcal{I}= \left( \max \left\{0, \frac{1}{2}-\Tilde{\phi} \right\}, \min \left\{1, \frac{1}{2}+\Tilde{\phi} \right\} \right),
\end{equation}
where 
\begin{equation} 
    \begin{aligned} 
       & \tilde{\phi} =\frac{1}{2} \sqrt{ \frac{ \left(\zeta  R_0^{\text{hom}}(R_0^{\text{hom}}-2)+(4-3R_0^{\text{hom}})\right)}{R_0^{\text{hom}} \left(1+2\zeta-3\zeta R_0^{\text{hom}}
        \right)}}
    \end{aligned}
\end{equation}
such that if $\phi_0 \in \mathcal{I}$, the DFE is stable, and otherwise, it is unstable.   

\begin{proof}
To establish the stability criteria for the DFE, we may perform a local stability analysis of the DFE. The Jacobian matrix at the DFE corresponding to the dynamical system of a $3 \times 3$ star-shaped flux matrix is given by the following expression:
\begin{equation} \label{jacob4}
   \begin{aligned}
    J_{3 \times 3}= \begin{pmatrix}
  \beta_1(-1+\phi_0)^2+ \beta_2 \phi_0^2-\gamma   & -\beta_2(-1+\phi_0)\phi_0 &   -\beta_1 \left(-1+\phi_0 \right)\phi_0\\
 - \beta_2(-1+\phi_0)\phi_0 &   \beta_2 \left(-1+\phi_0 \right)^2+\beta_3 \phi_0^2 -\gamma &   -\beta_3(-1+\phi_0)\phi_0\\
-\beta_1(-1+\phi_0)\phi_0 & -\beta_3(-1+\phi_0)\phi_0 &   \beta_3(-1-\phi_0)^2+ \beta_1 \phi_0^2-\gamma 
          \end{pmatrix}.
   \end{aligned}
   \end{equation}
As symmetric matrices have only real eigenvalues, the eigenvalues of the Jacobian matrix (Eq. \ref{jacob4}) are real. For simplicity, we calculate the eigenvalues using the symbolic packages of the software Mathematica, which are given by: 

\begin{equation*}
   \begin{aligned}
      & \lambda_1= (\beta_0-\gamma)-3\beta_0\phi_0(1-\phi_0) \\
         &   \lambda_2=-\gamma+\frac{1}{2} \beta_0 \left(1+\zeta- \phi_0-2 \zeta \phi_0+\phi_0^2 +2 \zeta \phi_0^2+B \right)\\
         &   \lambda_3=-\gamma+\frac{1}{2} \beta_0 \left(1+\zeta- \phi_0-2 \zeta \phi_0+\phi_0^2 +2 \zeta \phi_0^2-B \right),
    \end{aligned}
\end{equation*}
where $ B:=\sqrt{(1-\phi_0+\phi_0^2)^2 +\zeta^2 \left(1-2\phi_0+2\phi_0^2 \right)^2+\zeta \left(-2+6\phi_0-2\phi_0^2-8\phi_0^3+4 \phi_0^4 \right)}.$ As the jacobian matrix $J_{3 \times 3}$  (Eq.\ref{jacob4}) is a real symmetric matrix, the eigenvalues are real. Since $\lambda_2$ and $\lambda_3$ are real, B must be nonnegative for $\phi_0 \in (0,1)$ and for any $\zeta$, $B$ is well-defined and positive real number. We want to find conditions for which all eigenvalues are negative (LAS DFE), or at least one eigenvalue is positive (unstable DFE). We verify that $\lambda_1$ is always negative for all $\phi_0$, since by hypothesis $\frac{\beta_0}{\gamma}<1$. Then we may check $\lambda_2$ and $\lambda_3$, which are more complicated. However, a sufficient condition for stability is the following: since $\lambda_2>\lambda_3$, then we do have stability if the minimum value of $\lambda_2$ is negative. It can be easily verified that $\frac{1}{2}$ is the critical point of the eigenvalues $\lambda_2$, and $\lambda_3$, since  $\frac{\partial \lambda_2}{\partial \phi_0} (\phi_0=\frac{1}{2},\beta_0, \zeta, \gamma)=0$, and $\frac{\partial \lambda_3}{\partial \phi_0} (\phi_0=\frac{1}{2},\beta_0, \zeta, \gamma)=0$.   

We find that 
\begin{equation*}
   \begin{aligned}
  &  \frac{\partial^2 \lambda_2}{\partial \phi_0^2} (\phi_0=\frac{1}{2},\beta_0, \zeta, \gamma) = 1+ 2\zeta + \frac{3-16 \zeta +4 \zeta^2}{\sqrt{9-4\zeta+4\zeta^2}}>0\\
& \frac{\partial^2 \lambda_3}{\partial \phi_0^2} (\phi_0=\frac{1}{2},\beta_0, \zeta, \gamma) = 1+ 2\zeta - \frac{3-16 \zeta +4 \zeta^2}{\sqrt{9-4\zeta+4\zeta^2}}>0\\
   \end{aligned}
\end{equation*}
and so $\lambda_2$, and $\lambda_3$ have reached to the minimum value when $\phi_0=\frac{1}{2}$. Moreover, 
\begin{equation*}
\begin{aligned}
      & \lambda_1 (\phi_0,\beta_0, \zeta, \gamma) <0 ~~\forall ~~ \phi_0 \in (0,1)\\
           & \lambda_2(\phi_0=\frac{1}{2},\beta_0, \zeta, \gamma)=\frac{1}{8} \left(3+2 \zeta +\sqrt{9-4\zeta+4\zeta^2}\right) \beta_0 -\gamma \\
           & \lambda_3(\phi_0=\frac{1}{2},\beta_0, \zeta, \gamma)=\frac{1}{8} \left(3+2 \zeta -\sqrt{9-4\zeta+4\zeta^2}\right) \beta_0 -\gamma. 
   \end{aligned}
\end{equation*}
Since $\lambda_2 > \lambda_3$, so it is enough to find the condition for which $\lambda_2(\phi_0=\frac{1}{2},\beta_0, \zeta, \gamma)=\frac{1}{8} \left(3+2 \zeta +\sqrt{9-4\zeta+4\zeta^2}\right) \beta_0 -\gamma<0$. This implies that  
\begin{equation*}
  \begin{aligned}
         &  \zeta < \frac{\gamma (3\beta_0-4 \gamma)}{\beta_0(\beta_0-2\gamma)}=\frac{1}{R_0^{\text{hom}}} \left(  \frac{3 R_0^{\text{hom}} -4}{R_0^{\text{hom}}-2}  \right):=\zeta^{\text{crit}}.\\
  \end{aligned}
\end{equation*}
In addition, the function $\phi_0 \mapsto \lambda_2 (\phi_0, \zeta, \beta_0, \gamma)$ has roots 
\begin{equation*}
   \begin{aligned}
       &\phi_{c1}= \frac{1}{2} +\frac{1}{2} \frac{ \sqrt{ -\beta_0(3\zeta\beta_0-\gamma-2\zeta\gamma) \left(\zeta\beta_0(\beta_0-2\gamma)+\gamma(-\beta_0+4\gamma)\right)} }{\beta_0 (-3\zeta\beta_0+\gamma+2\zeta\gamma
       )}\\
       &=\frac{1}{2} +\frac{1}{2} \sqrt{ \frac{ \left(\zeta  R_0^{\text{hom}}(R_0^{\text{hom}}-2)+(4-3R_0^{\text{hom}})\right)}{  R_0^{\text{hom}} \left(1+2\zeta-3\zeta R_0^{\text{hom}}
       \right)}}\\
    &\phi_{c2}= \frac{1}{2} -\frac{1}{2} \frac{ \sqrt{ -\beta_0(3\zeta\beta_0-\gamma-2\zeta\gamma) \left(\zeta\beta_0(\beta_0-2\gamma)+\gamma(-\beta_0+4\gamma)\right)} }{\beta_0 (-3\zeta\beta_0+\gamma+2\zeta\gamma
       )}\\
    &=\frac{1}{2} -\frac{1}{2} \sqrt{ \frac{ \left(\zeta  R_0^{\text{hom}}(R_0^{\text{hom}}-2)+(4-3R_0^{\text{hom}})\right)}{  R_0^{\text{hom}} \left(1+2\zeta-3\zeta R_0^{\text{hom}}
       \right)}}\\
    \end{aligned}
\end{equation*}
Notice that $\phi_{c1}$, and $\phi_{c2}$ both are well defined since $\left(1+2\zeta-3\zeta R_0^{\text{hom}}
       \right)$, and $\zeta R_0^{\text{hom}}(R_0^{\text{hom}}-2)+(4-3R_0^{\text{hom}})$ both are negative and $R_0^{\text{hom}}>1$.
Therefore, according to the restriction $\phi_0 \in \left(0,1\right)$, the interval where there will be no epidemic is given by $\mathcal{I}(\zeta,R_0^{\text{hom}})=\left( \max \left\{0, \frac{1}{2}-\Tilde{\phi} \right\}, \min \left\{1, \frac{1}{2}+\Tilde{\phi} \right\} \right)$, where 
\begin{equation*}
   \begin{aligned}
       \Tilde{\phi}= \frac{1}{2} \sqrt{ \frac{ \left(\zeta  R_0^{\text{hom}}(R_0^{\text{hom}}-2)+(4-3R_0^{\text{hom}})\right)}{  R_0^{\text{hom}} \left(1+2\zeta-3\zeta R_0^{\text{hom}}
       \right)}}. 
   \end{aligned}
\end{equation*}
Therefore, if $\phi_0 \in \mathcal{I}$, then the DFE is stable and, otherwise, it is unstable.
\end{proof}

\textbf{Proposition \ref{star_cycle_without_adj}}

Let $V=\left\{1,2,3,4,5 \right\}$ be a node of \emph{cycle-supports} network and let the geometry of its flux matrix be defined by (Eq. \ref{cycleBacknodes}), where $\phi_0 \in (0, \frac{1}{4}]$. The infection rate $\beta_i$ is defined in (Eq. \ref{BetaConditionstarshape}), and we also assume that all the outer nodes are stable ($R_0^{\text{hom}}=\frac{\beta_0}{\gamma}<1$) except the center node $M=5$. Then the critical value of $\zeta$ called $\zeta^{\text{crit}}$ is given by 
\begin{equation*} 
    \zeta^{\text{crit}}= 1+\left( \frac{1-R_0^{\text{hom}}}{R_0^{\text{hom}}} \right)5.
\end{equation*}
Alternatively, the infection threshold of $\beta_0$, $\beta_0^{\text{crit}}$, is given by 
\begin{equation} 
    \beta_0^{\text{crit}}= \gamma\left( \frac{5}{4+\zeta} \right).
\end{equation}
If $\frac{1}{R_0^{\text{hom}}}< \zeta< \zeta^{\text{crit}}$, then there is an interval around the point $\frac{1}{5}$ given by $\mathcal{I}= \left( \max \left\{0, \frac{1}{5}-\Tilde{\phi} \right\}, \min \left\{\frac{1}{4}, \frac{1}{5}+\Tilde{\phi} \right\} \right)$, where 
\begin{equation*}
    \begin{aligned}
        \Tilde{\phi}= \frac{1}{5} \sqrt{ \frac{\gamma (5\gamma-4\beta_0-\zeta\beta_0 )}{\beta_0(\gamma+4\zeta\gamma-5\zeta\beta_0)}}= \frac{1}{5} \sqrt{ \frac{ 5(1-R_0^{\text{hom}})-R_0^{\text{hom}}(\zeta-1) )}{R_0^{\text{hom}} \left(5 \zeta(1-R_0^{\text{hom}})-(\zeta-1) \right)}},
    \end{aligned}
\end{equation*}
such that if $\phi_0 \in \mathcal{I}$, the DFE is stable; and otherwise it is unstable.  

\begin{proof} 
To establish the stability criteria for the DFE in this case, we may perform a local stability analysis of the DFE. The Jacobian matrix at the DFE corresponding to the dynamical system of the $5 \times 5$ cycle support flux matrix (Eq. \ref{cycleBacknodes}) is symmetric. As symmetric matrices have only real eigenvalues, the eigenvalues of the above Jacobian matrix are real. For simplicity, we calculate the eigenvalues using the symbolic packages of the software Mathematica, which are given by 
\begin{equation*}
   \begin{aligned}
      & \lambda_1= \beta_0(1-3\phi_0)^2-\gamma \\
       & \lambda_2= \beta_0 -\gamma-6 \beta_0\phi_0+11\beta_0\phi_0^2-2\sqrt{2}\sqrt{\beta_0^2(1-3\phi_0)^2\phi_0^2} \\
    & \lambda_3= \beta_0 -\gamma-6 \beta_0\phi_0+11\beta_0\phi_0^2+2\sqrt{2}\sqrt{\beta_0^2 (1-3\phi_0)^2\phi_0^2} \\
         &   \lambda_4=-\gamma+\frac{1}{2} \beta_0 \left(1+\zeta-2 \phi_0 -8 \zeta \phi_0+5\phi_0^2 + 20 \zeta \phi_0^2 + B \right)\\
         &   \lambda_5=-\gamma+\frac{1}{2} \beta_0 \left(1+\zeta-2 \phi_0 -8 \zeta \phi_0+5\phi_0^2 + 20 \zeta \phi_0^2 - B \right)
   \end{aligned}
\end{equation*}
where $ B:=\sqrt{(1-2\phi_0+5\phi_0^2)^2 +\zeta^2 \left(1-8\phi_0+20\phi_0^2 \right)^2+2\zeta \left(-1+10\phi_0-9\phi_0^2-80\phi_0^3+ 100 \phi_0^4 \right)}.$ Since all eigenvalues are real, B must be nonnegative for $\phi_0 \in (0, \frac{1}{4}]$, and for any $\zeta$, $B$ is a positive real number. We want to find conditions for which all eigenvalues are negative (LAS DFE), or at least one eigenvalue is positive (unstable DFE). One can see that $\lambda_1$,  $\lambda_2$, and $\lambda_3$ are negative for all $\phi_0 \in \left(0,\frac{1}{4} \right]$, since by hypothesis $\frac{\beta_0}{\gamma}<1$. It is thus sufficient to check the critical points of $\lambda_4$ and $\lambda_5$. However, a sufficient condition for stability is the following: since $\lambda_4>\lambda_5$, it is thus sufficient to check the stability of $\lambda_4$. It can be easily verified that $\frac{1}{5}$ is the critical point of both eigenvalues $\lambda_4$, and $\lambda_5$, since  $\frac{\partial \lambda_4}{\partial \phi_0} (\phi_0=\frac{1}{5},\beta_0, \zeta, \gamma)=0$, and $\frac{\partial \lambda_5}{\partial \phi_0} (\phi_0=\frac{1}{5},\beta_0, \zeta, \gamma)=0$.    
We find that 
\begin{equation*}
   \begin{aligned}
     &  \frac{\partial^2 \lambda_4}{\partial \phi_0^2} (\phi_0=\frac{1}{5},\beta_0, \zeta, \gamma) =  \frac{40 \beta_0 (\zeta-1)^2}{4+\zeta}>0\\
& \frac{\partial^2 \lambda_5}{\partial \phi_0^2} (\phi_0=\frac{1}{5},\beta_0, \zeta, \gamma) = \frac{250 \zeta \beta_0}{4+\zeta}>0\\
   \end{aligned}
\end{equation*}
and so $\lambda_4$ and $\lambda_5$ have reached to the minimum value when $\phi_0=\frac{1}{5}$. Moreover, 
\begin{equation*}
\begin{aligned}
      & \lambda_i (\phi_0,\beta_0, \zeta, \gamma) <0 ~~\forall ~\text{for}~ \phi_0 \in (0,\frac{1}{4}], ~\text{where}, i=1,2,3\\
           & \lambda_4(\phi_0=\frac{1}{5},\beta_0, \zeta, \gamma)=\frac{1}{10} \left(4\beta_0+\zeta\beta_0+\sqrt{(4+\zeta)^2} \beta_0 -10\gamma \right)=\frac{1}{5}\beta_0(4+\zeta)-\gamma \\
           & \lambda_5(\phi_0=\frac{1}{5},\beta_0, \zeta, 
   \end{aligned}
\end{equation*}
However, if $\frac{1}{5}\beta_0(4+\zeta)-\gamma <0$ and that implies  
\begin{equation*}
  \begin{aligned}
  & \beta_0 < \gamma\left( \frac{5}{4+\zeta} \right):=\beta_0^{\text{crit}} ~~\text{and~alternatively}  \\
         &  \zeta < \frac{5 \gamma- 4 \beta_0}{\beta_0}=  \frac{5 - 4 \frac{\beta_0}{\gamma}}{\frac{\beta_0}{\gamma}}= \frac{5-4R_0^{\text{hom}}}{R_0^{\text{hom}}}=1+\left( \frac{1-R_0^{\text{hom}}}{R_0^{\text{hom}}} \right)5 :=\zeta^{\text{crit}}.\\
   \end{aligned}
\end{equation*}
In addition, the function $\phi_0 \mapsto \lambda_4 (\phi_0, \zeta, \beta_0, \gamma)$ has roots 
\begin{equation*}
   \begin{aligned}
       & \phi_{c1}= \frac{1}{5} +\frac{1}{5} \sqrt{ \frac{\gamma (5\gamma-4\beta_0-\zeta\beta_0 )}{\beta_0(\gamma+4\zeta\gamma-5\zeta\beta_0)}}=\frac{1}{5}+ \frac{1}{5} \sqrt{ \frac{ 5(1-R_0^{\text{hom}})-R_0^{\text{hom}}(\zeta-1) )}{R_0^{\text{hom}} \left(5 \zeta(1-R_0^{\text{hom}})-(\zeta-1) \right)}},\\
       & \phi_{c2}= \frac{1}{5}-\frac{1}{5} \sqrt{ \frac{\gamma (5\gamma-4\beta_0-\zeta\beta_0 )}{\beta_0(\gamma+4\zeta\gamma-5\zeta\beta_0)}}=\frac{1}{5}- \frac{1}{5} \sqrt{ \frac{ 5(1-R_0^{\text{hom}})-R_0^{\text{hom}}(\zeta-1) )}{R_0^{\text{hom}} \left(5 \zeta(1-R_0^{\text{hom}})-(\zeta-1) \right)}}.
   \end{aligned}
\end{equation*}
Therefore, taking into account the restriction $\phi_0 \in \left(0,\frac{1}{4}\right]$, the interval where there will be no epidemic is given by $\mathcal{I}(\zeta,R_0^{\text{hom}})=\left( \max \left\{0, \frac{1}{5}-\Tilde{\phi} \right\}, \min \left\{\frac{1}{4}, \frac{1}{5}+\Tilde{\phi} \right\} \right)$, where 

\begin{equation*}
   \begin{aligned}
       \Tilde{\phi}=\frac{1}{5} \sqrt{ \frac{ 5(1-R_0^{\text{hom}})-R_0^{\text{hom}}(\zeta-1) )}{R_0^{\text{hom}} \left(5 \zeta(1-R_0^{\text{hom}})-(\zeta-1) \right)}}. 
   \end{aligned}
\end{equation*}
Therefore, if $\phi_0 \in \mathcal{I}$, then the DFE is stable 
\end{proof}

\newpage
\textbf{Proposition \ref{star_cycle_th}}

Let $V=\left\{1,2,3,4,5 \right\}$ be the nodes of a \emph{star-cycle} network, and the geometry of the corresponding network is defined by the flux (Eq. \ref{StarCycleReturn}), where $M=5, ~\text{and} ~\phi_0 \in (0, \frac{1}{4}]$. The infection rate $\beta_i$ is defined in (Eq. \ref{BetaConditionstarshape}), and we also assume that all the 
nodes are stable ($R_0^{\text{hom}}=\frac{\beta_0}{\gamma}<1$) except the node $M=5$. Then the critical value of $\zeta$ called $\zeta^{\text{crit}}$ is given by 
\begin{equation} 
    \zeta^{\text{crit}}= 1+\left( \frac{1-R_0^{\text{hom}}}{R_0^{\text{hom}}} \right)5.
\end{equation}
Alternatively, the infection threshold of $\beta_0$, $\beta_0^{\text{crit}}$ is given by 
\begin{equation} 
    \beta_0^{\text{crit}}= \gamma\left( \frac{5}{4+\zeta} \right).
\end{equation}
If $\frac{1}{R_0^{\text{hom}}}< \zeta< \zeta^{\text{crit}}$, there is an interval around the point $\frac{1}{5}$ given by $\mathcal{I}= \left( \max \left\{0, \frac{1}{5}-\Tilde{\phi} \right\}, \min \left\{\frac{1}{4}, \frac{1}{5}+\Tilde{\phi} \right\} \right)$, where 
\begin{equation*}
    \begin{aligned}
        \Tilde{\phi}= \frac{1}{5} \sqrt{ \frac{\gamma (5\gamma-4\beta_0-\zeta\beta_0 )}{\beta_0(\gamma+4\zeta\gamma-5\zeta\beta_0)}}= \frac{1}{5} \sqrt{ \frac{ 5(1-R_0^{\text{hom}})-R_0^{\text{hom}}(\zeta-1) )}{R_0^{\text{hom}} \left(5 \zeta(1-R_0^{\text{hom}})-(\zeta-1) \right)}},
    \end{aligned}
\end{equation*}
such that if $\phi_0 \in \mathcal{I}$, the DFE is stable, and otherwise it is unstable.      

\begin{proof}
We perform a local stability analysis to establish the stability criteria for DFE. The Jacobian matrix at DFE corresponding to the dynamical system of the $5 \times 5$ star-cycle flux matrix (Eq. \ref{StarCycleReturn}) is symmetric. Since symmetric matrices have only real eigenvalues, the eigenvalues of the Jacobian matrix are real. For simplicity, we calculate the eigenvalues using the symbolic packages of the software Mathematica, which are given by:

\begin{equation*}
   \begin{aligned}
      & \lambda_1= \beta_0(1-3\phi_0)^2-\gamma \text{~with~multiplicity~2}\\
       & \lambda_2= \beta_0(1-5\phi_0)^2-\gamma \\
         &   \lambda_3=-\gamma+\frac{1}{2} \beta_0 \left(1+\zeta-2 \phi_0 -8 \zeta \phi_0+5\phi_0^2 + 20 \zeta \phi_0^2 + B \right)\\
         &   \lambda_4=-\gamma+\frac{1}{2} \beta_0 \left(1+\zeta-2 \phi_0 -8 \zeta \phi_0+5\phi_0^2 + 20 \zeta \phi_0^2 - B \right)
   \end{aligned}
\end{equation*}
  where, $ B:=\sqrt{(1-2\phi_0+5\phi_0^2)^2 +\zeta^2 \left(1-8\phi_0+20\phi_0^2 \right)^2+2\zeta \left(-1+10\phi_0-9\phi_0^2-80\phi_0^3+ 100 \phi_0^4 \right)}$. 

Since all eigenvalues are real and hence B must be nonnegative for $\phi_0 \in (0, \frac{1}{4}]$ and any $\zeta$, $B$ is a well-defined and positive real number. We want to find conditions for which all eigenvalues are negative (LAS DFE), or at least one eigenvalue is positive (unstable DFE). One can see that both $\lambda_1$, and $\lambda_2$ are negative for all $\phi_0 \in \left(0,\frac{1}{4} \right]$, since by hypothesis $\frac{\beta_0}{\gamma}<1$. It is thus sufficient to check the critical points of the remaining eigenvalues. Then we may check $\lambda_3$ and $\lambda_4$, which are more complicated. However, a sufficient condition for stability is the following: since $\lambda_3>\lambda_4$, we have stability if the minimum value of $\lambda_3$ is negative. It can be easily verified that $\frac{1}{5}$ is the critical point of the eigenvalues $\lambda_3$, and $\lambda_4$, since  $\frac{\partial \lambda_3}{\partial \phi_0} (\phi_0=\frac{1}{5},\beta_0, \zeta, \gamma)=0$, and $\frac{\partial \lambda_4}{\partial \phi_0} (\phi_0=\frac{1}{5},\beta_0, \zeta, \gamma)=0$.    
We find that 
\begin{equation*}
   \begin{aligned}
     &  \frac{\partial^2 \lambda_3}{\partial \phi_0^2} (\phi_0=\frac{1}{5},\beta_0, \zeta, \gamma) =  \frac{40 \beta_0 (\zeta-1)^2}{4+\zeta}>0\\
& \frac{\partial^2 \lambda_4}{\partial \phi_0^2} (\phi_0=\frac{1}{5},\beta_0, \zeta, \gamma) = \frac{250 \zeta \beta_0}{4+\zeta}>0\\
   \end{aligned}
\end{equation*}
and so $\lambda_3$ and $\lambda_4$ have reached to the minimum value when $\phi_0=\frac{1}{5}$. Moreover, 
\begin{equation*}
\begin{aligned}
      & \lambda_1 (\phi_0,\beta_0, \zeta, \gamma) <0 ~~\forall ~~ \phi_0 \in (0,\frac{1}{4}]\\
          & \lambda_2 (\phi_0,\beta_0, \zeta, \gamma) <0 ~~\forall ~~ \phi_0 \in (0,\frac{1}{4}]\\
           & \lambda_3(\phi_0=\frac{1}{5},\beta_0, \zeta, \gamma)=\frac{1}{10} \left(4\beta_0+\zeta\beta_0+\sqrt{(4+\zeta)^2} \beta_0 -10\gamma \right)=\frac{1}{5}\beta_0(4+\zeta)-\gamma \\
           & \lambda_4(\phi_0=\frac{1}{5},\beta_0, \zeta, \gamma)=\frac{1}{10} \left(4\beta_0+\zeta\beta_0-\sqrt{(4+\zeta)^2} \beta_0 -10\gamma \right)=-\gamma <0 
    \end{aligned}
\end{equation*}
However, if $\frac{1}{5}\beta_0(4+\zeta)-\gamma <0$ and that implies  
\begin{equation*}
  \begin{aligned}
  & \beta_0 < \gamma\left( \frac{5}{4+\zeta} \right):=\beta_0^{\text{crit}} ~~\text{and~ alternatively~we ~have,}  \\
         &  \zeta < \frac{5 \gamma- 4 \beta_0}{\beta_0}=  \frac{5 - 4 \frac{\beta_0}{\gamma}}{\frac{\beta_0}{\gamma}}= \frac{5-4R_0^{\text{hom}}}{R_0^{\text{hom}}}=1+\left( \frac{1-R_0^{\text{hom}}}{R_0^{\text{hom}}} \right)5 :=\zeta^{\text{crit}}.\\
  \end{aligned}
\end{equation*}
In addition, the function $\phi_0 \mapsto \lambda_3 (\phi_0, \zeta, \beta_0, \gamma)$ has roots 
\begin{equation*}
   \begin{aligned}
       & \phi_{c1}= \frac{1}{5} +\frac{1}{5} \sqrt{ \frac{\gamma (5\gamma-4\beta_0-\zeta\beta_0 )}{\beta_0(\gamma+4\zeta\gamma-5\zeta\beta_0)}}=\frac{1}{5}+ \frac{1}{5} \sqrt{ \frac{ 5(1-R_0^{\text{hom}})-R_0^{\text{hom}}(\zeta-1) )}{R_0^{\text{hom}} \left(5 \zeta(1-R_0^{\text{hom}})-(\zeta-1) \right)}},\\
       & \phi_{c2}= \frac{1}{5}-\frac{1}{5} \sqrt{ \frac{\gamma (5\gamma-4\beta_0-\zeta\beta_0 )}{\beta_0(\gamma+4\zeta\gamma-5\zeta\beta_0)}}=\frac{1}{5}- \frac{1}{5} \sqrt{ \frac{ 5(1-R_0^{\text{hom}})-R_0^{\text{hom}}(\zeta-1) )}{R_0^{\text{hom}} \left(5 \zeta(1-R_0^{\text{hom}})-(\zeta-1) \right)}}.
   \end{aligned}
\end{equation*}
Therefore, taking into account the restriction $\phi_0 \in \left(0,\frac{1}{4}\right]$, the interval where there will be no epidemic is given by $\mathcal{I}(\zeta,R_0^{\text{hom}})=\left( \max \left\{0, \frac{1}{5}-\Tilde{\phi} \right\}, \min \left\{\frac{1}{4}, \frac{1}{5}+\Tilde{\phi} \right\} \right)$, where 

\begin{equation*}
   \begin{aligned}
       \Tilde{\phi}=\frac{1}{5} \sqrt{ \frac{ 5(1-R_0^{\text{hom}})-R_0^{\text{hom}}(\zeta-1) )}{R_0^{\text{hom}} \left(5 \zeta(1-R_0^{\text{hom}})-(\zeta-1) \right)}}. 
   \end{aligned}
\end{equation*}
Therefore, if $\phi_0 \in \mathcal{I}$, then the DFE is stable, and, otherwise, it is unstable.
\end{proof}

\textbf{Theorem \ref{theoremstartriangle}}

Let $V=\left\{1,2,3, \hdots, M \right\}$ be the nodes of a star-triangle network defined by the flux matrix (\ref{startriangleflux}), where $\phi_0+\phi_1<1, 0<\phi_0 \leq \frac{1}{M-1}$, and $\phi_1>0$. We assume that the infection rate $\beta_i = \beta_0$ at node $i$ for all $i = 1,2,3, \cdots, M-1$ and $\beta_M = \zeta \beta_0$. We also suppose that all the 
nodes are stable $(R_0^{\text{hom}}=\frac{\beta_0}{\gamma}<1)$  except the center node $M$. Then we have a minimum  $\zeta^{\text{crit}}$ given by (Eq. \ref{zeta_critic} or \ref{zeta_critic_star}). If $\frac{1}{R_0^{\text{hom}}}<\zeta< \zeta^{\text{crit}}$, there exists an interval $\mathcal{I}= \left( \max \left\{0, \frac{1}{M}-\Tilde{\phi} \right\}, \min \left\{\frac{1}{M-1}, \frac{1}{M}+\Tilde{\phi} \right\} \right)$, where $\Tilde{\phi}$ is given by (Eq. \ref{interv_critic}),
such that if $\phi_0 \in \mathcal{I}$, then the DFE is stable, and otherwise it is unstable.     

\begin{proof}
We first find the next generation matrix (NGM) $\kappa$ to check the stability by finding the conditions on the parameters at the DFE. As the next generation matrix $\kappa= ( \kappa_{ij})_{M\times M}$, where $ \kappa_{ij}=\frac{1}{\gamma} \sum_{j=1}^M \beta_j\phi_{ij}\phi_{kj} \frac{N_i}{N_j^p}$ is given by (Eq. \ref{NGMFormulya}). Similarly to what we did in Section \ref{The SIR-network model}, we start from the premise that the populations of each node are equal. As part of heterogeneous network models, we consider a different infection rate in the central node with (Eq. \ref{BetaConditionstarshape}). Now for any node $j$ we get $N_j^p= \sum_{k=1}^M\phi_{kj}N_k=(\sum_{k=1}^M\phi_{kj})N=1.N=N$, since $\sum_{k=1}^M\phi_{kj}=1$, and therefore $\frac{N_j}{N_j^p}=\frac{N}{N}=1$.
Thus the NGM (Eq. \ref{NGMFormulya}) becomes (Eq. \ref{NGMFormulyaReduction}), and applying the flux matrix (Eq. \ref{startriangleflux}) we have   
\begin{equation}  \label{NGM3}
\begin{aligned}
&\kappa=\frac{\beta_0}{\gamma} \begin{pmatrix}
p+c_0  &  p+c_2 & \hdots & p & p & q  \\
p+c_2 & p+c_0   & \hdots & p & p & q  \\
 \vdots & \vdots &  \ddots & \vdots &  \vdots\\
                  p & p &  p+c_0 & p+c_2 & p & q\\
                          p & p &  p+c_2 & p+c_0 & p & q\\
   \vdots & \vdots & \vdots & \vdots & \ddots  & q\\
     q  &  q & q & q & q & p+c_1  \\
       \end{pmatrix}, 
      \end{aligned}
   \end{equation}
 
where $p=\zeta \phi_0^2, q=\phi_0(1-\phi_0+\zeta \left(1-(M-1)\phi_0 \right)) $, $c_0=\phi_1^2+(-1+\phi_0+\phi_1)^2$, and $c_1= (M-1)\phi_0^2-\zeta\phi_0^2+\zeta \left(1-(M-1)\phi_0\right)^2$, and $c_2=-2\phi_1 \left(-1+\phi_0+\phi_1 \right)$. 

Since the next generation matrix, $\kappa$ in (Eq. \ref{NGM3} ) is symmetric and positive definite, the eigenvalues of $\kappa$ are all real and positive. Using Collings (22) method of expansion, the characteristic polynomial of $\kappa$ is 
\begin{equation}
 (c_0+c_2-\lambda)^{\frac{M-3}{2}} (c_0-c_2-\lambda)^{\frac{M-1}{2}} \left( \lambda^2-(c_0+c_1+pM)\lambda+\left(c_0c_1+c_1c_2+ c_0p+ (M-1)c_1 p+c_2p +(M-1)(p^2-q^2) \right) \right)=0.  
\end{equation}
For the square matrix $\kappa$, and for the nonzero vector $V$, we have 
\begin{equation}
 \kappa V =  R_0^{\text{hom}}\lambda V.
\end{equation}

Therefore, $R_0^{\text{hom}}\lambda$ is the eigenvalue of the next generation matrix $\kappa$. Therefore, eigenvalues of $\kappa$ are $\frac{\beta_0}{\gamma} (c_0+c_2)=R_0^{\text{hom}}(1-\phi_0)^2$ with multiplicity $\frac{M-3}{2}$, also the eigenvalue of $\kappa$ is $\frac{\beta_0}{\gamma} (c_0-c_2)=R_0^{\text{hom}}(-1+\phi_0+2\phi_1)^2$ with multiplicity $\frac{M-1}{2}$. Finally, the other two eigenvalues are determined from the quadratic term 
\begin{equation}\label{quardraticstartriangle}
   p(\lambda)=\lambda^2-(c_0+c_1+c_2+pM)\lambda+\left(c_1c_0+c_1c_2+c_0p+ p(M-1)c_1 +c_2p+(M-1)(p^2-q^2) \right).
\end{equation}
To determine the stability of the NGM $\kappa$ in  (Eq. \ref{NGM3}), we need to find the conditions on parameters such that all of its eigenvalues are within the unit circle, so that $R_0<1$. The eigenvalues automatically meet the condition 
\begin{equation} \label{eigenvcon1startriangle}
   R_0^{\text{hom}}(1-\phi_0)^2 <1 ~\text{and}~R_0^{\text{hom}}(-1+\phi_0+2\phi_1)^2<1
\end{equation}
since $R_0^{\text{hom}}<1$, and $1-(M-1)\phi_0>0$, and $1-\phi_0-(M-2)\phi_1 >0$ for all $M$. Even the condition (Eq. \ref {eigenvcon1startriangle}) holds when $R_0^{\text{hom}}>1$, if $\phi_0$ near enough to $1$.  

The quadratic (Eq. \ref{quardraticstartriangle}) gives the other two eigenvalues lying within the unit circle, which can be tested by applying the Jury conditions (see Murray \cite{murray2002mathematical}, page $507$). Using the Jury condition, a quadratic equation $P(\lambda)=\lambda^2+a_1\lambda+a_0=0$, these conditions are $P(1)=1+a_1+a_0>0$, $P(-1)=1-a_1+a_0>0$, and $P(0)=a_0<1$. 

As $R_0^{\text{hom}}\lambda:=\lambda'$ is the eigenvalue of $\kappa$,substituting $\lambda=\frac{\lambda'}{R_0^{\text{hom}}}$ in (Eq. \ref{quardraticstartriangle}) we have 
\begin{equation}\label{quardratic2startriangle}
   P(\lambda')=\frac{1}{(R_0^{\text{hom}})^2}[(\lambda')^2-a_1 R_0^{\text{hom}}\lambda'+a_0(R_0^{\text{hom}})^2,
\end{equation}
where $a_0:=\left(c_1c_0+ p(M-1)c_1 +pc_0+(M-1)(p^2-q^2) \right)$, and $a_1:=(c_0+c_1+pM)$.
Applying the Jury conditions to (Eq. \ref{quardratic2startriangle}), we obtain the following inequalities: 
\begin{equation}\label{cond1startriangle}
   \begin{aligned}
         \left\{R_{0}^{\text{hom}} [ \zeta \left( R_{0}^{\text{hom}} (M \phi_0-1)^2-(M-1)\phi_0(M\phi_0-2)-1  \right)  +\phi_0(2-M\phi_0)-1]+1  \right\}>0,
             \end{aligned}
\end{equation}

\begin{equation}
   \begin{aligned}
         \left\{R_{0}^{\text{hom}} [ \zeta \left( R_{0}^{\text{hom}} (M \phi_0-1)^2+(M-1)\phi_0(M\phi_0-2)+1  \right)  -\phi_0(2-M\phi_0)+1]+1  \right\}>0,
             \end{aligned}
\end{equation}

\begin{equation}\label{cond3startriangle}
   \begin{aligned}
        (R_{0}^{\text{hom}})^2 \zeta (M \phi_0-1)^2 <1.
             \end{aligned}
\end{equation}

The Jury conditions imply that (Eqs. \ref{cond1startriangle}-\ref{cond3startriangle}) are identical to (Eqs. A.2-A.4) in Stolerman et al. \cite{lucas_dengue_model}. Therefore the rest of the proof follows.    

\end{proof}

\textbf{Proposition \ref{star_triangle_gen_M_11}}

Let $V=\left\{1,2,3,4, \hdots, 11 \right\}$ be the nodes of a star-triangle network defined by the symmetric flux matrix (Eq. \ref{genstartriangle}) and $M=11$ be an odd number. We assume that the infection rate $\beta_i = \beta_0$ at node $i$ for all $i = 1,2,3, \cdots, 10$ and $\beta_{11} = \zeta \beta_0$. We also suppose that all the nodes are stable $(R_0^{\text{hom}}=\frac{\beta_0}{\gamma}<1)$  except the center node $M$. Then the critical value of $\zeta$ called $\zeta^{\text{crit}}$ is given by 

\begin{equation} 
    \zeta^{\text{crit}}= 1+\left( \frac{1-R_0^{\text{hom}}}{R_0^{\text{hom}}} \right)11. 
\end{equation}
Alternatively, the infection threshold of $\beta_0$, $\beta_0^{\text{crit}}$ is given by 
\begin{equation} 
    \beta_0^{\text{crit}}= \gamma\left( \frac{11}{10+\zeta} \right).
\end{equation}
If $\frac{1}{R_0^{\text{hom}}}<\zeta< \zeta^{\text{crit}}$, there exists an interval $\mathcal{I}= \left( \max \left\{0, \frac{1}{11}-\Tilde{\phi} \right\}, \min \left\{\frac{1}{10}, \frac{1}{11}+\Tilde{\phi} \right\} \right)$, where $\Tilde{\phi}$ is given by (Eq. \ref{interv_critic}),
such that if $\phi_0 \in \mathcal{I}$, then the DFE is stable, and otherwise it is unstable. 

\begin{proof}
To establish the stability criteria for DFE in this case, we may perform a local stability analysis of the DFE. The Jacobian matrix at the DFE corresponding to the dynamical system of the $11 \times 11$ \emph{star-triangle} network and the flux matrix (Eq. \ref{genstartriangle}) is symmetric. Since symmetric matrices have only real eigenvalues, the eigenvalues of the above Jacobian matrix are real. For simplicity, we calculate the eigenvalues using the symbolic packages of the software Mathematica which are given by  

\begin{equation*}
   \begin{aligned}
      & \lambda_1= -\gamma+\beta_0(-1-\phi_6)^2 \text{~with~multiplicity~4}\\
       & \lambda_2=-\gamma+ \beta_0(-1+2\phi_1+\phi_6)^2 \\
           & \lambda_3=-\gamma+ \beta_0(-1+2\phi_2+\phi_6)^2 \\
               & \lambda_4=-\gamma+ \beta_0(-1+2\phi_3+\phi_6)^2 \\
                  & \lambda_5=-\gamma+ \beta_0(-1+2\phi_4+\phi_6)^2 \\
                     & \lambda_6=-\gamma+ \beta_0(-1+2\phi_5+\phi_6)^2 \\
         &   \lambda_7=-\gamma+\frac{1}{2} \beta_0 \left(1+\zeta-2 \phi_6 -20 \zeta \phi_6 + 11\phi_6^2 + 110 \zeta \phi_6^2 + B \right)\\
            &   \lambda_8=-\gamma+\frac{1}{2} \beta_0 \left(1+\zeta-2 \phi_6 -20 \zeta \phi_6 + 11\phi_6^2 + 110 \zeta \phi_6^2 - B \right)\\
   \end{aligned}
\end{equation*}
  where $ B:=\sqrt{(1-2\phi_6+11\phi_6^2)^2 +\zeta^2 \left(1-20\phi_6+110\phi_6^2 \right)^2+2\zeta \left(-1+22\phi_6-81\phi_6^2-440\phi_6^3+ 1210 \phi_6^4 \right)}$. 
  Since all eigenvalues are real and B must be nonnegative for $\phi_0 \in (0, \frac{1}{10}]$ and any $\zeta$, $B$ is a positive real number. We want to find conditions for which all eigenvalues are negative (LAS DFE), or at least one eigenvalue is positive (unstable DFE). Because $\phi_{ij}>0$, we can see that all $\lambda_i$'s are negative for all $ i=1,2, \hdots, 6$. It is thus sufficient to check the critical points of $\lambda_7$ and $\lambda_8$. However, a sufficient condition for their stability is the following: since $\lambda_7>\lambda_8$, it is thus sufficient to check the stability of $\lambda_7$. It can be easily verified that $\frac{1}{11}$ is the critical point of both eigenvalues $\lambda_7$ and $\lambda_8$, since  $\frac{\partial \lambda_7}{\partial \phi_6} (\phi_6=\frac{1}{11},\beta_6, \zeta, \gamma)=0$, and $\frac{\partial \lambda_8}{\partial \phi_6} (\phi_6=\frac{1}{11},\beta_0, \zeta, \gamma)=0$. We can find that 
\begin{equation*}
   \begin{aligned}
 &  \frac{\partial^2 \lambda_7}{\partial \phi_6^2} (\phi_6=\frac{1}{11},\beta_0, \zeta, \gamma) =  \frac{220 \beta_0 (\zeta-1)^2}{10+\zeta}>0\\
& \frac{\partial^2 \lambda_8}{\partial \phi_6^2} (\phi_6=\frac{1}{11},\beta_0, \zeta, \gamma) = \frac{2662 \zeta \beta_0}{10+\zeta}>0\\
   \end{aligned}
\end{equation*}
and so $\lambda_7$ and $\lambda_8$ have reached to their minimum value when $\phi_6=\frac{1}{11}$. Moreover, 
\begin{equation*}
\begin{aligned}
      & \lambda_i (\phi_6,\beta_0, \zeta, \gamma) <0 ~~\text{for ~the ~choice~of} ~~ \phi_i, \text{where} ~ i=1,2,\hdots, 6\\
      & \lambda_7(\phi_6=\frac{1}{11},\beta_0, \zeta, \gamma)=\frac{1}{22} \left(10\beta_0+\zeta\beta_0+\sqrt{(10+\zeta)^2} \beta_0 -22\gamma \right)=\frac{1}{11}\beta_0(10+\zeta)-\gamma \\
           & \lambda_8(\phi_6=\frac{1}{11},\beta_0, \zeta, \gamma)=\frac{1}{10} \left(10\beta_0+\zeta\beta_0-\sqrt{(4+\zeta)^2} \beta_0 -10\gamma \right)=-\gamma <0 
   \end{aligned}
\end{equation*}
However, if $\frac{1}{11}\beta_0(10+\zeta)-\gamma <0$ that implies  
\begin{equation*}
  \begin{aligned}
  & \beta_0 < \gamma\left( \frac{11}{10+\zeta} \right):=\beta_0^{\text{crit}} ~~\text{and}  \\
         &  \zeta < \frac{11 \gamma- 10 \beta_0}{\beta_0}=  \frac{11-10 \frac{\beta_0}{\gamma}}{\frac{\beta_0}{\gamma}}= \frac{11-10R_0^{\text{hom}}}{R_0^{\text{hom}}}=1+\left( \frac{1-R_0^{\text{hom}}}{R_0^{\text{hom}}} \right)11 :=\zeta^{\text{crit}}.\\
   \end{aligned}
\end{equation*}
In addition, the function $\phi_6 \mapsto \lambda_7 (\phi_6, \zeta, \beta_0, \gamma)$ has roots 
\begin{equation*}
   \begin{aligned}
       & \phi_{c1}= \frac{1}{11} +\frac{1}{11} \sqrt{ \frac{\gamma (11\gamma-10\beta_0-\zeta\beta_0 )}{\beta_0(\gamma+10\zeta\gamma-11\zeta\beta_0)}}=\frac{1}{11}+ \frac{1}{11} \sqrt{ \frac{ 11(1-R_0^{\text{hom}})-R_0^{\text{hom}}(\zeta-1) )}{R_0^{\text{hom}} \left(11 \zeta(1-R_0^{\text{hom}})-(\zeta-1) \right)}},\\
       & \phi_{c2}= \frac{1}{11} -\frac{1}{11} \sqrt{ \frac{\gamma (11\gamma-10\beta_0-\zeta\beta_0 )}{\beta_0(\gamma+10\zeta\gamma-11\zeta\beta_0)}}=\frac{1}{11}- \frac{1}{11} \sqrt{ \frac{ 11(1-R_0^{\text{hom}})-R_0^{\text{hom}}(\zeta-1) )}{R_0^{\text{hom}} \left(11 \zeta(1-R_0^{\text{hom}})-(\zeta-1) \right)}}.
   \end{aligned}
\end{equation*}

Therefore, according to the restriction $\phi_6 \in \left(0,\frac{1}{10}\right]$, the interval where there will be no epidemic is given by $\mathcal{I}(\zeta,R_0^{\text{hom}})=\left( \max \left\{0, \frac{1}{11}-\Tilde{\phi} \right\}, \min \left\{\frac{1}{10}, \frac{1}{11}+\Tilde{\phi} \right\} \right)$, where 

\begin{equation*}
   \begin{aligned}
       \Tilde{\phi}=\frac{1}{11} \sqrt{ \frac{ 11(1-R_0^{\text{hom}})-R_0^{\text{hom}}(\zeta-1) )}{R_0^{\text{hom}} \left(5 \zeta(1-R_0^{\text{hom}})-(\zeta-1) \right)}}. 
   \end{aligned}
\end{equation*}
Therefore, if $\phi_6 \in \mathcal{I}$, then the DFE is stable, and otherwise, it is unstable.
\end{proof}

\textbf{Theorem \ref{theoremstarbg}}
Let $V=\left\{1,2,3, \hdots, M \right\}$ be the nodes of a \emph{star-background} networks defined by the flux matrix (Eq. \ref{starshapebackgroundgn}), where $\phi_0+(M-2)\phi_1<1, 0<\phi_0 \leq \frac{1}{M-1}$, and $\phi_1>0$. We assume that the infection rate $\beta_i = \beta_0$ at node $i$ for all $i = 1,2,3, \cdots, M-1$ and $\beta_M = \zeta \beta_0$. We also suppose that all the 
nodes are stable ($R_0^{\text{hom}}=\frac{\beta_0}{\gamma}<1$)  except the center $M$. Then we have a minimum  $\zeta^{\text{crit}}$ given by (Eq. \ref{zeta_critic} or \ref{zeta_critic_star}). If $\frac{1}{R_0^{\text{hom}}}<\zeta< \zeta^{\text{crit}}$, there exists an interval $\mathcal{I}= \left( \max \left\{0, \frac{1}{M}-\Tilde{\phi} \right\}, \min \left\{\frac{1}{M-1}, \frac{1}{M}+\Tilde{\phi} \right\} \right)$, where $\Tilde{\phi}$ is given by (Eq. \ref{interv_critic}),
such that if $\phi_0 \in \mathcal{I}$, then the DFE is stable, and otherwise it is unstable. Finally, if  $\zeta<\frac{1}{R_0^{\text{hom}}}$, i.e., $\beta_M < \gamma$, then the DFE is always stable.     

\begin{proof}
We first find the next generation matrix (NGM) $\kappa$ to check the stability by finding the conditions on the parameters at the DFE. As the next generation matrix $\kappa= ( \kappa_{ij})_{M\times M}$, where $ \kappa_{ij}=\frac{1}{\gamma} \sum_{j=1}^M \beta_j\phi_{ij}\phi_{kj} \frac{N_i}{N_j^p}$ is given by (Eq. \ref{NGMFormulya}). Similar to Section \ref{The SIR-network model}, we start from the premise that the populations of each node are equal. i.e. $N_i=N>0$ for all $i \in V$. As part of heterogeneous network models, we consider a different infection rate in the central node with the (Eq. \ref{BetaConditionstarshape}). 

Now for any node $j$ we get $N_j^p= \sum_{k=1}^M\phi_{kj}N_k=(\sum_{k=1}^M\phi_{kj})N=1.N=N$, since $\sum_{k=1}^M\phi_{kj}=1$, and therefore $\frac{N_j}{N_j^p}=\frac{N}{N}=1$.
Thus the NGM (Eq. \ref{NGMFormulya}) becomes (Eq. \ref{NGMFormulyaReduction}),  and applying the flux matrix (Eq. \ref{starshapebackgroundgn}) we obtain   
\begin{equation} \label{NGM4} 
\begin{aligned}
&\kappa=\frac{\beta_0}{\gamma} \begin{pmatrix}
  p+c_0  &  p & p & \hdots & q  \\
 p & p+c_0   & p & \hdots & q  \\
                  p & p &  p+c_0 & \hdots & q\\
   \vdots & \vdots & \vdots & \ddots  & \vdots\\
     q  &  q & q & \hdots & p+c_1  \\
       \end{pmatrix} 
       =\frac{\beta_0}{\gamma} \begin{pmatrix}
  p  &  p & p & \hdots & q  \\
           p & p & p & \hdots & q\\
               p & p & p & \hdots & q\\
    \vdots & \vdots  & \vdots & \ddots  & \vdots\\
     q  &  q & q & \hdots & p  \\
       \end{pmatrix} + \frac{\beta_0}{\gamma} \begin{pmatrix}
  c_0  \\
 c_0   \\
 c_0 \\
      \vdots \\
     c_1  \\
       \end{pmatrix}
\end{aligned}
   \end{equation}
which is equal to 
\begin{equation}
           \kappa= \frac{\beta_0}{\gamma} (\mathcal{P} + \mathcal{D}) 
\end{equation}
where $p=\zeta \phi_0^2 +\phi_1 (2-2 \phi_0+\phi_1-M \phi_1)$, $q=\phi_0(1-\phi_0+\zeta \left(1-(M-1)\phi_0 \right)) $, $c_0=(1-\phi_0+(M-1)\phi_1)^2$, and $c_1= (m-1)\phi_0^2-\zeta\phi_0^2+\zeta \left(\phi_0(m-1)-1 \right)^2+\phi_1 \left(2\phi_0+(m-1)\phi_1-2 \right)$. We find the entries of the NGM in the following way: 
For $u \in {1,2,3,4,5,...,M-1}$, and $v \in {1,2,3,4,5,...,M-1}$, and $u=v$, the diagonal elements of the NGM $\kappa$ are  
\begin{equation}
\begin{aligned}
   &\kappa_{uu}=   \sum_{j=1}^M \beta_j\phi_{uj}\phi_{uj}= \sum_{j=1}^M \beta_j\phi_{uj}^2\\
       &= \beta_1 \phi_{u1}^2+\beta_2\phi_{u2}^2+ \beta_3 \phi_{u3}^2+ \hdots+\beta_u\phi_{uu}^2 +\hdots +\beta_M\phi_{uM}^2  \\
      & = \beta_0 .\phi_1^2+ \beta_0. \phi_1^2+\beta_0 . \phi_1^2+ \hdots+\beta_0(1-\phi_0-(M-2)\phi_1)^2 +\hdots+ \zeta \beta_0 \phi_0^2 \\ 
       & =\beta_0 \left( \zeta \phi_0^2 + \phi_1^2 (M-2) +(1-\phi_0-(M-2)\phi_1)^2  \right)\\
       & =\beta_0 \left( p+c_0 \right),
 \end{aligned}
\end{equation}
and, if $u=v=M$, 
\begin{equation*}
   \begin{aligned}
       &\kappa_{MM}= \sum_{j=1}^M \beta_j\phi_{Mj}^2\hdots\\
       & = \beta_1\phi_{M1}^2+\beta_2\phi_{M2}^2+ \beta_3 \phi_{M3}^2+...+\beta_M\phi_{MM}^2  \\
      & = \beta_0\phi_0^2+ \beta_0\phi_0^2+\beta_0\phi_0^2+\hdots+ \zeta \beta_0 (1-(M-1)\phi_0)^2  \\ 
   & = \zeta \beta_0 (1-(M-1)\phi_0)^2 +(M-1) \beta_0\phi_0^2 \\ 
   & =(M-1) \beta_0\phi_0^2+\zeta \beta_0 (1-(M-1)\phi_0)^2 \\
    & =\beta_0 \left( M \phi_0^2-\phi_0^2+ \zeta (1-2(M-1)\phi_0+(M-1)^2\phi_0^2 \right) \\ 
 & = \beta_0 \left( M \phi_0^2-\phi_0^2 + \zeta-2 \zeta M\phi_0+2\zeta\phi_0+\zeta M^2\phi_0^2-2\zeta M\phi_0^2+\zeta \phi_0^2 \right)\\ 
 & = \beta_0 \left(\zeta \phi_0^2 + M \phi_0^2-\phi_0^2 + \zeta -2 \zeta M\phi_0 + 2\zeta\phi_0+\zeta M^2\phi_0^2-2\zeta M\phi_0^2 \right)\\ 
  & = \beta_0 \left( \zeta \phi_0^2  +\zeta-2\zeta(M-1)\phi_0+\phi_0^2 \left( M-1-2M\zeta+M^2\zeta \right) \right)\\  
    & = \beta_0 \left( p+ c_1\right).
   \end{aligned}
\end{equation*}
Finally, for any non-diagonal elements of the NGM $\kappa$, when $ u \in {1,2,3,4,5,...,M-1}$, and $ v \in {1,2,3,4,5,...,M-1}$, and $u \not=v$ we have 
\begin{equation*}
\begin{aligned}
        &\kappa_{uv}=   \sum_{j=1}^M \beta_j\phi_{uj}\phi_{vj} \\
       &= \beta_1 \phi_{u1}\phi_{v1}+\beta_2 \phi_{u2}\phi_{v2}+\hdots+ \beta_u \phi_{uu}\phi_{vu}.+\beta_v \phi_{uv}\phi_{vv}...+\beta_M \phi_{uM} \phi_{vM}  \\
         &=\beta_0 \phi_1 (\phi_1) + \beta_0 .\phi_1\phi_1 + \hdots+ 2\beta_0 (1-\phi_0(M-2)\phi_1).\phi_1+\hdots+ \beta_0 \phi_1 \phi_1 + \zeta \beta_0 \phi_0\phi_0 \\
         &= \beta_0 (M-3) \phi_1^2 +2\beta_0 (1-\phi_0(M-2)\phi_1).\phi_1+ \zeta \phi_0^2 \\
                   & =\beta_0 p,
 \end{aligned}
\end{equation*}
and for $u=1,2,3,4, \hdots, M-1$, 
\begin{equation*}
  \begin{aligned}
         &\kappa_{uM}=   \sum_{j=1}^M \beta_j\phi_{uj}\phi_{Mj} \\
       &= \beta_1 \phi_{u1}\phi_{M1}+\beta_2 \phi_{u2}\phi_{M2}+ \beta_3 \phi_{u3}\phi_{M3}+\hdots+\beta_u \phi_{uu} \phi_{Mu} +\hdots+\beta_M \phi_{uM} \phi_{MM}   \\
         &= \beta_0 \phi_1. \phi_0+ \beta_0 \phi_1. \phi_0+\hdots+\beta_0 (1-\phi_0(M-2)\phi_1).\phi_0+...+\zeta \beta_0 \phi_0 (1-(M-1)\phi_0)  \\
         &= \beta_0[ (M-2)\phi_1\phi_0+\beta_0 (1-\phi_0(M-2)\phi_1).\phi_0+\zeta \left(\phi_0- (M-1)\phi_0^2 \right)]  \\
         &= \phi_0(1-\phi_0+\zeta \left(1-(M-1)\phi_0 \right)) \\
         & =\beta_0 q,
  \end{aligned}
\end{equation*}
and also for $v=1,2,3,4, \hdots, M-1$,
\begin{equation*}
   \begin{split}
       &\kappa_{Mv}=   \sum_{j=1}^M \beta_j\phi_{Mj}\phi_{vj} \\
       &= \beta_1 \phi_{M1}\phi_{v1}+\beta_2 \phi_{M2}\phi_{v2}+ \beta_3 \phi_{M3}\phi_{v3}+\hdots+\beta_v \phi_{Mv} \phi_{vv} +\hdots+\beta_M \phi_{MM} \phi_{vM}   \\
         &= \beta_0 \phi_0.0+ \beta_0  \phi_0 . 0+\beta_0  \phi_0 . 0+\hdots+\beta_0 \phi_0 (1-\phi_0)+...+\zeta \beta_0 \phi_0 (1-(M-1)\phi_0)  \\
         &= \beta_0[ \phi_0-\phi_0^2+\zeta \left(\phi_0- (M-1)\phi_0^2 \right)]  \\
         & =\beta_0 q.
   \end{split}
\end{equation*}
Since the next generation matrix $\kappa$ in (Eq. \ref{NGM4}) is symmetric and positive definite, the eigenvalues of $\kappa$ are all real. By doing elementary row or column operations, we see that the rank of the matrix $\mathcal{P}$ is $2$. So all submatrices bigger than $2 \times 2$ have the determinant zero (see, e. g. Friedberg (21)). As $\kappa$ in (Eq. \ref{NGM4}) is the sum of a low-rank matrix and the diagonal matrix $\mathcal{D}$; therefore, we write the characteristic polynomial of $\mathcal{P}+\mathcal{D}$ 

\begin{equation}
 (c_0-\lambda)^{M-2} [ \lambda^2-(c_0+c_1+pM)\lambda+\left(c_1c_0+ p(M-1)c_1 +pc_0+(M-1)(p^2-q^2) \right) ]=0.  
\end{equation}
using the expansion (see Collings \cite{collings} for the eigenvalues $\lambda$.  
For the square matrix, $\kappa=\mathcal{P}+\mathcal{D}$, an  eigenvector $V$ and eigenvalue $\lambda$ make the equation 
\begin{equation}
  \left(\mathcal{P}+\mathcal{D}\right)V= \lambda V
\end{equation}
true, and so for any scalar $R_0^{\text{hom}}$ we have 
\begin{equation}
  R_0^{\text{hom}} \kappa V = R_0^{\text{hom}}  \left(\mathcal{P}+\mathcal{D}\right)V= R_0^{\text{hom}}\lambda V.
\end{equation}
Therefore, $R_0^{\text{hom}}\lambda$ is the eigenvalue of the next generation matrix $\kappa$.  

Therefore, first eigenvalue of $\kappa$ is $\frac{\beta_0}{\gamma}c_0=R_0^{\text{hom}}c_0=R_0^{\text{hom}}(1-\phi_0)^2$ with multiplicity $M-2$, and the other two eigenvalues are determined from the quadratic term 
\begin{equation}\label{quardraticstarbackgr}
   p(\lambda)=\lambda^2-(c_0+c_1+pM)\lambda+\left(c_1c_0+ p(M-1)c_1 +pc_0+(M-1)(p^2-q^2) \right).
\end{equation}
To determine the stability of the NGM $\kappa$, we need to find the conditions on parameters such that all of its eigenvalues are within the unit circle, so that $R_0<1$.
The first eigenvalue 
\begin{equation} \label{eigenvcon1starbackgr}
   R_0^{\text{hom}}(1-\phi_0)^2 <1 
\end{equation}
that automatically meets the condition by hypothesis $R_0^{\text{hom}}<1$, and $\phi_0 \in  (0,1]$ for all $M$ . Even the condition (Eq. \ref{eigenvcon1starbackgr}) holds when $R_0^{\text{hom}}>1$, if $\phi_0$ near enough to $1$. 

The quadratic (Eq. \ref{quardraticstarbackgr}) gives the other two eigenvalues lying within the unit circle, which can be tested by applying the Jury conditions (see \cite{murray2002mathematical}, page $507$). For a quadratic equation $P(\lambda)=\lambda^2+a_1\lambda+a_0=0$, the Jury conditions tells us $P(1)=1+a_1+a_0>0$, $P(-1)=1-a_1+a_0>0$, and $P(0)=a_0<1$. 

As $R_0^{\text{hom}}\lambda:=\lambda'$ is the eigenvalue of $\kappa$, and so substituting $\lambda=\frac{\lambda'}{R_0^{\text{hom}}}$ in (Eq. \ref{quardraticstarbackgr}) we obtain 
\begin{equation}\label{quardratic2starbackgr}
   P(\lambda')=\frac{1}{(R_0^{\text{hom}})^2}[(\lambda')^2-a_1 R_0^{\text{hom}}\lambda'+a_0(R_0^{\text{hom}})^2,
\end{equation}
where $a_0:=\left(c_1c_0+ p(M-1)c_1 +pc_0+(M-1)(p^2-q^2) \right)$, and $a_1:=(c_0+c_1+pM)$. By applying the Jury conditions in (Eq. \ref{quardratic2starbackgr}) reduces to 

\begin{equation} \label{cond1starbackgr}
   \begin{aligned}
  \left\{R_{0}^{\text{hom}} [ \zeta \left( R_{0}^{\text{hom}} (M \phi_0-1)^2-(M-1)\phi_0(M\phi_0-2)-1  \right)  +\phi_0(2-M\phi_0)-1]+1  \right\}>0,
             \end{aligned}
\end{equation}

\begin{equation}
   \begin{aligned}
         \left\{R_{0}^{\text{hom}} [ \zeta \left( R_{0}^{\text{hom}} (M \phi_0-1)^2+(M-1)\phi_0(M\phi_0-2)+1  \right)  -\phi_0(2-M\phi_0)+1]+1  \right\}>0,
             \end{aligned}
\end{equation}

\begin{equation}\label{cond3starbackgr}
   \begin{aligned}
        (R_{0}^{\text{hom}})^2 \zeta (M \phi_0-1)^2 <1.
             \end{aligned}
\end{equation}

The Jury conditions imply (Eqs.\ref{cond1starbackgr}-\ref{cond3starbackgr}) are identical to (Eqs. A.2-A.4) from Stolerman et al. \cite{lucas_dengue_model}, and therefore the rest of the proof follows.

\end{proof}

\newpage

\section{ More flexible network structures of Section \ref{gen_net_models}}

This appendix gives another form of the more flexible network structures of Section \ref{gen_net_models} like the matrix of (Eq. \ref{generalfluxModd}) for a network of an even number of nodes $M$.  Then the flux matrix $\phi_{M \times M}$ is defined by

{
  \centering
  \begin{adjustwidth}{-.5in}{0in} 
\footnotesize
\setlength{\arraycolsep}{2.5pt} 
\medmuskip = 1mu 
\begin{equation}\label{generalfluxMeven}
   \phi_{M \times M}= \begin{pmatrix}
  \phi_d & \phi_1  & \phi_2 & \phi_3 & \hdots & \phi_{\frac{M-2}{2}-1}  & \phi_{\frac{M-2}{2}} & \phi_{\frac{M-2}{2}} & \phi_{\frac{M-2}{2}-1} & \hdots & \phi_4  & \phi_3 & \phi_2 & \phi_1 & \phi_0\\
   
   \phi_1 & \phi_d  & \phi_1 & \phi_2 & \hdots & \phi_{\frac{M-2}{2}-2} & \phi_{\frac{M-2}{2}-1} &  \phi_{\frac{M-2}{2}} & \phi_{\frac{M-2}{2}} & \hdots & \phi_5 & \phi_4 & \phi_3 & \phi_2 & \phi_0 \\
    
 \phi_2 & \phi_1 & \phi_d  & \phi_1 & \hdots & \phi_{\frac{M-2}{2}-3} & \phi_{\frac{M-2}{2}-2} & \phi_{\frac{M-2}{2}-1} &  \phi_{\frac{M-2}{2}}  & \hdots  & \phi_6 & \phi_5 & \phi_4 & \phi_3 & \phi_0 \\
 
 \phi_3 & \phi_2 & \phi_1 & \phi_d  & \hdots & \phi_{\frac{M-2}{2}-4} & \phi_{\frac{M-2}{2}-3} &  \phi_{\frac{M-2}{2}-2} & \phi_{\frac{M-2}{2}-1}  & \hdots & \phi_7 & \phi_6 & \phi_5 & \phi_4 & \phi_0 \\

\vdots & \vdots & \vdots & \vdots & \ddots & \vdots & \vdots & \vdots  & \vdots & \ddots &  \vdots & \vdots & \vdots &  \vdots  & \vdots \\   

\phi_{\frac{M-2}{2}-1} & \phi_{\frac{M-2}{2}-2} & \phi_{\frac{M-2}{2}-3} & \phi_{\frac{M-2}{2}-4}  & \hdots & \phi_{d} & \phi_{1} & \phi_{2} &  \phi_{3}  & \hdots &  \phi_{\frac{M-2}{2}-2} & \phi_{\frac{M-2}{2}-1} & \phi_{\frac{M-2}{2}} &  \phi_{\frac{M-2}{2}} & \phi_0 \\

\phi_{\frac{M-2}{2}} & \phi_{\frac{M-2}{2}-1} & \phi_{\frac{M-2}{2}-2} & \phi_{\frac{M-2}{2}-3}  & \hdots & \phi_{1} & \phi_{d} &  \phi_{1} & \phi_{2}  & \hdots & \phi_{\frac{M-2}{2}-3}  & \phi_{\frac{M-2}{2}-2} & \phi_{\frac{M-2}{2}-1} & \phi_{\frac{M-2}{2}} & \phi_0 \\

\phi_{\frac{M-2}{2}} & \phi_{\frac{M-2}{2}} & \phi_{\frac{M-2}{2}-1} & \phi_{\frac{M-2}{2}-2}  & \hdots & \phi_{2} & \phi_{1} &  \phi_{d} & \phi_{1}  & \hdots & \phi_{\frac{M-2}{2}-4}  & \phi_{\frac{M-2}{2}-3} & \phi_{\frac{M-2}{2}-2} & \phi_{\frac{M-2}{2}-1} & \phi_0 \\

\phi_{\frac{M-2}{2}-1} & \phi_{\frac{M-2}{2}} & \phi_{\frac{M-2}{2}} & \phi_{\frac{M-2}{2}-1}  & \hdots & \phi_{3} & \phi_{2} &  \phi_{1} & \phi_{d}  & \hdots & \phi_{\frac{M-2}{2}-5}  & \phi_{\frac{M-2}{2}-4} & \phi_{\frac{M-2}{2}-3}& \phi_{\frac{M-2}{2}-2} & \phi_0 \\

\vdots & \vdots & \vdots & \vdots & \ddots & \vdots & \vdots & \vdots  & \vdots & \ddots &  \vdots & \vdots & \vdots &  \vdots  & \vdots \\   

\phi_4 & \phi_5  & \phi_6 & \phi_7 & \hdots & \phi_{\frac{M-2}{2}-2} & \phi_{\frac{M-2}{2}-3} &  \phi_{\frac{M-2}{2}-4} & \phi_{\frac{M-2}{2}-5} & \hdots & \phi_d & \phi_1 & \phi_2 & \phi_3 & \phi_0 \\
 
\phi_3 & \phi_4  & \phi_5 & \phi_6 & \hdots & \phi_{\frac{M-2}{2}-1} & \phi_{\frac{M-2}{2}-2} &  \phi_{\frac{M-2}{2}-3} & \phi_{\frac{M-2}{2}-4} & \hdots & \phi_1 & \phi_d & \phi_1 & \phi_2 & \phi_0 \\

\phi_2 & \phi_3  & \phi_4 & \phi_5 & \hdots & \phi_{\frac{M-2}{2}} & \phi_{\frac{M-2}{2}-1} &  \phi_{\frac{M-2}{2}-2} & \phi_{\frac{M-2}{2}-3} & \hdots & \phi_2 & \phi_1 & \phi_d & \phi_1 & \phi_0 \\
  
\phi_1 & \phi_2  & \phi_3 & \phi_4 & \hdots & \phi_{\frac{M-2}{2}} & \phi_{\frac{M-2}{2}} &  \phi_{\frac{M-2}{2}-1} & \phi_{\frac{M-2}{2}-2} & \hdots & \phi_3 & \phi_2 & \phi_1 & \phi_d & \phi_0 \\

\phi_0 & \phi_0 & \phi_0 & \phi_0 & \hdots & \phi_0 & \phi_0 & \phi_0 & \phi_0 & \hdots & \phi_0 & \phi_0 & \phi_0 & \phi_0  & \phi_{MM}
\end{pmatrix},
\end{equation}
 \end{adjustwidth}
}

where $\phi_d:=1-2(\phi_0+\phi_1+ \hdots + \phi_{\frac{M-2}{2}-1}+\phi_{\frac{M-2}{2}})-\phi_0$ and $\phi_{MM}=1-(M-1)\phi_0$.

Since the general network is complicated, we first find the \emph{star-triangle} (Figure \ref{fig:diff_networks}) networks by extending the \emph{star-shaped} (Figure \ref{fig:diff_networks}) networks model prescribed in Section \ref{StarShapeSection}. Later, in the arbitrary \emph{star-shaped} (Figure \ref{fig:diff_networks}) networks, we connect the rest of the nodes by a different flux ($\phi_1$). This process results in the \emph{star-background} (Figure \ref{fig:diff_networks}) networks, a more general class of networks compared to the \emph{fully connected} (Figure \ref{fig:diff_networks}) heterogeneous networks (Theorem \ref{theoremfullcon}) and \emph{star-shaped} networks (Theorem \ref{theoremstar}).  Our results show that the \emph{star-triangle} (Figure \ref{fig:diff_networks}) and the \emph{star-background} (Figure \ref{fig:diff_networks}) networks belong to the family of networks in which the same critical value $\zeta$ (Eq. \ref{zeta_critic}) and interval $\mathcal{I}$ (Eq. \ref{interval}) emerge for the heterogeneous node $M$. The current section will provide a \emph{more general network} (Figure \ref{fig:StarToGeneral}) model for controlling epidemic outbreaks. In particular,  a \emph{star-background} (Figure \ref{fig:diff_networks}) network mimics a city where a central node receives a daily flux of people. Finally, we do a conjecture for the \emph{more general network} (Figure \ref{fig:StarToGeneral}) model and justify it for node $M=9$ in the proposition \ref{general_M9}. The proof is given in the following.

\textbf{Proposition \ref{general_M9}}

Let $V=\left\{1,2, \hdots, 9\right\}$ be the nodes of a \emph{general networks} defined by the symmetric flux matrix (Eq. \ref{generalfluxModd}). We assume that the infection rate $\beta_i = \beta_0$ at node $i$ for all $i = 1,2,3, \hdots, 9$ and $\beta_9 = \zeta \beta_0$. We also suppose that all the nodes are stable i.e. $R_0^{\text{hom}}=\frac{\beta_0}{\gamma}<1$  except the center node $M$. Then the critical value of $\zeta$ called $\zeta^{\text{crit}}$ is given by 
\begin{equation*} 
    \zeta^{\text{crit}}= 1+\left( \frac{1-R_0^{\text{hom}}}{R_0^{\text{hom}}} \right)9.
\end{equation*}
Alternatively, the infection threshold of $\beta_0$, $\beta_0^{\text{crit}}$ is given by 
\begin{equation*} 
    \beta_0^{\text{crit}}= \gamma\left( \frac{9}{8+\zeta} \right).
\end{equation*}
If $\frac{1}{R_0^{\text{hom}}}<\zeta< \zeta^{\text{crit}}$, and taking into account the restriction $\phi_{ij}$, where $i, j = 1,2,3, \hdots, 9$, then the interval where there will be no epidemic is given by $\mathcal{I}(\zeta,R_0^{\text{hom}})=\left( \max \left\{0, \frac{1}{9}-\Tilde{\phi} \right\}, \min \left\{\frac{1}{8}, \frac{1}{9}+\Tilde{\phi} \right\} \right)$, where $\Tilde{\phi}$ is given by 

\begin{equation*}
    \begin{aligned}
        \Tilde{\phi}=\frac{1}{9} \sqrt{ \frac{ 9(1-R_0^{\text{hom}})-R_0^{\text{hom}}(\zeta-1) )}{R_0^{\text{hom}} \left(9 \zeta(1-R_0^{\text{hom}})-(\zeta-1) \right)}}
    \end{aligned}
\end{equation*}

such that if $\phi_6 \in \mathcal{I}$, then the DFE is stable, and, otherwise, it is unstable. 

\begin{proof}
To establish the stability criteria for DFE, we may perform a local stability analysis of DFE. The Jacobian matrix at DFE corresponding to the dynamical system of the general network of the flux matrix (Eq. \ref{generalfluxModd}) is symmetric.  As Symmetric matrices have only real eigenvalues, the eigenvalues of the above Jacobian matrix are real. For simplicity, we calculate the eigenvalues using the symbolic packages of the software Mathematica which are given by 

\begin{equation*}
   \begin{aligned}
      & \lambda_1= -\gamma+\beta_0(-1+\phi_0+2\phi_1+4\phi_2+2\phi_3)^2 \\
        & \lambda_2= -\gamma+\beta_0(-1+\phi_0+2\phi_1+4\phi_2+2\phi_3)^2 \\
        & \lambda_3= -\gamma+\beta_0(-1+\phi_0+4\phi_1+4\phi_3)^2 \\
   & \lambda_4=-\gamma-2 \sqrt{2} \sqrt{\beta_0^2(\phi_1-\phi_3)^2(-1+\phi_0+2\phi_1+2\phi_2+2\phi_3+2\phi_4)^2}+ \beta_0(1+\phi_0^2   +6 \phi_1^2-4\phi_2+4 \phi_2^2-4\phi_3+8 \phi_2 \phi_3\\
   & +6\phi_3^2-4\phi_4+8\phi_2\phi_4+8\phi_3\phi_4+4\phi_4^2+4\phi_1(-1+2\phi_2+\phi_3+2\phi_4)+\phi_0(-2+4\phi_1+4\phi_2+4 \phi_3+4\phi_4) \\
& \lambda_5=-\gamma+2 \sqrt{2} \sqrt{\beta_0^2(\phi_1-\phi_3)^2(-1+\phi_0+2\phi_1+2\phi_2+2\phi_3+2\phi_4)^2}+ \beta_0(1+\phi_0^2+6 \phi_1^2-4\phi_2+4 \phi_2^2-4\phi_3+8 \phi_2 \phi_3\\
   & +6\phi_3^2-4\phi_4+8\phi_2\phi_4+8\phi3\phi_4+4\phi_4^2+4\phi_1(-1+2\phi_2+\phi_3+2\phi_4)+\phi_0(-2+4\phi_1+4\phi_2+4 \phi_3+4\phi_4) \\
&   \lambda_6=-\gamma+\frac{1}{2} \beta_0 \left(1+\zeta-2 \phi_0 -16 \zeta \phi_0 + 9\phi_0^2 + 72 \zeta \phi_0^2 + B \right)\\
&   \lambda_7=-\gamma+\frac{1}{2} \beta_0 \left(1+\zeta-2 \phi_0 -16 \zeta \phi_0 + 9\phi_0^2 + 72 \zeta \phi_0^2 - B \right)\\
   \end{aligned}
\end{equation*}
  where $ B:=\sqrt{(1-2\phi_0+9\phi_0^2)^2 +\zeta^2 \left(1-16\phi_0+72\phi_0^2 \right)^2+2\zeta \left(-1+18\phi_0-49\phi_0^2-288\phi_0^3+ 648 \phi_0^4 \right)}$. Here $\lambda_4$ and $\lambda_5$ both have multiplicity 2. Furthermore, since all eigenvalues are real and hence B must be nonnegative for the choice of $\phi_{ij}$ and for any $\zeta$, $B$ is well defined and positive real number. We want to find conditions for which all eigenvalues are negative (LAS DFE) or at least one eigenvalue is positive (unstable DFE). We can see that all $\lambda_i$'s are negative for all $ i=1,2, \hdots,5$ because $\phi_{ij}>0$. It is thus sufficient to check the eigenvalues $\lambda_6$ and $\lambda_7$, which are more complicated. However, a sufficient condition for their stability is the following: since $\lambda_6>\lambda_7$, it is thus sufficient to check the stability of $\lambda_6$. It can be easily verified that $\frac{1}{9}$ is the critical point of both eigenvalues $\lambda_6$ and $\lambda_7$, since  $\frac{\partial \lambda_6}{\partial \phi_6} (\phi_0=\frac{1}{9},\beta_0, \zeta, \gamma)=0$, and $\frac{\partial \lambda_7}{\partial \phi_6} (\phi_0=\frac{1}{9},\beta_0, \zeta, \gamma)=0$. We can find that 
\begin{equation*}
   \begin{aligned}
 &  \frac{\partial^2 \lambda_6}{\partial \phi_0^2} (\phi_0=\frac{1}{9},\beta_0, \zeta, \gamma) =  \frac{144 \beta_0 (\zeta-1)^2}{8+\zeta}>0\\
& \frac{\partial^2 \lambda_7}{\partial \phi_0^2} (\phi_0=\frac{1}{9},\beta_0, \zeta, \gamma) = \frac{1458 \zeta \beta_0}{8+\zeta}>0\\
   \end{aligned}
\end{equation*}
and so $\lambda_6$, and $\lambda_7$ have reached to their minimum value when $\phi_0=\frac{1}{9}$. Moreover, 
\begin{equation*}
\begin{aligned}
      & \lambda_i (\phi_0,\beta_0, \zeta, \gamma) <0 ~~\text{for ~the ~choice~of} ~~ \phi_i, \text{where} ~ i=1,2,\hdots, 5\\
      & \lambda_6(\phi_0=\frac{1}{9},\beta_0, \zeta, \gamma)=\frac{1}{18} \left(8\beta_0+\zeta\beta_0+\sqrt{(8+\zeta)^2} \beta_0 -18\gamma \right)=\frac{1}{9}\beta_0(8+\zeta)-\gamma \\
           & \lambda_7(\phi_0=\frac{1}{9},\beta_0, \zeta, \gamma)=\frac{1}{6} \left(6\beta_0+\zeta\beta_0-\sqrt{(6+\zeta)^2} \beta_0 -14\gamma \right)=-\gamma <0 
    \end{aligned}
\end{equation*}
However, if $\frac{1}{9}\beta_0(8+\zeta)-\gamma <0$ that implies  $ \beta_0 < \gamma\left( \frac{9}{8+\zeta} \right):=\beta_0^{\text{crit}}$. Alternatively, we have $\zeta < \frac{9 \gamma- 8 \beta_0}{\beta_0}=  \frac{9-8 \frac{\beta_0}{\gamma}}{\frac{\beta_0}{\gamma}}= \frac{9-8R_0^{\text{hom}}}{R_0^{\text{hom}}}=1+\left( \frac{1-R_0^{\text{hom}}}{R_0^{\text{hom}}} \right)9 :=\zeta^{\text{crit}}.$
In addition, the function $\phi_0 \mapsto \lambda_6 (\phi_0, \zeta, \beta_0, \gamma)$ has roots 
\begin{equation*}
   \begin{aligned}
       & \phi_{c1}= \frac{1}{9} +\frac{1}{9} \sqrt{ \frac{\gamma (9\gamma-8\beta_0-\zeta\beta_0 )}{\beta_0(\gamma+8\zeta\gamma-9\zeta\beta_0)}}=\frac{1}{9}+ \frac{1}{9} \sqrt{ \frac{ 9(1-R_0^{\text{hom}})-R_0^{\text{hom}}(\zeta-1) )}{R_0^{\text{hom}} \left(9 \zeta(1-R_0^{\text{hom}})-(\zeta-1) \right)}},\\
       & \phi_{c2}= \frac{1}{9} -\frac{1}{9} \sqrt{ \frac{\gamma (9\gamma-8\beta_0-\zeta\beta_0 )}{\beta_0(\gamma+8\zeta\gamma-9\zeta\beta_0)}}=\frac{1}{9}- \frac{1}{9} \sqrt{ \frac{ 9(1-R_0^{\text{hom}})-R_0^{\text{hom}}(\zeta-1) )}{R_0^{\text{hom}} \left(9 \zeta(1-R_0^{\text{hom}})-(\zeta-1) \right)}}.
   \end{aligned}
\end{equation*}
Therefore, taking into account the restriction $\phi_0 \in \left(0,\frac{1}{8}\right]$, the interval where there will be no epidemic is given by $\mathcal{I}(\zeta,R_0^{\text{hom}})=\left( \max \left\{0, \frac{1}{9}-\Tilde{\phi} \right\}, \min \left\{\frac{1}{8}, \frac{1}{9}+\Tilde{\phi} \right\} \right)$, where 

\begin{equation*}
   \begin{aligned}
       \Tilde{\phi}=\frac{1}{9} \sqrt{ \frac{ 9(1-R_0^{\text{hom}})-R_0^{\text{hom}}(\zeta-1) )}{R_0^{\text{hom}} \left(9 \zeta(1-R_0^{\text{hom}})-(\zeta-1) \right)}}. 
   \end{aligned}
\end{equation*}
Therefore, if $\phi_6 \in \mathcal{I}$, then the DFE is stable and otherwise it is unstable.
\end{proof}

 Proposition \ref{general_M9} holds for other networks, for example, Theorems on arbitrary size \emph{Star-shaped}, \emph{Fully connected}, \emph{Star-background} networks, and other networks \emph{star-cycle}, \emph{star-support} (Figure \ref{fig:StarToGeneral}) networks also follow by the general flux matrix (Eq. \ref{generalfluxModd}) or (Eq. \ref{generalfluxMeven}). Therefore, we present a conjecture in the following, which contains our findings in Proposition \ref{general_M9}, and we will also justify it by simulations.  


\textbf{Proposition 7} 
Let $V=\left\{1,2, \hdots, 7\right\}$ be the nodes of a \emph{general network} defined by the symmetric flux matrix (Eq. \ref{alt_generalflux}). We assume that the infection rate $\beta_i = \beta_0$ at node $i$ for all $i = 1,2,3, \hdots, 7$ and $\beta_7 = \zeta \beta_0$. We also suppose that all the nodes are stable ($R_0^{\text{hom}}=\frac{\beta_0}{\gamma}<1$)  except the center node $M$. Then the critical value of $\zeta$ called $\zeta^{\text{crit}}$ is given by 
\begin{equation} 
  \zeta^{\text{crit}}= 1+\left( \frac{1-R_0^{\text{hom}}}{R_0^{\text{hom}}} \right)7.
\end{equation}
Alternatively, the infection threshold of $\beta_0$, $\beta_0^{\text{crit}}$ is given by 
\begin{equation} 
  \beta_0^{\text{crit}}= \gamma\left( \frac{7}{6+\zeta} \right).
\end{equation}
If $\frac{1}{R_0^{\text{hom}}}<\zeta< \zeta^{\text{crit}}$, and taking into account the restriction $\phi_{ij}$, where $i, j = 1,\hdots, 7$, then the interval where there will be no epidemic is given by $\mathcal{I}(\zeta,R_0^{\text{hom}})=\left( \max \left\{0, \frac{1}{7}-\Tilde{\phi} \right\}, \min \left\{\frac{1}{6}, \frac{1}{7}+\Tilde{\phi} \right\} \right)$, where $\Tilde{\phi}$ is given by 

\begin{equation*}
  \begin{aligned}
      \Tilde{\phi}=\frac{1}{7} \sqrt{ \frac{ 7(1-R_0^{\text{hom}})-R_0^{\text{hom}}(\zeta-1) )}{R_0^{\text{hom}} \left(7 \zeta(1-R_0^{\text{hom}})-(\zeta-1) \right)}}
  \end{aligned}
\end{equation*}

such that if $\phi_6 \in \mathcal{I}$, then the DFE is stable, and otherwise it is unstable. 

\begin{proof}
To establish the stability criteria for the DFE in this particular case, we may perform a local stability analysis of the DFE. We see that the Jacobian matrix at the DFE corresponding to the dynamical system of this network's flux matrix (Eq. \ref{alt_generalflux}) is symmetric. Since symmetric matrices have only real eigenvalues, the eigenvalues of the above Jacobian matrix are real. For simplicity, we calculate the eigenvalues using the symbolic packages of the software Mathematica, which are given by 

\begin{equation*}
   \begin{aligned}
    & \lambda_1= -\gamma+\beta_0(-1+2\phi_1+\phi_2+\phi_3+\phi_4+\phi_5+\phi_6)^2 \\
   & \lambda_2=-\gamma+\beta_0(-1+3\phi_2+\phi_3+\phi_4+\phi_5+\phi_6)^2 \\
    & \lambda_3=-\gamma+\beta_0(-1+4\phi_3+\phi_4+\phi_5+\phi_6)^2 \\
       & \lambda_4=-\gamma+\beta_0(-1+5\phi_4+\phi_5+\phi_6)^2 \\
      & \lambda_5=-\gamma+\beta_0(-1+6\phi_5+\phi_6)^2 \\
         &   \lambda_6=-\gamma+\frac{1}{2} \beta_0 \left(1+\zeta-2 \phi_6 -12 \zeta \phi_6 + 7\phi_6^2 + 42 \zeta \phi_6^2 + B \right)\\
            &   \lambda_7=-\gamma+\frac{1}{2} \beta_0 \left(1+\zeta-2 \phi_6 -12 \zeta \phi_6 + 7\phi_6^2 + 42 \zeta \phi_6^2 -B \right)\\
   \end{aligned}
\end{equation*}
  where $ B:=\sqrt{(1-2\phi_6+7\phi_6^2)^2 +\zeta^2 \left(1-12\phi_6+42\phi_6^2 \right)^2+2\zeta \left(-1+14\phi_6-25\phi_6^2-168\phi_6^3+ 294 \phi_6^4 \right)}$. 
Since all eigenvalues are real, B must be nonnegative for the choice of $\phi_{ij}$, and for any $\zeta$, $B$ is a positive real number. We want to find conditions for which all eigenvalues are negative (LAS DFE), or at least one eigenvalue is positive (unstable DFE). We see that for all $\lambda_i$'s are negative for all $ i=1,2, \hdots, 5$ since all $\phi_{ij}>0$.  It is thus sufficient to check the critical points of $\lambda_6$ and $\lambda_7$, which are more complicated. However, a sufficient condition for stability is the following: since $\lambda_6>\lambda_7$, then we do have stability if the minimum value of $\lambda_6$ is negative. It can be easily verified that $\frac{1}{7}$ is the critical point of both eigenvalues $\lambda_6$, and $\lambda_7$, since  $\frac{\partial \lambda_6}{\partial \phi_6} (\phi_6=\frac{1}{7},\beta_0, \zeta, \gamma)=0$, and $\frac{\partial \lambda_7}{\partial \phi_6} (\phi_6=\frac{1}{7},\beta_0, \zeta, \gamma)=0$.    
We can find that 
\begin{equation*}
   \begin{aligned}
&  \frac{\partial^2 \lambda_6}{\partial \phi_6^2} (\phi_6=\frac{1}{7},\beta_0, \zeta, 
& \frac{\partial^2 \lambda_7}{\partial \phi_6^2} (\phi_6=\frac{1}{7},\beta_0, \zeta, \gamma) 
   \end{aligned}
\end{equation*}
and so $\lambda_6$, and $\lambda_7$ have reached to the minimum value when $\phi_6=\frac{1}{7}$. Moreover, 
\begin{equation*}
\begin{aligned}
      & \lambda_i (\phi_6,\beta_0, \zeta, \gamma) <0 ~~\text{for ~the ~choice~of} ~~ \phi_i, \text{where} ~ i=1,2,\hdots, 5\\
      & \lambda_6(\phi_6=\frac{1}{7},\beta_0, \zeta, \gamma)=\frac{1}{14} \left(6\beta_0+\zeta\beta_0+\sqrt{(6+\zeta)^2} \beta_0 -14\gamma \right)=\frac{1}{7}\beta_0(6+\zeta)-\gamma \\
           & \lambda_7(\phi_6=\frac{1}{7},\beta_0, \zeta, \gamma)=\frac{1}{6} \left(6\beta_0+\zeta\beta_0-\sqrt{(6+\zeta)^2} \beta_0 -14\gamma \right)=-\gamma <0 
   \end{aligned}
\end{equation*}
However, if $\frac{1}{7}\beta_0(6+\zeta)-\gamma <0$ and this implies  $ \beta_0 < \gamma\left( \frac{7}{6+\zeta} \right):=\beta_0^{\text{crit}}$. Alternatively, we have $\zeta < \frac{7 \gamma- 6 \beta_0}{\beta_0}=  \frac{7-6 \frac{\beta_0}{\gamma}}{\frac{\beta_0}{\gamma}}= \frac{7-6R_0^{\text{hom}}}{R_0^{\text{hom}}}=1+\left( \frac{1-R_0^{\text{hom}}}{R_0^{\text{hom}}} \right)7 :=\zeta^{\text{crit}}.$
In addition, the function $\phi_6 \mapsto \lambda_6 (\phi_6, \zeta, \beta_0, \gamma)$ has roots 
\begin{equation*}
   \begin{aligned}
       & \phi_{c1}= \frac{1}{7} +\frac{1}{7} \sqrt{ \frac{\gamma (7\gamma-6\beta_0-\zeta\beta_0 )}{\beta_0(\gamma+6\zeta\gamma-7\zeta\beta_0)}}=\frac{1}{7}+ \frac{1}{7} \sqrt{ \frac{ 7(1-R_0^{\text{hom}})-R_0^{\text{hom}}(\zeta-1) )}{R_0^{\text{hom}} \left(7 \zeta(1-R_0^{\text{hom}})-(\zeta-1) \right)}},\\
       & \phi_{c2}= \frac{1}{7} -\frac{1}{7} \sqrt{ \frac{\gamma (7\gamma-6\beta_0-\zeta\beta_0 )}{\beta_0(\gamma+6\zeta\gamma-7\zeta\beta_0)}}=\frac{1}{7}- \frac{1}{7} \sqrt{ \frac{ 7(1-R_0^{\text{hom}})-R_0^{\text{hom}}(\zeta-1) )}{R_0^{\text{hom}} \left(7 \zeta(1-R_0^{\text{hom}})-(\zeta-1) \right)}}.
   \end{aligned}
\end{equation*}
Therefore, according to the restriction $\phi_6 \in \left(0,\frac{1}{6}\right]$, the interval where there will be no epidemic is given by $\mathcal{I}(\zeta,R_0^{\text{hom}})=\left( \max \left\{0, \frac{1}{7}-\Tilde{\phi} \right\}, \min \left\{\frac{1}{6}, \frac{1}{7}+\Tilde{\phi} \right\} \right)$, where 

\begin{equation*}
   \begin{aligned}
       \Tilde{\phi}=\frac{1}{7} \sqrt{ \frac{ 7(1-R_0^{\text{hom}})-R_0^{\text{hom}}(\zeta-1) )}{R_0^{\text{hom}} \left(7 \zeta(1-R_0^{\text{hom}})-(\zeta-1) \right)}}. 
   \end{aligned}
\end{equation*}
Therefore, if $\phi_6 \in \mathcal{I}$, then the DFE is stable, and otherwise, it is unstable.
\end{proof}


\newpage
\section{More numerical simulations} 
This appendix shows the threshold convergence through colormaps for networks: general star-class, unidirectional, and bidirectional cycles as the network size $M$ increases.

\begin{figure}[!htbp] 
\centering
  \begin{adjustwidth}{0.25in}{0in}  
   \includegraphics[width=.95\textwidth]{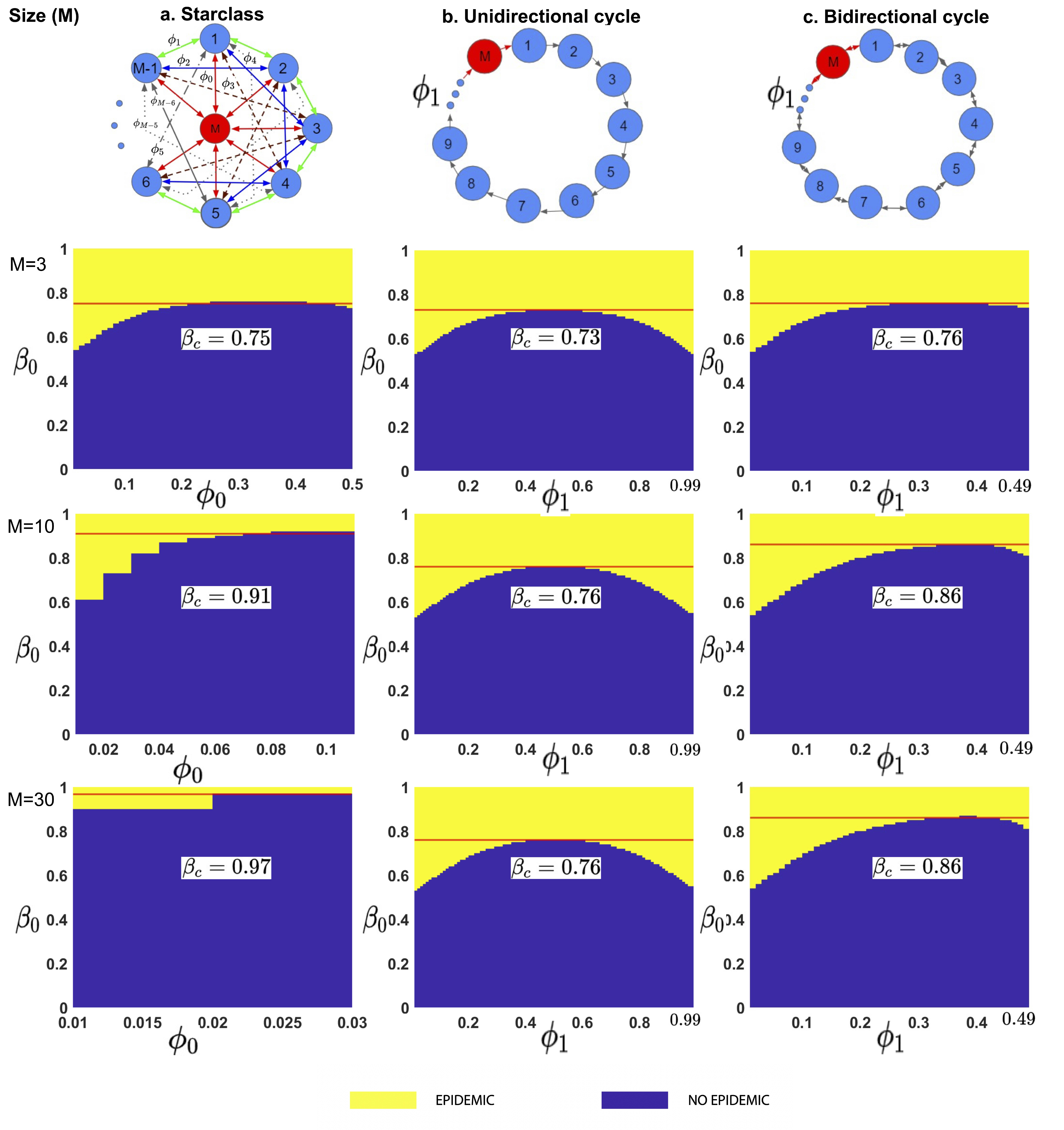}
 \end{adjustwidth}
   \caption{\textbf{Epidemic thresholds as network size $M$ increases:} We present colormaps for $M =3, 10$ and $30$, where $\zeta=2$ and the same data as in Fig \ref{fig:Pcolor_panel_1} for (a) the general star-class network, (b) unidirectional cycles, (c) bidirectional cycles. For each simulation, we assume that every node has the same infection rate $\beta_0$, except the center node $M$, where the infection rate is $\beta_M= \zeta \beta_0$.}
  \label{fig:Pcolor_thersholds_compare_with_M_increased} 
\end{figure}

\end{document}